\documentclass[journal,twoside,web]{ieeecolor}
\usepackage{generic}
\usepackage{cite}
\usepackage{amsmath,amssymb,amsfonts}
\usepackage{algorithmic}
\usepackage{graphicx}
\usepackage{algorithm,algorithmic}
\usepackage{hyperref}
\hypersetup{hidelinks=true}
\usepackage{textcomp}
\def\BibTeX{{\rm B\kern-.05em{\sc i\kern-.025em b}\kern-.08em
    T\kern-.1667em\lower.7ex\hbox{E}\kern-.125emX}}
\usepackage{color,booktabs}

\usepackage{multicol}

\usepackage{amsthm}

\newcommand{\R}{\ensuremath{\mathbb{R}}}

\newcommand{\Rlo}{\ensuremath{\mathbb{R}_{\geq 0}}}

\newcommand{\Zo}{\ensuremath{\mathbb{Z}_{\geq 0}}}
\newcommand{\Zp}{\ensuremath{\mathbb{Z}_{> 0}}}
\newcommand{\Z}{\ensuremath{\mathbb{Z}}}

\definecolor{bleucit}{rgb}{0.2,0.4,0.6} 

\definecolor{blue_cv}{rgb}{0.09,0.35,0.78}

\newcommand{\KL}{\ensuremath{\mathcal{KL}}}
\newcommand{\K}{\ensuremath{\mathcal{K}}}
\newcommand{\Kinf}{\ensuremath{\mathcal{K}_{\infty}}}

\newcommand{\dom}{\ensuremath{\text{dom}\,}}

\newcommand{\sign}[1]{\ensuremath{\text{sign}{(#1)}}}

\newcommand{\norm}[1]{\ensuremath{\left\|{#1}\right\|}}

\newtheorem{defn}{Definition}

\newtheorem{ass}{Assumption}
\newtheorem{prop}{Proposition}

\newtheorem{lem}{Lemma}

\newtheorem{thm}{Theorem}
\newtheorem{cor}{Corollary}

\newtheorem{rem}{Remark}

\let\labelindent\relax
\usepackage{enumitem}

\usepackage{comment}

\usepackage{tikz}
\usetikzlibrary{shapes,arrows,calc}
\definecolor{MyGreen}{RGB}{50,140,80}

\usepackage{pgfplots}
\usetikzlibrary{backgrounds,intersections,spy}
\pgfplotsset{compat=newest}
\pgfplotsset{plot coordinates/math parser=false}
\pgfplotsset{every axis plot/.append style={line width=1pt}}

\definecolor{mycolor1}{rgb}{0.00000,0.44700,0.74100}%
\definecolor{mycolor2}{rgb}{0.85000,0.32500,0.09800}%
\definecolor{mycolor3}{rgb}{0.49400,0.18400,0.55600}%
	
\usepackage{ circuitikz }
\usetikzlibrary{arrows}

\begin{document}
	\allowdisplaybreaks
\title{
	Emulation-based Neuromorphic Control \\ for the Stabilization of LTI Systems
} 
\author{Elena Petri, Koen J.A. Scheres, Erik Steur, and W.P.M.H. (Maurice) Heemels
\thanks{Elena Petri, Erik Steur and Maurice Heemels are with the Department of Mechanical Engineering,	Eindhoven University of Technology, The Netherlands.
	{\tt\small \{e.petri, e.steur, m.heemels\}@tue.nl}}%
\thanks{Koen Scheres is with the Department of Electrical Engineering (ESAT), STADIUS, KU Leuven, Leuven 3001, Belgium. {\tt\small koen.scheres@kuleuven.be}
    }
}

\maketitle

\begin{abstract}
Neuromorphic engineering aims at designing computing and control systems inspired by the neurons and the brain. For the control community, neuromorphic control is an emerging topic that focuses on designing event-based spiking controllers in the form of spiking neural networks (SNNs). At present,
systematic methods for designing and analyzing 
such controllers are lacking. Therefore in this paper we present a systematic approach for stabilizing linear time-invariant (LTI) systems using SNN-based controllers, in the form of a network of integrate-and-fire neurons, whose input is the measured output from the plant, and which generate spiking control signals. 
The new approach consists of a two-step emulation-based design procedure.  In the first step, we establish conditions on the neuron parameters to ensure that the spiky signal generated by a pair of  neurons emulates any continuous-time signal input to the neurons with arbitrary accuracy in terms of a special metric for spiky signals.  In the second step, we propose a novel  stability notion, called spiky-Input-to-State Stability (sISS) 
building on this metric, and prove that an asymptotically stable LTI system has this sISS property. By combining these steps, a certifiable practical stability property of the closed-loop system can be established. 
The approach is illustrated in a numerical case study.
\end{abstract}


\section{Introduction}\label{Introduction}

Many potential advantages of asynchronous and event-driven brain-inspired technologies over classical digital clock-based technologies, including energy efficiency, adaptability and robustness \cite{mead1990neuromorphic}, have motivated numerous research activities in various engineering fields in recent years, see e.g., \cite{gallego2022event, krauhausen2021organic, shrestha2022survey, yamazaki2022spiking, singh2018regulation} and the references therein. This trend is also witnessed in the systems and control community, where neuromorphic control is an emerging research topic, aiming at designing controllers whose dynamics is inspired by biological neurons in the brain. In fact, the potential of brain-inspired control was first illustrated in experimental studies in \cite{deweerth1990neuron,deweerth1991simple} for the regulation of the speed of an electric motor. 
Moreover, \cite{astrom2002comparison,meng2012optimal} showed that event-based impulsive control can give better performance than periodic control.
Recent works have addressed the challenge of designing neuron-inspired controllers, 
 see, e.g., \cite{ribar2021neuromorphic, sepulchre2022spiking,schmetterling2024neuromorphic,  petri2025rhythmic, medvedeva2025formalizing, feketa2023artificial, rosito2025design, agliati2026spiking, jansen2025nuclear} 
 and the references therein. Despite these interesting contributions, a key fundamental open problem from a control-theoretic perspective is the creation of a formal framework for the modeling, design and analysis of  spiking neuromorphic controllers.   
 A natural way to model neuromorphic control  systems is through  spiking neural networks (SNNs). Modeling spiking controllers as neuronal networks, naturally connects this line of research to the literature of neural network-based control, where artificial neural networks (ANNs) have been widely studied, see, e.g., \cite{antsaklis1990neural,bavarian2002introduction, psaltis2002multilayered}. Unlike ANNs, SNNs operate through event-based asynchronous spikes, and therefore more closely mimic biological neuronal dynamics. 
 However, the event-based-spiking nature also introduces crucial theoretical challenges for their analysis and design. 

 The goal of this work is to address these challenges by contributing towards a mathematically rigorous analysis and design framework for spiking neuromorphic controllers for LTI systems (and beyond). To stress, our objective is  not to illustrate improvements in terms of energy-efficiency, robustness and/or adaptability compared to existing non-spiking control techniques, 
 but we aim to provide a systematic approach to design and analyze neuromorphic controllers thereby ensuring important  closed-loop stability guarantees. Indeed, the motivation for the study of these controllers comes from the promising advantages of brain-inspired technologies, but we believe that to achieve these advantages for event-based spiking controllers, a first fundamental step that needs to be taken consists in providing formal mathematical models that can be analyzed. This forms the scope of the present paper. 
Note that spiking control is related to impulsive control problems, see, e.g., \cite{hespanha2008lyapunov} and the reference therein. However, a key difference is that in impulsive control the amplitude of the impulses is typically not fixed. 
Moreover, in impulsive control 
 impulses' times are typically not considered as control variables \cite{haddad2014impulsive}. 
 The fixed-amplitude constraint and the use of spiking times as the only freedom in the control action introduce several open questions for the analysis and design of neuromorphic control systems.
 In our previous work \cite{petri2024analysis, petri2026onTwo}, we already proved a practical stability property of a spiking controller, based on the leaky integrate-and-fire neuronal model, interconnected with a  {\em first-order} linear time-invariant (LTI) system. The assumption of first-order plant dynamics was crucial in \cite{petri2024analysis, petri2026onTwo}, as the proof relied on the explicit solution of the closed-loop system between spikes, making generalization to higher-dimensional systems virtually impossible. In this paper, we therefore follow a completely different route, which, as we will show, applies naturally to higher-order and multi-input multi-output (MIMO) systems. 

The approach followed in this work is based on {\em emulation-based} design of neuromorphic controllers that practically stabilize LTI plants, in which the SNN-based dynamics of the controllers is  inspired by the integrate-and-fire neuronal model \cite{lapicque1907recherches,abbott1999lapicque, izhikevich2010hybrid,gerstner2002spiking}. 
To not blur the main conceptual ideas by too many technicalities, we first focus on a single-input single-output (SISO) LTI system (of any order), assumed to be stabilizable via static output feedback. Later we will extend this approach to MIMO systems and other types of controllers. 
The proposed emulation-based design approach consists of two steps. In {\em the first step}, we provide conditions on the neuron's parameters (i.e., the spike's amplitude and the firing threshold) such that the spiking control signal generated by the neuromorphic controller approximates the ``positive part of the input signal'', i.e., the time-function $t\mapsto \max(0,y(t))$ (reminding of the famous RELU functions used in ANNs), in an appropriate sense, where $y$ represents the  continuous-time  input signal to the neuron. This result is closely related to the well-known \emph{universal approximation property}, see, e.g., \cite{iannella2001spiking, maass1997fast}, where static nonlinear mappings can be arbitrarily closely approximated by ANNs on compact domains. However, our result is different in the sense that we demonstrate that the error between the positive part of the input signal and the spiky output signal of the neuron can be made arbitrarily small in an {\em integral sense} by tuning the neuron's parameters. We will demonstrate that any piecewise affine (PWA) continuous mapping can be approximated (``emulated'') by a network of  integrate-and-fire neurons on the signal level in an integral sense. 
%
{\em The second step} of the emulation-based design approach consists in introducing and guaranteeing a new stability notion, 
for which we coin the name \emph{spiky-Input-to-State Stability} (sISS). This property allows to consider spiky input signals as the input space is equipped with a particular integral-based norm. We prove 
that an asymptotically stable LTI system is sISS with respect to spiky inputs. 
%
By combining the sISS property (second step) with the ``spiky signal-based universal approximation property'' (first step), we can ensure that the trajectory of the LTI system in closed loop with a spiking controller emulates the continuous-time trajectory we would have obtained if a stabilizing static output feedback controller would have been used, in the sense that we provide a tunable bound on the difference between the two trajectories for all times. As a consequence,
practical stability property of the LTI system in closed loop with the spiking controller is established. 


Compared to the preliminary conference version of this work \cite{petri2025spiking}, several results are generalized and proofs are now provided. In particular, we present the spiky signal-based universal approximation property for each neuron separately, and not only for a two-neuron controller approximating a linear scalar function, and we now also consider networks of spiky neurons with more than two neurons, emulating continuous PWA functions. 
Compared to \cite{petri2025spiking}, more insights in the sISS property are given, and the stability result is generalized to MIMO LTI systems. 
An additional property ensured in this paper is that no accumulation of spiking times can happen, and thus the Zeno phenomenon cannot occur.  
Finally, a new numerical example is presented. 

\color{black}

\color{black}

\section{Preliminaries}\label{Notation}

The notation $\R$ stands for the set of real numbers, $\Rlo:= [0, +\infty)$ and $\R_{> 0}:= (0, +\infty)$.  
We use $\Z$ to denote the set of integers, $\Zo:= \{0,1,2,...\}$ and $\Zp:= \{1,2,...\}$.
Given a function $f:\R_{\geq0} \to \R$, for any $t \geq 0$ we denote by $f( t^+)$ the right limit of $f$ at $t$, i.e., $f(t^+) = \lim_{s\downarrow t}f(s)$, and with $f(t^-)$ the left limit, i.e., $f(t^-) = \lim_{s\uparrow t}f(s)$ . 
For a vector $x \in \R^n$, $|x|$ denotes its Euclidean norm. For a matrix $A \in \R^{n  \times m}, \norm{A}$ stands for its induced 2-norm. For a Lebesgue measurable signal $v: \R_{\geq0} \to \R^{n_v}$ with $n_v \in \Zp$, and $t_1,t_2\in\R_{\geq 0}\cup\{\infty\}$ with $t_1\leq t_2$,  $\norm{v}_{[t_1, t_2]}:= \textnormal{ess.} \sup_{t \in [t_1, t_2]} |v(t)|$. 
A continuous function $\alpha: [0, +\infty) \rightarrow [0, +\infty)$, 
is of class $\K$, if $\alpha(0) = 0$ and $\alpha$ is strictly increasing. Moreover, $\alpha$ is of class $\Kinf$ if, additionally, 
$\lim_{r \rightarrow \infty} \alpha(r) = +\infty$. A continuous function $\beta: [0,+\infty) \times [0, \infty) \rightarrow [0, +\infty)$ is of class $\KL$ if, for any fixed $s \in \R_{\geq0}$, $\beta(\cdot,s) \in \K$  and, for each fixed $r \in \R_{\geq0}$, $\beta (r,\cdot)$ is non-increasing and satisfies $\lim_{s \rightarrow +\infty} \beta(r,s) = 0$.
Given a set $\mathcal{Y} \subseteq \R^{n_y}$  with $n_y \in \Zp$, $\mathcal{L}_{\mathcal{Y}}$ denotes the set of all functions $y$ from $\R_{\geq0}$ to $\mathcal{Y}$ that are Lebesgue measurable and locally essentially bounded. 
By $\delta$, we denote the Dirac delta function (distribution), which is defined as $\delta(t) = 0$ when $t \neq 0$, $\delta(t) = +\infty$ when $t = 0$ and it is such that $\int_{-\infty}^{+\infty}\delta(t) dt = 1$.
We use $\text{sign}$ for the set-valued map $\text{sign: }\R\rightrightarrows\{-1,1\}$ defined as, for any $z\in\R$, $\text{sign}(z)=\{1\}$ if $z>0$, $\text{sign}(0)=\{-1,1\}$ and $\text{sign}(z)=\{-1\}$ if $z<0$.

\section{System and neuromorphic controller}\label{SystemAndController}

\begin{figure}
	\begin{center}
		\tikzstyle{blockB} = [draw, fill=blue!30, rectangle, 
		minimum height=1.5em, minimum width=2.5em]  
		\tikzstyle{blockG} = [draw, fill=MyGreen!40, rectangle, 
		minimum height=2em, minimum width=3em]
		\tikzstyle{blockR} = [draw, fill=red!40, rectangle, 
		minimum height=1.5em, minimum width=2.5em]
		\tikzstyle{blockO} = [draw,minimum height=1.5em, fill=orange!20, minimum width=2.5em]
		\tikzstyle{input} = [coordinate]
		\tikzstyle{blockW} = [draw,minimum height=1.5em, fill=white!20, minimum width=2.5em]
		\tikzstyle{input1} = [coordinate]
		\tikzstyle{blockCircle} = [draw, circle]
		\tikzstyle{sum} = [draw, circle, minimum size=.3cm]
		\tikzstyle{blockSensor} = [draw, fill=white!20, draw= blue!80, line width= 0.8mm, minimum height=10em, minimum width=13em]
		
		\begin{tikzpicture}[auto, node distance=2cm,>=latex , scale=0.7,transform shape] 
			
			\node [input, name=stateFeedback] {};
			\node [input, right of= stateFeedback, node distance=1cm] (stateFeedbackRight){};
			\node [input, above of= stateFeedbackRight, node distance=0.7cm] (stateFeedbackUp){};
			\node [input, below of= stateFeedbackRight, node distance=0.7cm] (stateFeedbackDown){};
			
			\node [blockB, right of=stateFeedbackUp, node distance=2cm] (Neuron1) { 
				%
				$\begin{array}{c}
					\text{Neuron } 1\\  (\xi_1, \Delta_1, \alpha_1)
				\end{array}$
			};
			\node [blockO, right of=stateFeedbackDown, node distance=2cm] (Neuron2) { 
				$\begin{array}{c}
					\text{Neuron } 2\\  (\xi_2, \Delta_2, \alpha_2)
				\end{array}$
			};
			
			\draw [draw,-] (stateFeedback) -- node [pos=0.5]{$y$} (stateFeedbackRight);
			\draw [draw,-] (stateFeedbackRight) --  (stateFeedbackUp);
			\draw [draw,-] (stateFeedbackRight) --  (stateFeedbackDown);
			\draw [draw,->] (stateFeedbackUp) -- node {} (Neuron1);
			\draw [draw,->] (stateFeedbackDown) -- node {} (Neuron2);
			
			\node [input, right of= Neuron1, node distance=4cm] (Neuron1End){};
			\node [input, right of= Neuron2, node distance=4cm] (Neuron2End){};
			\node [blockW, right of=Neuron2, node distance=3.2cm] (Neuron2Gain) { 
				$-1$
			};
			
			\draw [draw,-] (Neuron1) -- node {} (Neuron1End);
			\draw [draw,->] (Neuron2) -- node {} (Neuron2Gain);
			\draw [draw,-] (Neuron2Gain) -- node {} (Neuron2End);
			
			\node [input, below of= Neuron1End, node distance=0.6cm] (Neuron1Point){};
			\node [input, above of= Neuron2End, node distance=0.6cm] (Neuron2Point){};
			\node [input, below of= Neuron1Point, node distance=0.1cm] (NeuronCircle){};
			\draw [draw,->] (Neuron1End) -- node {} (Neuron1Point);
			\draw [draw,->] (Neuron2End) -- node {} (Neuron2Point);
			
			\draw [fill=white] (NeuronCircle) circle (0.1cm);
			
			\node [input, right of= NeuronCircle, node distance=0.1cm] (NeuronCircleRight){};
			
			\node [blockG, right of=NeuronCircle, node distance=2.5cm] (system) { 
				$\begin{array}{c}
					\text{Plant }
				\end{array}$
			};
			
			\draw [draw,->] (NeuronCircleRight) -- node [pos=0.8]{$u$} (system);
			
			\node [input, right of= system, node distance=2cm] (SystemOutput){};
			\node [input, right of= system, node distance=1.5cm] (SystemFeedback){};
			\node [input, below of= SystemFeedback, node distance=1.5cm] (SystemFeedbackBelow){};
			\node [input, below of= stateFeedback, node distance=1.5cm] (stateFeedbackBelow){};
			
			\draw [draw,->] (system) -- node [pos=0.5]{$y$} (SystemOutput);
			\draw [draw,-] (SystemFeedback) --  (SystemFeedbackBelow);
			\draw [draw,-] (stateFeedback) --  (stateFeedbackBelow);
			\draw [draw,-] (SystemFeedbackBelow) --  (stateFeedbackBelow);

			\node [input, right of= Neuron1, node distance=1.5cm] (Neuron1SpikeBeginDown){};
			\node [input, above of= Neuron1SpikeBeginDown, node distance=0.2cm] (Neuron1SpikeBegin){};
			\node [input, above of= Neuron1SpikeBegin, node distance=0.6cm] (Neuron1Spike1Up){};
			\node [input, right of= Neuron1SpikeBegin, node distance=0.1cm] (Neuron1Spike2Down){};
			\node [input, left of= Neuron1SpikeBegin, node distance=0.2cm] (Neuron1SpikeBeginLeft){};
			\node [input, above of= Neuron1Spike2Down, node distance=0.6cm] (Neuron1Spike2Up){};
			\node [input, right of= Neuron1Spike2Down, node distance=0.4cm] (Neuron1Spike3Down){};
			\node [input, above of= Neuron1Spike3Down, node distance=0.6cm] (Neuron1Spike3Up){};
			\node [input, right of= Neuron1Spike3Down, node distance=0.3cm] (Neuron1Spike4Down){};
			\node [input, above of= Neuron1Spike4Down, node distance=0.6cm] (Neuron1Spike4Up){};
			\node [input, right of= Neuron1Spike3Down, node distance=0.2cm] (Neuron1Spike4Down){};
			\node [input, above of= Neuron1Spike4Down, node distance=0.6cm] (Neuron1Spike4Up){};
			\node [input, right of= Neuron1Spike4Down, node distance=0.2cm] (Neuron1SpikeFinal){};
			
			\draw [draw,-, blue] (Neuron1SpikeBegin) --  (Neuron1SpikeBeginLeft);
			\draw [draw,-, blue] (Neuron1SpikeBegin) --  (Neuron1Spike1Up);
			\draw [draw,-, blue] (Neuron1SpikeBegin) --  (Neuron1Spike2Down);
			\draw [draw,-, blue] (Neuron1Spike2Down) --  (Neuron1Spike2Up);
			\draw [draw,-, blue] (Neuron1Spike2Down) --  (Neuron1Spike3Down);
			\draw [draw,-, blue] (Neuron1Spike3Down) --  (Neuron1Spike3Up);
			\draw [draw,-, blue] (Neuron1Spike3Down) --  (Neuron1Spike4Down);
			\draw [draw,-, blue] (Neuron1Spike4Down) --  (Neuron1Spike4Up);
			\draw [draw,-, blue] (Neuron1Spike4Down) --  (Neuron1SpikeFinal);
			
			\node [input, right of= Neuron2, node distance=1.5cm] (Neuron2SpikeBeginDown){};
			\node [input, above of= Neuron2SpikeBeginDown, node distance=0.2cm] (Neuron2SpikeBegin){};
			\node [input, above of= Neuron2SpikeBegin, node distance=0.45cm] (Neuron2Spike1Up){};
			\node [input, right of= Neuron2SpikeBegin, node distance=0.3cm] (Neuron2Spike2Down){};
			\node [input, left of= Neuron2SpikeBegin, node distance=0.1cm] (Neuron2SpikeBeginLeft){};
			\node [input, above of= Neuron2Spike2Down, node distance=0.45cm] (Neuron2Spike2Up){};
			\node [input, right of= Neuron2Spike2Down, node distance=0.2cm] (Neuron2Spike3Down){};
			\node [input, above of= Neuron2Spike3Down, node distance=0.45cm] (Neuron2Spike3Up){};
			\node [input, right of= Neuron2Spike3Down, node distance=0.2cm] (Neuron2Spike4Down){};
			\node [input, above of= Neuron2Spike4Down, node distance=0.45cm] (Neuron2Spike4Up){};
			\node [input, right of= Neuron2Spike3Down, node distance=0.3cm] (Neuron2Spike4Down){};
			\node [input, above of= Neuron2Spike4Down, node distance=0.45cm] (Neuron2Spike4Up){};
			\node [input, right of= Neuron2Spike4Down, node distance=0.2cm] (Neuron2SpikeFinal){};
			
			\draw [draw,-, orange] (Neuron2SpikeBegin) --  (Neuron2SpikeBeginLeft);
			\draw [draw,-, orange] (Neuron2SpikeBegin) --  (Neuron2Spike1Up);
			\draw [draw,-, orange] (Neuron2SpikeBegin) --  (Neuron2Spike2Down);
			\draw [draw,-, orange] (Neuron2Spike2Down) --  (Neuron2Spike2Up);
			\draw [draw,-, orange] (Neuron2Spike2Down) --  (Neuron2Spike3Down);
			\draw [draw,-, orange] (Neuron2Spike3Down) --  (Neuron2Spike3Up);
			\draw [draw,-, orange] (Neuron2Spike3Down) --  (Neuron2Spike4Down);
			\draw [draw,-, orange] (Neuron2Spike4Down) --  (Neuron2Spike4Up);
			\draw [draw,-, orange] (Neuron2Spike4Down) --  (Neuron2SpikeFinal);

			\node [input, right of= NeuronCircle, node distance=0.35cm] (Neuron1SpikeBeginDown_tot){};
			\node [input, above of= Neuron1SpikeBeginDown_tot, node distance=0.8cm] (Neuron1SpikeBegin_tot){};
			\node [input, above of= Neuron1SpikeBegin_tot, node distance=0.6cm] (Neuron1Spike1Up_tot){};
			\node [input, right of= Neuron1SpikeBegin_tot, node distance=0.1cm] (Neuron1Spike2Down_tot){};
			\node [input, left of= Neuron1SpikeBegin_tot, node distance=0.2cm] (Neuron1SpikeBeginLeft_tot){};
			\node [input, above of= Neuron1Spike2Down_tot, node distance=0.6cm] (Neuron1Spike2Up_tot){};
			\node [input, right of= Neuron1Spike2Down_tot, node distance=0.4cm] (Neuron1Spike3Down_tot){};
			\node [input, above of= Neuron1Spike3Down_tot, node distance=0.6cm] (Neuron1Spike3Up_tot){};
			\node [input, right of= Neuron1Spike3Down_tot, node distance=0.3cm] (Neuron1Spike4Down_tot){};
			\node [input, above of= Neuron1Spike4Down_tot, node distance=0.6cm] (Neuron1Spike4Up_tot){};
			\node [input, right of= Neuron1Spike3Down_tot, node distance=0.2cm] (Neuron1Spike4Down_tot){};
			\node [input, above of= Neuron1Spike4Down_tot, node distance=0.6cm] (Neuron1Spike4Up_tot){};
			\node [input, right of= Neuron1Spike4Down_tot, node distance=0.2cm] (Neuron1SpikeFinal_tot){};
			
			\draw [draw,-] (Neuron1SpikeBegin_tot) --  (Neuron1SpikeBeginLeft_tot);
			\draw [draw,-, blue] (Neuron1SpikeBegin_tot) --  (Neuron1Spike1Up_tot);
			\draw [draw,-] (Neuron1SpikeBegin_tot) --  (Neuron1Spike2Down_tot);
			\draw [draw,-, blue] (Neuron1Spike2Down_tot) --  (Neuron1Spike2Up_tot);
			\draw [draw,-] (Neuron1Spike2Down_tot) --  (Neuron1Spike3Down_tot);
			\draw [draw,-, blue] (Neuron1Spike3Down_tot) --  (Neuron1Spike3Up_tot);
			\draw [draw,-] (Neuron1Spike3Down_tot) --  (Neuron1Spike4Down_tot);
			\draw [draw,-, blue] (Neuron1Spike4Down_tot) --  (Neuron1Spike4Up_tot);
			\draw [draw,-] (Neuron1Spike4Down_tot) --  (Neuron1SpikeFinal_tot);
			
			\node [input, right of= NeuronCircle, node distance=0.3cm] (Neuron2SpikeBeginDown_tot){};
			\node [input, above of= Neuron2SpikeBeginDown_tot, node distance=0.8cm] (Neuron2SpikeBegin_tot){};
			\node [input, below of= Neuron2SpikeBegin_tot, node distance=0.45cm] (Neuron2Spike1Up_tot){};
			\node [input, right of= Neuron2SpikeBegin_tot, node distance=0.3cm] (Neuron2Spike2Down_tot){};
			\node [input, left of= Neuron2SpikeBegin_tot, node distance=0.1cm] (Neuron2SpikeBeginLeft_tot){};
			\node [input, below of= Neuron2Spike2Down_tot, node distance=0.45cm] (Neuron2Spike2Up_tot){};
			\node [input, right of= Neuron2Spike2Down_tot, node distance=0.2cm] (Neuron2Spike3Down_tot){};
			\node [input, below of= Neuron2Spike3Down_tot, node distance=0.45cm] (Neuron2Spike3Up_tot){};
			\node [input, right of= Neuron2Spike3Down_tot, node distance=0.2cm] (Neuron2Spike4Down_tot){};
			\node [input, below of= Neuron2Spike4Down_tot, node distance=0.45cm] (Neuron2Spike4Up_tot){};
			\node [input, right of= Neuron2Spike3Down_tot, node distance=0.3cm] (Neuron2Spike4Down_tot){};
			\node [input, below of= Neuron2Spike4Down_tot, node distance=0.45cm] (Neuron2Spike4Up_tot){};
			\node [input, right of= Neuron2Spike4Down_tot, node distance=0.2cm] (Neuron2SpikeFinal_tot){};
			
			\draw [draw,-, orange] (Neuron2SpikeBegin_tot) --  (Neuron2Spike1Up_tot);
			\draw [draw,-, orange] (Neuron2Spike2Down_tot) --  (Neuron2Spike2Up_tot);
			\draw [draw,-, orange] (Neuron2Spike3Down_tot) --  (Neuron2Spike3Up_tot);
			\draw [draw,-, orange] (Neuron2Spike4Down_tot) --  (Neuron2Spike4Up_tot);
		\end{tikzpicture}
	\end{center}
	\caption{Block diagram representing the system architecture for a SISO setup. 
	}
	\label{Fig:blockDiagram_op2}
\end{figure}
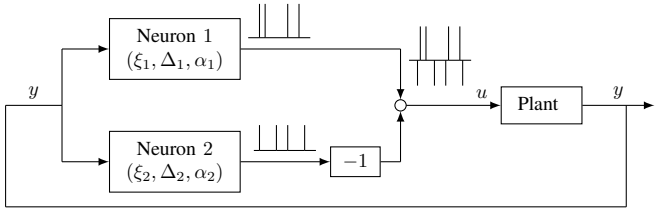

Consider the control of a linear time-invariant (LTI) system, given by 
\begin{equation}
	\begin{aligned}
		&\dot x = Ax + Bu, \qquad
		y = Cx, 
	\end{aligned}
	\label{eq:system}
\end{equation}
where 
$x \in \R^{n_x}$  is the state, with $x(0) = x_0 \in \R^{n_x}$, $y \in \R^{n_y}$ is the output 
and $u \in \R^{n_u}$ is the (spiking) control input and $n_x, n_y, n_u \in \Zp$.  
%
We start from the next assumption.
\begin{ass}
There exists $K \in \R^{n_u \times n_y}$ such that the matrix $A + BKC$ is Hurwitz, with $A \in \R^{n_x\times n_x}$, $B \in \R^{n_x \times n_u}$ and $C \in \R^{n_y\times n_x}$ from \eqref{eq:system}. 
\label{Ass:outputFeedbackControl}
\end{ass}
A sufficient and necessary condition to satisfy Assumption~\ref{Ass:outputFeedbackControl} can be found in, e.g., \cite[Theorem 5]{astolfi2007static}. Our goal is to design a neuromorphic controller generating a spiking control signal $u \in \R^{n_u}$ such that the closed-loop system emulates, in a sense that will be clarified below, the asymptotically stable closed loop, corresponding to $\dot x = (A+BKC)x$, from Assumption~\ref{Ass:outputFeedbackControl}. 
For ease of exposition we first consider single-input single-output (SISO) LTI systems, namely \eqref{eq:system} with $n_u = n_y =1$. In Section \ref{ExtensionsMIMO} we will generalize the results to multiple-input multiple-output (MIMO) LTI systems.


We consider the control configuration shown in Fig.~\ref{Fig:blockDiagram_op2}, where the spiky input $u$ is generated by a neuromorphic controller consisting of two neurons, whose dynamics is inspired by the integrate-and-fire neuron model \cite{lapicque1907recherches,abbott1999lapicque, izhikevich2010hybrid}, \cite[Chapter 4]{gerstner2002spiking}. In particular, $\xi_\ell \in \R_{\geq0}$, $\ell \in \{1,2\}$, represents the membrane potential of the neuron, whose continuous-time dynamics is a first-order differential equation affected by a continuous and nonnegative (current) input, and it is then reset whenever it reaches a positive threshold.
The choice of the integrate-and-fire neuronal model has been made because it is a simple, yet insightful mathematical model of neuronal dynamics and it allows us to study properties of systems in closed loop with spiking controllers. Moreover, this neuronal model is a well-known and well-accepted model for describing neuronal dynamics. 
Although neurons can be sensitive to negative inputs via the so-called rebound spiking mechanism \cite{izhikevich2007dynamical}, in this work we only consider the integrate-and-fire neuron model, which as explained in \cite{gerstner2002spiking}, generates spiking times only when fed with nonnegative inputs. 
%
As such, the input to each of the neurons is taken as a signed version of the continuous-time measured output $y$ of system \eqref{eq:system} with $n_y =1$, which, depending on its sign, influences either $\xi_1$ or $\xi_2$, in our neuromorphic control setup. This leads to one neuron being sensitive to positive values of $y$, through $\max{\{0,y\}}$, and one to negative values of $y$, through $\max{\{0,-y\}}$. 
%
%
%
We thus consider the following dynamics between two spiking actions, 
\begin{subequations}
  \label{eq:thresholdFlow}
\begin{eqnarray}
	\frac{d}{dt} \xi_{1}(t) = \max\{0, y(t)\},  \quad t \in  [t_{1,j_1}, t_{1,j_1 +1}), j_1 \in \Z_{\geq 0}, \label{eq:thresholdFlow1}\\
	 \frac{d}{dt} \xi_{2}(t) = \max\{0, -y(t)\},  \  t \in  [t_{2,j_2}, t_{2,j_2 +1}), j_2 \in \Z_{\geq 0}, \label{eq:thresholdFlow2}
\end{eqnarray}
\end{subequations}
where $t_{1,0} = t_{2,0} = 0$ and $\xi_\ell(0) \in [0, \Delta_\ell)$, where $\Delta_\ell \in \R_{> 0}$ is the firing threshold, which is a neuron parameter, and the sequence $\{t_{\ell,j_\ell}\}_{j_\ell \in \Z_{> 0}}$ represents the spiking instants of $\xi_\ell$, $\ell \in \{1,2\}$. Moreover, for any $t \in \R_{\geq0}$, $j_\ell(t) \in \Zo$ counts the number of spikes generated by neuron $\ell$ up to, and including, time $t \in \R_{\geq0}$.  
Neuron  $\ell \in \{1,2\}$, generates spikes when 
\begin{equation}
\xi_{\ell} \geq \Delta_\ell
\label{eq:triggeringRule}
\end{equation}
is satisfied. The sequence $\{t_{\ell,j_\ell}\}_{j_\ell \in \Z_{\geq 0}}$, $\ell \in \{1,2\}$, is therefore defined as 
\begin{equation}
	t_{\ell, 0} = 0, \quad t_{\ell, j_\ell +1}:= \inf\{t> t_{\ell, j_\ell}:\xi_\ell(t) \geq \Delta_\ell \}.
	\label{eq:triggeringRuleSpikingTimesEachNeuron}
\end{equation}
Note that, according to \eqref{eq:triggeringRuleSpikingTimesEachNeuron}, $t_{\ell, 0} = 0$ is not a spiking time of neuron $\ell \in \{1,2\}$, and the first spiking time is $t_{\ell, 1} \in \R_{>0}$. 
The sequence of spiking times of the two-neuron network is given by, with some abuse of notation, $\{t_j\}_{j\in \Z_{\geq 0}} := \{t_{1, j_1}\}_{j_1 \in \Z_{\geq 0}}\cup \{t_{2, j_2}\}_{j_2 \in \Z_{\geq 0}}$, where, similarly to \eqref{eq:triggeringRuleSpikingTimesEachNeuron}, $t_0 = t_{1,0} = t_{2,0} = 0$ is not a spiking time.  

Inspired by the integrate-and-fire neuron model, where the membrane potential of the neuron is reset whenever it reaches a threshold, we reset the variable $\xi_{\ell}$, $\ell \in \{1,2\}$, when it generates a spike, while it is not modified when the other neuron reaches the threshold, i.e., when neuron $\ell \in \{1,2\}$ satisfies \eqref{eq:triggeringRule} at time $t_{\ell,j_\ell}$, $j_\ell\in \Z_{> 0}$, and thus is the firing neuron,  
\begin{equation}
\begin{aligned}
	&\xi_{\ell}(t_{\ell, j_\ell}^{+}) = 0,   
	\qquad \xi_{3-\ell}(t_{\ell, j_\ell}^{+}) = \xi_{3-\ell}(t_{\ell, j_\ell}).
\end{aligned}
\label{eq:thresholdJump1}
\end{equation}
It is important to notice that the integrate-and-fire neuron model does not include a spike generation mechanism, but, as explained in \cite{izhikevich2010hybrid}, it is only a threshold mechanism and therefore does not intrinsically produce a spike. 
 In our design, we add a spike generation mechanism to the integrate-and-fire neuron model. Indeed, whenever $\xi_{\ell}$, $\ell \in \{1,2\}$, reaches the threshold $\Delta_\ell$, a spiking control action, 
is fired, whose sign depends on which neuron reaches $\Delta_\ell$. 
Note that hardware implementations of neuromorphic controllers in the literature, see, e.g., \cite{feketa2023artificial, agliati2026spiking, schmetterling2024neuromorphic, medvedeva2025formalizing}, produce spiky signals. 
In particular, when $\xi_1$ reaches the threshold, the resulting spiking control action at the spiking time $t_{1,j_1}, j_1 \in \Zp$, is given by $+\alpha_1 \delta(t-t_{1,j_1})$, with $\alpha_1 \in \R_{> 0}$ the amplitude of the spiking control action applied to the plant and $t \in \R_{\geq 0}$.
Similarly, when $\xi_2$ reaches the threshold, the resulting spiking control action at the spiking time $t_{2,j_2}, j_2 \in \Zp$, is given by $-\alpha_2 \delta(t-t_{2,j_2})$, with  $\alpha_2 \in \R_{> 0}$ and $t \in \R_{\geq 0}$. In the following, 
with some abuse of terminology, we consider $\alpha_\ell \in \R_{> 0}$, $\ell \in \{1,2\}$, a neuron parameter.
\color{black}
%

We now define the spiky input generated by each neuron as, for each $\ell \in \{1,2\}$, and all $t\in \R_{\geq0}$,
\begin{equation}
	u_\ell(t) =  \sum\nolimits_{i = 1}^{+\infty} (3-2\ell)\alpha_\ell\delta (t-t_{\ell,i}).  
	\label{eq:SpikingInput_EachNeuron}
\end{equation}
Note that $(3-2\ell) = 1$, when $\ell = 1$ and $(3-2\ell) = -1$ when $\ell = 2$. Therefore, $(3-2\ell)$ simply considers the sign of the spiking output of the neuron.  
Using \eqref{eq:SpikingInput_EachNeuron}, we define the spiky signal $u$ generated by the two integrate-and-fire neurons in the parallel configuration in Fig. \ref{Fig:blockDiagram_op2} as, 
for all $t \in \R_{\geq 0}$,
\begin{equation}
\begin{aligned}
	u(t) &= u_1(t) + u_2(t)\\
    &= \sum\nolimits_{i = 1}^{+\infty} \alpha_1\delta (t-t_{1,i}) + \sum\nolimits_{i = 1}^{+\infty} -\alpha_2\delta (t-t_{2,i}). 
    \end{aligned}
	\label{eq:SpikingInput}
\end{equation}
%
%

Defining solutions to general (nonlinear) systems with spiky inputs is not straightforward. However, for LTI systems, using the convolution integral and the well-known sifting property of the Dirac impulse, solutions of system \eqref{eq:system} with input $u$ in \eqref{eq:SpikingInput}, are well defined, see, e.g., \cite[Section 3.3]{hespanha2009linear}.  In particular, in view of the Dirac impulse definition, the spiky signal in \eqref{eq:SpikingInput} implies that the control input $u$ is different from zero only at the spiking times and thus system \eqref{eq:system} evolves in open loop between two spiking control instants, i.e., $\dot x(t) = Ax(t)$ for all $t \in [t_j, t_{j+1})$, $j \in \Zo$. Moreover, using the convolution integral and sifting property for Dirac impulses again, the solution has a fixed-amplitude discontinuity at the spiking times $t_j\in \R_{\geq0}$, i.e., $x(t_j^+) = x(t_j)+ B\alpha_1$ or $x(t_j^+) = x(t_j) - B\alpha_2$, when neuron \eqref{eq:thresholdFlow1} or \eqref{eq:thresholdFlow2}, respectively, reaches its threshold and thus generates the Dirac impulse at time $t_j$.

Now that we have proposed the neuromorphic controller and provided a suitable modeling framework, it is our objective to derive conditions on the neuron parameters $\alpha_\ell$ and $\Delta_\ell$, $\ell \in \{1,2\}$, to guarantee a formal practical stability property of system \eqref{eq:system} with spiky input \eqref{eq:SpikingInput}, which is the closed-loop system considered in this paper. 
For this purpose we follow an emulation-based approach, as described in the next section. 

\begin{rem}\label{Rem:ImplementationSpikes}
Equations \eqref{eq:SpikingInput_EachNeuron} and \eqref{eq:SpikingInput} assume that the actuator (embedded in the linear time-invariant model) can cope with spiking control signals, which is reasonable in various applications, see, e.g., \cite{van2018organic,krauhausen2021organic}. Interestingly, the number and type of neuromorphic devices, such as actuators and sensors, which are important for the real-life implementation of neuromorphic controllers 
	are growing in the recent years due the engineering community's interest in developing neuromorphic closed-loop systems. 
    For instance, neuromorphic circuit implementations proposed in, e.g., \cite{chicca2003vlsi, indiveri2006vlsi, arthur2010silicon, schultz1995analogue} are inspired by the integrate-and-fire neuronal model, and thus could form promising platforms to implement the proposed controllers.
     In addition, an application example where the control action can be modeled as fixed-amplitude spikes is pellet fueling in a fusion reactor, see, e.g., \cite{jansen2025nuclear, jansen2026fuelling}, where neuromorphic control ideas are applied.
    Moreover, modeling spikes as Dirac impulses is a modeling idealization, which is reasonable in view of the different time scale of spikes compared to the system dynamics, as also done in \cite{petri2024analysis,petri2026onTwo ,petri2025spiking, petri2025rhythmic, agliati2026spiking}.
\end{rem}

\section{Emulation-based neuromorphic controller}\label{EmulationSection}

To ensure a stability property of the closed-loop system with a spiking controller we follow an emulation-based approach in the sense that we first assume that we know a stabilizing continuous-time controller for a LTI system (see Assumption~\ref{Ass:outputFeedbackControl}) and then we design the spiking neural network (SNN) to approximate the closed-loop behavior we would have obtained if the system had been controlled using the stabilizing continuous-time controller.
For SISO LTI systems Assumption~\ref{Ass:outputFeedbackControl} implies that there exists $K \in \R$ such that $A + BKC$ is Hurwitz, with  $A \in \R^{n_x\times n_x}$, $B \in \R^{n_x \times 1}$ and $C \in \R^{1\times n_x}$ from \eqref{eq:system}. Moreover, without loss of generality, we assume $K \in \R_{> 0}$. When Assumption~\ref{Ass:outputFeedbackControl} is satisfied for some $K \in \R_{< 0}$, then such a positive $K$ exists because we can then flip the sign of the generated control input.
The goal of our emulation-based design approach is to approximate the closed-loop trajectories of the ``ideal" closed-loop system 
\begin{equation}
\begin{aligned}
    \dot{\bar{x}} &= (A+BKC) \bar x
    = \bar A \bar x, 
    \end{aligned}
    \label{eq:ClosedLoopToEmulate}
\end{equation}
with $\bar A := (A + BKC) \in \R^{n_x \times n_x}$ and $\bar x(0) = x_0 \in \R^{n_x}$.
For this purpose, we consider the difference $ \tilde x:=x-\bar x \in \R^{n_x}$ between the closed-loop trajectory $x$ given by \eqref{eq:system} with the spiky controller generating \eqref{eq:SpikingInput} 
and the trajectory $\bar x$ of the ideal continuous-time closed loop in \eqref{eq:ClosedLoopToEmulate}. 
This difference $\tilde x$, which we also call the state emulation error, is governed by 
\begin{equation}
\begin{aligned}
\dot {\tilde {x}}  &= \dot x - \dot{\bar{x}}  \stackrel{\eqref{eq:system}, \eqref{eq:ClosedLoopToEmulate}}{=} Ax + Bu- (A+BKC)\bar x\\  
&= Ax + BKCx - BKCx + Bu- (A+BKC)\bar x\\
&= (A + BKC)\tilde{x} - BKCx  + Bu  
= \bar A \tilde{x} - Be, \label{eq:systemWithEmulationError}
\end{aligned}
\end{equation}
with $\tilde x(0)= x(0) - \bar x(0) = x_0 - x_0 = 0$,  and $e := KCx -u \in \R$ the emulation error. Since $\tilde{x}(0) =0$ and $\bar A$ is Hurwitz from Assumption \ref{Ass:outputFeedbackControl}, from \eqref{eq:systemWithEmulationError} we have that the trajectory of system \eqref{eq:system} with spiky input \eqref{eq:SpikingInput} approximates the desired closed-loop trajectory $\bar x$, i.e., $\tilde{x}$ is bounded, if the emulation error $e$ is bounded and small. 
%
Note that $e(t) = Ky(t) \in \R$ for all $t \in [t_j, t_{j+1})$, $j \in \Z_{\geq 0}$, and at the spiking times $t_j$, $j \in \Z_{> 0}$, $e$ has Dirac impulses due to the spikes in $u$ from \eqref{eq:SpikingInput}. Moreover $Ky = KCx$ could be discontinuous at the spiking times $t_j$, $j \in \Z_{> 0}$. 
Hence, $e$ will not be small in a standard Euclidean sense and an alternative (integral) metric will be proposed and used for this purpose in Theorem \ref{Thm:BoundedEmulationErrorIntegral} and Definition \ref{Def:SpikingSignalDefinition} below. 
Thus, our goal is to provide conditions on the neurons parameters $\alpha_\ell$, $\Delta_\ell$ such that the spiky signal $u$ from \eqref{eq:SpikingInput} emulates, in a sense that will be clarified below, the signal $Ky$ with $K$ as in  Assumption~\ref{Ass:outputFeedbackControl}.
We also stress that the proposed neuromorphic controller does not need to implement $Ky$ and only the knowledge of the stabilizing gain $K \in \R_{>0}$ is required for its design and implementation. 
%
%


To guarantee a practical stability property of system~\eqref{eq:system} with $n_u = n_y =1$ in closed loop with the two-neuron spiking controller generating \eqref{eq:SpikingInput},
 we use a two-step approach:
\begin{description}[leftmargin=0cm]
\item[Step 1:]  
We provide conditions on the neuron parameters $\alpha_\ell$ and $\Delta_\ell$, $\ell \in \{1,2\}$, to guarantee an upperbound on the norm of the integral of the emulation error $e$. This bound uses a new metric to formally study a \textit{signal-based} version of a \emph{universal approximation property} of two-neuron integrate-and-fire spiking neural networks (see Section~\ref{IntegralEmulationErrorBound}), which is useful to prove a closed-loop stability property for system \eqref{eq:system} with spiky input \eqref{eq:SpikingInput} in Step 2. 
\item[Step 2:] 
We show that an asymptotically stable LTI system with spiky input, such as \eqref{eq:systemWithEmulationError} with $\bar{A} \in \R^{n_x \times n_x}$ Hurwitz, 
satisfies a new type of input-to-state stability, for which we introduce the term \emph{spiky-Input-to-State Stability} (sISS) (see Section~\ref{StrongIntegralISS}), relating to the metric used in the first step. 
\end{description}

By combining Steps 1 and 2, we can guarantee an arbitrary small bound on the state emulation error $\tilde{x}$ in \eqref{eq:systemWithEmulationError}. This result can be used to obtain a 
practical stability property of system \eqref{eq:system} with spiky input \eqref{eq:SpikingInput} (see Section~\ref{PracticalStabilitySection}). 
%
\color{black}

\section{Approximation properties of integrate-and-fire spiking neural networks
} \label{IntegralEmulationErrorBound}

The goal of this section is to prove a signal-based version of a universal approximation property 
for the two-neuron spiking network in \eqref{eq:thresholdFlow}-\eqref{eq:SpikingInput}. 
For this purpose, we first establish well-posed behavior of the neuron in the sense that in the next proposition, whose proof is given in Appendix \ref{Appendix_ProofProp1}, we ensure that 
%
 the spiky signal $u_\ell$ in \eqref{eq:SpikingInput} generated by the neuron $\ell \in \{1,2\}$ is such that no accumulation of spiking times can happen. Since the number of neurons is finite, this property implies that there is no accumulation of the spiking times in the overall spiking sequence $\{t_j\}_{j \in \Zo} =  \{t_{1,j_1}\}_{j_1 \in \Zo} \cup \{t_{2,j_2}\}_{j_2 \in \Zo}$.  
 This property excludes the Zeno phenomenon and ensures that solutions to system \eqref{eq:system} with spiky input $u$ in \eqref{eq:SpikingInput} are well-defined. This is an important property which will be essential in Sections~\ref{StrongIntegralISS} and~\ref{PracticalStabilitySection} to guarantee stability of the closed-loop system \eqref{eq:system}, \eqref{eq:SpikingInput}. Moreover, this property can be proven following similar lines for any network with finite number of integrate-and-fire neurons.

 \begin{prop}
 	Consider neuron $\ell \in \{1,2\}$ in \eqref{eq:thresholdFlow}-\eqref{eq:thresholdJump1}, where $y \in \mathcal{L}_{\R}$ is 
 	its signal input. 
    Then, the sequence of spiking times 
is $\{t_{\ell,j_\ell}\}_{j_\ell\in S_\ell}$, where $S_\ell$ is either $\Z_{>0}$ or $\{1,2,\dots, n\}$ for some $n \in \Z_{>0}$
  %
  for all $\ell \in \{1,2\}$ and all $\xi_\ell(0) \in [0, \Delta_\ell)$. 
   Moreover, 
   $\{t_j\}_{j \in S} =  \{t_{1,j_1}\}_{j_1 \in S_1} \cup \{t_{2,j_2}\}_{j_2 \in S_2}$ where $S$ is either $\Z_{>0}$ or $\{1,2,\dots, n\}$ for some $n \in \Z_{>0}$
   for all $\xi_\ell(0) \in [0, \Delta_\ell)$, $\ell \in \{1,2\}$.
 \label{Prop:InterSpikingTime}
 \end{prop}

\color{black}

In the next theorem we ensure a signal-based version of the universal approximation property. 

\begin{thm}
	Consider the two-neuron integrate-and-fire spiking neural network in \eqref{eq:thresholdFlow}-\eqref{eq:SpikingInput}, where $y \in \mathcal{L}_{\R}$ is 
	a signal input to the neurons. 
    Then, for any $t \in \R_{\geq 0}$, and each $\ell \in \{1,2\}$,
	\begin{equation}
		\left|\int_{0}^{t} \max{\{0,(3-2\ell)K_\ell y(s)\}} - (3-2\ell)u_\ell(s) ds\right|\leq   \alpha_\ell,
		\label{eq:boundEmulationErrorTheorem_EachNeuron}
	\end{equation}
	with $u_\ell$ in \eqref{eq:SpikingInput_EachNeuron}, and $ K_\ell = \frac{\alpha_\ell}{\Delta_\ell}$. 
	Moreover, 
    for any $t \in \R_{\geq 0}$, 
	\begin{equation}
		\left|\int_{0}^{t} \psi(s) ds\right| \leq \alpha_1 + \alpha_2, 
		 \label{eq:boundEmulationErrorTheorem}
	\end{equation}
	with $\psi(t)= \max{\{0,K_1y(t)\}} - \max{\{0,-K_2y(t)\}} - u(t) \in \R$, $t \in \R_{\geq0}$, where $u$ is defined in \eqref{eq:SpikingInput}. 
    Moreover, in case $\xi_\ell(0) = 0$, $\ell \in \{1,2\}$, then,  for any $t \in \R_{\geq 0}$, 
	\begin{equation}
		\left|\int_{0}^{t} \psi(s) ds\right| \leq \max{\{\alpha_1,\alpha_2\}}. 
		 \label{eq:boundEmulationErrorTheoremMax}
	\end{equation}
   
\label{Thm:BoundedEmulationErrorIntegral}
\end{thm}

Theorem \ref{Thm:BoundedEmulationErrorIntegral}, whose proof is given in Appendix \ref{Appendix_ProofTh1}, 
formally establishes important approximation properties for an integrate-and-fire neuron outputting spiky signals in a ``sup-integral'' sense. 
Note that, the bound is given by $\alpha_\ell$, and thus can be made arbitrary small for fixed $K_\ell$ if there is freedom in the selection of the neuron's firing threshold $\Delta_\ell$ and spike amplitude $\alpha_\ell$.
Moreover, Theorem \ref{Thm:BoundedEmulationErrorIntegral} also ensures a bound, given by $\alpha_1 + \alpha_2$, (or $\max\{\alpha_1,\alpha_2\}$)
on the norm of the integral of $\psi(t)=\max{\{0,K_1y(t)\}} - \max{\{0,-K_2y(t)\}}-u(t)$, for any $t \in \R_{\geq0}$, with $u(t)$ defined in \eqref{eq:SpikingInput} and thus it formally guarantees an approximation property for the two-neuron network in \eqref{eq:thresholdFlow}-\eqref{eq:SpikingInput}. 
%
\color{black}
Note that, \eqref{eq:boundEmulationErrorTheorem_EachNeuron} shows that the spiky signal output of neuron $\ell \in \{1,2\}$, approximates the function $\max{\{0,(3-2\ell)K_\ell y(s)\}}$, which is known as Rectifier Linear Unit (ReLU),  see, e.g., \cite{nair2010rectified, ramachandran2017searching},
and is a well-known activation function in artificial neural networks (ANNs). Hence, Theorem \ref{Thm:BoundedEmulationErrorIntegral} provides a formal link between feedforward ANNs and spiking neural networks (SNNs), where each integrate-and-fire neuron emulates a node of an ANN with a RELU activation function. 
Note that Theorem \ref{Thm:BoundedEmulationErrorIntegral} holds for any $y \in \mathcal{L}_\R$. Therefore if the input of the neurons is $g(y) \in \mathcal{L}_\R$, with $g:\R^{n_y} \to \R$ and $y \in \R^{n_y}$, then $\psi$ in \eqref{eq:boundEmulationErrorTheorem} and \eqref{eq:boundEmulationErrorTheoremMax} is given by $\psi(t)= \max{\{0,K_1g(y(t))\}} - \max{\{0,-K_2g(y(t))\}} - u(t) \in \R$, for all $t \in \R_{\geq0}$. 


\begin{rem}\label{rem:sigmaDelta}
    The neuromorphic control setup in \eqref{eq:thresholdFlow}-\eqref{eq:SpikingInput_EachNeuron} is reminiscent of $\Sigma-\Delta$ modulators used in analog-to-digital signal conversion. In particular, it reminds of so-called continuous-time asynchronous $\Sigma-\Delta$ modulators \cite{kikkert1975asynchronous}, see also \cite{pavan2017understanding}, \cite{lazar2004perfect} and \cite{lazar2004time}, where in the latter a connection to integrate-and-fire neurons with a refractory period is made. In fact, bound as in \eqref{eq:boundEmulationErrorTheorem} is hinted upon, but is not formally proven and not used for practical stabilization of dynamical systems, as done in the next section,  to the best of the authors' knowledge.
\end{rem}

The results of Theorem~\ref{Thm:BoundedEmulationErrorIntegral} can be used to emulate the static output feedback controller from Assumption \ref{Ass:outputFeedbackControl}. Indeed, a special case of  Theorem~\ref{Thm:BoundedEmulationErrorIntegral}  is when $K_1 = \frac{\alpha_1}{\Delta_1} =  K_2 = \frac{\alpha_2}{\Delta_2}=  K \in \R_{> 0}$. In this case, for any $y \in \R$ we have $Ky = \max\{0,Ky\} - \max\{0,-Ky\}$ and thus 
we have, for all $t \in \R_{\geq 0}$,
		$\left|\int_{0}^{t} \tilde{\psi}(s) ds\right| \leq \alpha_1 + \alpha_2$, 
	with $\tilde{\psi}(t)= Ky(t) - u(t)$, $t \in \R_{\geq0}$. 
This special case of Theorem \ref{Thm:BoundedEmulationErrorIntegral} extends the result in the preliminary conference version of this work \cite[Theorem~1]{petri2025spiking}, allowing $\xi_\ell (0) \in [0, \Delta_\ell)$, not only $\xi_\ell (0) = 0$, $\ell \in \{1,2\}$.
Moreover, this special case (with $K_1 = K_2 = K$) of Theorem \ref{Thm:BoundedEmulationErrorIntegral} will be used in Section \ref{PracticalStabilitySection} to ensure that $\tilde{x}$ in \eqref{eq:systemWithEmulationError} is tunable and small, i.e., the difference $\tilde x$ between the state $x$ of the spiky control loop and the state $\bar x$ of the ``ideal" continuous-time system \eqref{eq:ClosedLoopToEmulate} remains small. As \eqref{eq:ClosedLoopToEmulate} is asymptotically stable, this will imply
a global practical stability property of system \eqref{eq:system} with input \eqref{eq:SpikingInput}.

The Dirac impulses are a modeling abstraction and modeling choice of (neuronal) spiky signals, which allow to provide a clean and formal analysis. In practice, implementing Dirac impulses is not feasible, but it is possible to implement spikes modeled as, for instance, finite-width impulses, of the form 
\begin{equation}\label{eq:finiteWidthImpulse}
    \delta_\varepsilon(t-t_i) := \begin{cases} 
    \frac{1}{\varepsilon} \ \ &t\in [t_i,t_i+\varepsilon]\\
    0 \ \ & \text{elsewhere},
    \end{cases}
\end{equation}
where $t_i$ denotes the spiking time, $\varepsilon \in \R_{>0}$ is the impulse width, and $\frac{1}{\varepsilon} \in  \R_{>0}$ is the finite impulse amplitude. Note that, $\int_{0}^{t}(\delta_\varepsilon(s-t_i) - \delta(s-t_i))ds \to 0$ as $\varepsilon \to 0$. 
Similarly to \eqref{eq:SpikingInput_EachNeuron}, using \eqref{eq:finiteWidthImpulse}, we can define the spiky input generated by each integrate-and-fire neuron in \eqref{eq:thresholdFlow}-\eqref{eq:SpikingInput} as the sequence of finite-width impulses, instead of Dirac impulses, for each $\ell \in \{1,2\}$, and all $t \in \R_{\geq0}$,
\begin{equation}
	u_{\varepsilon,\ell}(t) =  \sum\nolimits_{i = 1}^{+\infty} (3-2\ell)\alpha_\ell\delta_\varepsilon (t-t_{\ell,i}),   
	\label{eq:SpikingInput_EachNeuron_finiteWidth}
\end{equation}
where $\{t_{\ell,j_\ell}\}_{j_\ell \in \Z_{\geq 0}}$ denotes the sequence of spiking times of neuron $\ell \in \{1,2\}$, as in \eqref{eq:triggeringRuleSpikingTimesEachNeuron}. 
In the next corollary we show that, when the minimum inter-spiking time is larger than $\varepsilon$, then the universal approximation property in Theorem \ref{Thm:BoundedEmulationErrorIntegral} holds similarly for finite-width impulses, modeled as in \eqref{eq:finiteWidthImpulse}. 

\begin{cor}\label{Cor:UniversalApproximationFiniteWidthImpulse}
    Consider the two-neuron integrate-and-fire spiking neural network in \eqref{eq:thresholdFlow}-\eqref{eq:thresholdJump1}, \eqref{eq:SpikingInput_EachNeuron_finiteWidth}, where $y \in \mathcal{L}_{\R}$ is 
	a signal input to the neurons. 
    There exists $\mu \in \R_{>0}$ such that for each $i \in \{1,2, \dots, j_\ell\}$, $t_{i+1}-t_i \geq \mu$. 
    Then, for any $\varepsilon \in (0,\mu)$, for any $t \in \R_{\geq 0}$, and each $\ell \in \{1,2\}$,
	\begin{equation}
		\left|\int_{0}^{t} \max{\{0,(3-2\ell)K_\ell y(s)\}} - (3-2\ell)u_{\varepsilon,\ell}(s) ds\right|\leq   2\alpha_\ell,
		\label{eq:boundEmulationErrorTheorem_EachNeuron_corollary}
	\end{equation}
	with $u_{\varepsilon,\ell}$ defined in \eqref{eq:SpikingInput_EachNeuron_finiteWidth}, and $K_\ell = \frac{\alpha_\ell}{\Delta_\ell}$. 
\end{cor}
The proof is in Appendix \ref{Proof:CorollaryFiniteWidth}. 
%
%
Corollary \ref{Cor:UniversalApproximationFiniteWidthImpulse} is stated for a single neuron, similarly to \eqref{eq:boundEmulationErrorTheorem_EachNeuron} in Theorem \ref{Thm:BoundedEmulationErrorIntegral}. Results similar to \eqref{eq:boundEmulationErrorTheorem} and \eqref{eq:boundEmulationErrorTheoremMax} can be obtained \textit{mutatis mutandi} for the two-neuron network in  \eqref{eq:thresholdFlow}-\eqref{eq:SpikingInput}. 

\label{rem:PracStab_finiteWidthPulse}
\begin{rem} Corollary \ref{Cor:UniversalApproximationFiniteWidthImpulse} is a first relevant step towards finite-width impulses implementations as in \eqref{eq:finiteWidthImpulse}. However, more work is needed to use this result to ensure a practical stability property of a closed-loop system, similar to the one we will prove later in Theorem \ref{Thm:StabilityTheoremEmulation-variant} in Section \ref{PracticalStabilitySection} for spiking controllers with Dirac impulses, as in \eqref{eq:thresholdFlow}-\eqref{eq:SpikingInput}. In particular, in a closed-loop setting, the property on the minimal time between two consecutive spikes, i.e., $t_{i+i}- t_i \geq \mu$, depends on the plant output $y$, and thus on the applied spiking input in \eqref{eq:SpikingInput_EachNeuron_finiteWidth} which, in the finite-width impulse case depends again on $\varepsilon$ in \eqref{eq:finiteWidthImpulse}. This makes the analysis harder and is left for future work. 
\end{rem}

\color{black}

\section{Spiky-input-to-state stability 
} \label{StrongIntegralISS}
In this section we introduce the \emph{spiky-Input-to-State Stability} (sISS) property, which will be instrumental to ensure a stability property 
for systems controlled by spiking controllers, and we compare it with the notions of input-to-state stability \cite{sontag2008input} and integral input-to-state stability property  \cite{sontag1998comments,angeli2000characterization}. 
%
We then show that
an asymptotically stable LTI system, like the one we consider in \eqref{eq:systemWithEmulationError}, satisfies this stability property with respect to a class of spiky inputs, like $e$ in \eqref{eq:systemWithEmulationError}, which is defined below. 
\begin{defn} 
	A signal $v$ defined on $\R_{\geq 0} $ is a \emph{spiky signal},
	denoted by $v \in \mathcal{S}^{n_v}$, if it can be written as 
		$v = v_1 + v_2$,
	where $v_1 \in \mathcal{L}_{\R^{n_v}}$, and 
	$v_{2}(t):= \sum_{i = 1}^{\infty} \Theta_{i} \delta(t-t_{i})$ with $t \in \R_{\geq 0}$ is a sequence of Dirac impulses with
	$\Theta_{i} \in \R^{n_v}$ for all $i \in \Zp$, and 
	$\{t_{i}\}_{{i} \in \Z_{> 0}}$ is a sequence of spiking times such that  
	%
	%
	$t_{i+1} > t_{i}$, with $t_0 = 0$, for all $i \in \Zo$ and $t_{i} \to \infty$ when $i \to \infty$, and $v$ satisfies 
	\begin{equation}
		\norm{v}_\star :=\sup_{t \in \R_{\geq 0}} \left|\int_{0}^{t}v(s) ds\right| < + \infty, 
		\label{eq:newNormDef}
	\end{equation}
    where $\int_{0}^{t}v(s) ds = \int_{0}^{t}v_1(s) ds + \int_{0}^{t}v_2(s) ds$, where $\int_{0}^{t}v_1(s) ds$ denotes the Lebesgue integral from $0$ to $t \in \R_{\geq0}$ of $v_1 \in \mathcal{L}_{\R^{n_v}}$ and $\int_{0}^{t}v_{2}(s)ds = \sum_{i \in \Zp, t_{i} \in [0,t]}\Theta_{i}$. 
	
	\label{Def:SpikingSignalDefinition}
\end{defn}



%
The proof that $\norm{v}_\star$ in \eqref{eq:newNormDef} is a norm on $\mathcal{S}^{n_v}$ is given in Appendix~\ref{Appendix_NormStar}.  
%
Note that Definition \ref{Def:SpikingSignalDefinition} allows to include multiple spikes at a spiking time $t_i$. In this case, $\Theta_i$ is the signed sum of all the spike amplitudes generated at time $t_i$.
Note that, the emulation error signal $e$, defined in Section~\ref{EmulationSection}, belongs to the class of spiky signals $\mathcal{S}^1$ defined in Definition~\ref{Def:SpikingSignalDefinition} (with $n_v=1$). 
Indeed, $e = Ky -u$, with $u$ the spiky input in \eqref{eq:SpikingInput} and $Ky$, with $K\in \R$ from Assumption \ref{Ass:outputFeedbackControl} with $n_u = n_y =1$, is a Lebesgue measurable and locally essentially bounded function, 
as required by  Definition~\ref{Def:SpikingSignalDefinition}. Moreover, from Theorem~\ref{Thm:BoundedEmulationErrorIntegral}, we have $\norm{e}_\star \leq \alpha_1 +\alpha_2 < + \infty$, and from Proposition~\ref{Prop:InterSpikingTime} 
it follows that $t_i \to +\infty$ when $i \to +\infty$, as required by Definition~\ref{Def:SpikingSignalDefinition}.
%
%

Since we have a bound on the emulation error in the norm $\norm{\cdot}_\star$ and a particular signal class $\mathcal{S}$, we need a new definition of (integral) input-to-state stability \cite{sontag2008input,sontag1998comments}. This is proposed next.
Although we focus on LTI systems in this paper, the sISS property is defined 
for general nonlinear systems of the form
\begin{equation}
	\begin{aligned}
		\dot z &= f(z,v),
	\end{aligned}
	\label{eq:generalNonlinear}
\end{equation}
with $z \in \R^{n_z}$ and $v \in \mathcal{S}^{n_v}$, $n_z, n_v \in \Zp$, provided that solutions to \eqref{eq:generalNonlinear} with $z(0) = z_0 \in \R^{n_z}$ and spiky input $v \in \mathcal{S}^{n_v}$ are well-defined, exist on $\R_{\geq0}$ and are unique. 
\begin{defn}
	System \eqref{eq:generalNonlinear} is \emph{spiky-Input-to-State Stable} (sISS), 
	if there exist functions $\gamma \in \K$ and $\beta \in \KL$ such that, for any $z(0)= z_0 \in \R^{n_z}$  and all $v \in \mathcal{S}^{n_v}$, 
	the corresponding solution $z$ is defined for all $t \in \R_{\geq 0}$ 
	and satisfies
	\begin{equation}
		|z(t)| \leq \beta(|z_0|,t) + \gamma\left(\norm{v}_\star\right)
		\label{eq:StrongISSInequalityDefinition}
	\end{equation}
	for all $t \in \R_{\geq 0}$. 
	\label{Def:StrongIntegralISSDefinition}
\end{defn}


%


Definition \ref{Def:StrongIntegralISSDefinition} provides the sISS property for system \eqref{eq:generalNonlinear} with respect to spiky inputs only when the solutions are well-defined. For LTI systems as in \eqref{eq:system} with input \eqref{eq:SpikingInput}, solutions are well-defined as indicated in Section \ref{SystemAndController} before. This is not straightforward for a general nonlinear system in the form \eqref{eq:generalNonlinear}, see, e.g., \cite{tanwani2015stability}. However, the notion of solution is properly defined for some classes of nonlinear systems, e.g., 
		$\dot z = f(z) + Gv$.
Definition \ref{Def:StrongIntegralISSDefinition} implies that when system \eqref{eq:generalNonlinear} is sISS, then the norm of solution obtained starting from the initial condition $z_0 \in \R^{n_z}$, with input signal $v \in \mathcal{S}^{n_v}$ converges to a neighborhood of the origin, whose size is determined by $\gamma(\norm{v}_\star$), as $t \to +\infty$. 
Note that, when \eqref{eq:generalNonlinear} is an LTI system, then the function $\beta \in \KL$ in \eqref{eq:StrongISSInequalityDefinition} can be taken as exponential, i.e., $\beta(r,s) = ce^{-\lambda s}r$, with $r,s \in \R_{\geq0}, \lambda \in \R_{>0}$,  and $\gamma \in \Kinf$ is linear, i.e., $\gamma(s) = \tilde{\gamma} s$, with $s, \tilde{\gamma} \in \R_{\geq0}$, as will be shown in Theorem \ref{Thm:SystemEmulation_iSISS} later.

The notion of sISS in Definition \ref{Def:StrongIntegralISSDefinition}  for system \eqref{eq:generalNonlinear} with input $v$ differs from the original notions of input-to-state stability (ISS) defined in  \cite{sontag2008input} and integral input-to-state stability (iISS) in \cite{sontag1998comments,angeli2000characterization}, \cite[Proposition 4.7]{mironchenko2023input}  because sISS considers input signals in $\mathcal{S}$, rather than 
 $\mathcal{L}_\infty$ (ISS) or $\mathcal{L}_1$ (iISS). Moreover,  ISS considers $|z(t)|  \leq \beta(|z_0|,t)+ \gamma(\norm{v}_\infty)$ with $\norm{v}_\infty= \text{ess.} \sup_{t\in \R_{\geq0}}|v(t)|$, with $\beta \in \KL$ and $\gamma \in \Kinf$ (see \cite[Section~2.9]{sontag2008input}), iISS 
considers $|z(t)|  \leq \beta(|z_0|,t)+ \gamma_1(\int_{0}^{\infty}\gamma_2(|v(s)|))ds$ with  $\beta \in \KL$, $\gamma_1 \in \Kinf$, $\gamma_2 \in \K$ (see \cite[Definition~4.1]{mironchenko2023input}), and sISS considers $|z(t)|  \leq \beta(|z_0|,t)+ \gamma (\norm{v}_\star)$, as defined in Definition~\ref{Def:StrongIntegralISSDefinition}, thereby using essentially $\mathcal{L}_\infty$, $\mathcal{L}_1$, and $\norm{\cdot}_\star$ norms. Note the different position of the norm in the iISS and the sISS definitions (inside versus outside the integral, respectively). Moreover, the sISS property implies the iISS property, as stated in the next proposition. 

\begin{prop}
\label{prop:sISS_iISS}
If system \eqref{eq:generalNonlinear}  is sISS, it is iISS.
\end{prop}
\begin{proof}
Since the system is sISS, using Definition \ref{Def:StrongIntegralISSDefinition} and $\norm{v}_\star=\sup_{t \in \R_{\geq 0}} |\int_{0}^{t}v(s) ds|$, we obtain, for all $v \in \mathcal{L}_1$ (and thus $v \in \mathcal{S}^{n_v}$) and for all $t \in \R_{\geq0}$,
	$|z(t)| 
	\leq \beta(|z_0|,t)+ \gamma\left({\mathrm{sup}}_{\tau \in [0,\, t]}\left|\int_0^\tau v(s) ds\right|\right)  
	\leq \beta(|z_0|,t)+ \gamma\left({\mathrm{sup}}_{\tau \in [0,\, t]}\int_0^\tau \left|v(s)\right| ds\right) 
	= \beta(|z_0|,t)+\gamma\left(\int_0^t \left|v(s)\right| ds\right)$,
which corresponds to the right-hand side of the iISS definition, see, e.g., \cite[Definition 4.1]{mironchenko2023input}, with $\gamma_2(s)=s, s \in \R_{\geq0}$. 
\end{proof}

Note that considering a different normed space in the sISS definition allows us to consider a larger class of input signals, including ``Dirac impulses'' compared to the original notions of ISS and iISS.  Indeed, for instance, the emulation error defined in Section~\ref{EmulationSection} is such that $\norm{e}_\star$ is finite, but $\norm{e}_\infty = \sup_{t\in \R_{\geq0}}|e(t)|$ and $\norm{e}_1= \int_{0}^{\infty}|e(s)|ds$ are typically not.

\color{black}

In the next theorem, whose proof is in Appendix \ref{Appendix_ProofTh2}, we  prove that sISS holds for an asymptotically stable LTI system.  
\begin{thm}
	Consider system \eqref{eq:generalNonlinear} with $f(z,v) = Fz + Gv$, and
	$F \in \R^{n_z \times n_z}$ Hurwitz and $G \in \R^{n_z\times n_v}$. 
	Then it
	is sISS  
	with $\beta \in \KL$ given by $\beta(r, s) = ce^{-\lambda s}r $, $r, s \in \R_{\geq 0}$, for some $c, \lambda \in \R_{> 0}$,  and $\gamma\in \K$ given by $\gamma(s) =  \tilde{\gamma} s \in \K$, $ s \in \R_{\geq 0}$, with $\tilde{\gamma}:= \norm{G} + \int_{0}^{\infty} \norm{F e^{Fs}G} ds \in \R_{\geq0}$. 
	\label{Thm:SystemEmulation_iSISS}
\end{thm}

The next theorem shows that for system \eqref{eq:generalNonlinear} with $f(z,v) = Fz + Gv$ the properties $F$ Hurwitz, ISS, iISS and sISS are equivalent. 
\begin{thm}\label{thm:ISS_iISS_sISS}
Consider system \eqref{eq:generalNonlinear} with $f(z,v) = Fz + Gv$. Then the following statements are equivalent:
\begin{multicols}{2}
\begin{itemize}
	\item[I.] $F$ is Hurwitz;
	\item[II.] The system \eqref{eq:generalNonlinear} is ISS;
	\item[III.] The system \eqref{eq:generalNonlinear}  is iISS;
	\item[IV.] The system \eqref{eq:generalNonlinear} is sISS.
\end{itemize}
\end{multicols}
\end{thm}
\begin{proof}
The equivalence between items I, II, and III is proved in, e.g., \cite{mironchenko2023input}. From Propositions \ref{prop:sISS_iISS} and Theorem \ref{Thm:SystemEmulation_iSISS}, we conclude that item I implies item IV, which implies item III. 
\end{proof}

In the next proposition we show that the cascade of an sISS system and an ISS system  is sISS.
%
%

\begin{prop}
	Consider the cascade interconnection 
	\begin{equation}
			\dot z = f(z) + Gv
			\label{eq:systemiSISS_forCascadeProp}
		\end{equation}
		\begin{equation}
			\dot q = f_q(q,z),
		\label{eq:systemISS_forCascadeProp}
	\end{equation}
	where $z \in \R^{n_z}$, $q \in \R^{n_q}$ and $v \in \mathcal{S}^{n_v}$, with $n_z, n_q \in \Z_{> 0}$. 
	Suppose \eqref{eq:systemiSISS_forCascadeProp} is sISS with respect to $v \in \mathcal{S}^{n_v}$, as defined in Definition \ref{Def:StrongIntegralISSDefinition}, and \eqref{eq:systemISS_forCascadeProp} is ISS with respect to $z \in \R^{n_z}$, then there exist $\tilde{\beta} \in \KL$ and $\tilde{\gamma} \in \K$ such that, for all $t \in \R_{\geq0}$,
	\begin{equation}
		|q(t)| \leq \tilde{\beta}(|(q_0, z_0)|,t) + \tilde{\gamma}\left(\norm{v}_\star\right), 
		\label{eq:StrongISSInequalityCascadeProposition}
	\end{equation}
	with $q_0 = q(0) \in \R^{n_q}$ and $z_0 = z(0) \in \R^{n_z}$. 
		\label{Prop:cascade_iSISS_ISS} 
	\end{prop}

	The proof of Proposition~\ref{Prop:cascade_iSISS_ISS} is given in Appendix \ref{Appendix_ProofPropCascadeiSISS_ISS} and is based on the  sISS and ISS definitions, and on the properties of compositions of $\K$ and $\KL$ functions. 
    Note that, in case \eqref{eq:systemiSISS_forCascadeProp} is such that $f(z) = Fz$ with $F$ Hurwitz from Theorem \ref{Thm:StabilityTheoremEmulation} we have that it is sISS. Then by showing that is \eqref{eq:systemISS_forCascadeProp} is ISS, using Proposition~\ref{Prop:cascade_iSISS_ISS} we can conclude that the nonlinear system \eqref{eq:systemiSISS_forCascadeProp}-\eqref{eq:systemISS_forCascadeProp} is sISS.

\section{Practical stability property of the closed-loop system}\label{PracticalStabilitySection}
By exploiting the results from Theorems~\ref{Thm:BoundedEmulationErrorIntegral} and~\ref{Thm:SystemEmulation_iSISS}, we first show in Theorem \ref{Thm:StabilityTheoremEmulation-variant} below that the trajectory of system \eqref{eq:system} with spiky input \eqref{eq:SpikingInput} remains (tunably) close to the ideal continuous-time closed-loop trajectory \eqref{eq:ClosedLoopToEmulate} from the same initial state. From this new contribution, 
 we can ensure a practical stability property for LTI systems  in closed loop with spiking controllers (see Corollary \ref{Thm:StabilityTheoremEmulation} below).
 
%
\begin{thm}
	Consider system \eqref{eq:system} with $n_u = n_y =1$ and suppose that Assumption \ref{Ass:outputFeedbackControl} holds with $K \in \R_{>0}$.  Let $u$ be given by \eqref{eq:SpikingInput} using \eqref{eq:thresholdFlow}-\eqref{eq:SpikingInput_EachNeuron} with 
	$\alpha_\ell \in \R_{> 0}$ and $\Delta_\ell \in \R_{> 0}$, $\ell \in \{1,2\}$, in \eqref{eq:triggeringRuleSpikingTimesEachNeuron}, \eqref{eq:SpikingInput_EachNeuron} such that $ K = \frac{\alpha_\ell}{\Delta_\ell}$, $\ell \in \{1,2\}$. 
    Define $\tilde{x} = x - \bar x $ as in \eqref{eq:systemWithEmulationError}, with $\bar x$ from \eqref{eq:ClosedLoopToEmulate}.  Then, for any $x(0) = \bar{x}(0) \in \R^{n_x}$, and any $\xi_\ell(0) \in [0, \Delta_\ell)$, $\ell \in \{1,2\}$, solutions are well-defined for all $t \in \R_{\geq0}$ and it holds that  
	\begin{equation}
		|\tilde x(t)|\leq  \gamma (\alpha_1 + \alpha_2), 
        \label{eq:Th3withSum-variant}
	\end{equation}
	for all $t \in \R_{\geq 0}$, where  
	$\gamma:= \norm{B} +\int_{0}^{\infty} \norm{\bar Ae^{\bar As} B} ds \in \R_{\geq 0}$. Moreover, if $\xi_1(0) = \xi_2(0) = 0$, then for any $x(0) = \bar{x}(0) = x_0 \in \R^{n_x}$, for all $t \in \R_{\geq 0}$, 
    \begin{equation}
		|\tilde x(t)|\leq  \gamma \max{\{\alpha_1, \alpha_2\}}. 
        \label{eq:Th3withMax-variant}
	\end{equation}
	\label{Thm:StabilityTheoremEmulation-variant}
\end{thm}

\begin{proof}
	Let all conditions of Theorem \ref{Thm:StabilityTheoremEmulation-variant} hold. Using $\tilde{x} = x-\bar{x}$, from \eqref{eq:systemWithEmulationError} we have 
		$\dot{\tilde{x}} = \bar{A} \tilde x - Be$,
	where $\bar{A} = (A + BKC) \in \R^{n_x \times n_x}$ is Hurwitz with $ K = \frac{\alpha_\ell}{\Delta_\ell}$, $\ell \in \{1,2\}$, from Assumption \ref{Ass:outputFeedbackControl}, and $e = Ky- u$, where $Ky = \max{\{0,Ky\}} - \max{\{0,-Ky\}}$ and $u$ in \eqref{eq:SpikingInput}. 
    Since $x(0) = \bar x(0)$, it holds that $\tilde x(0) = 0$. From Theorem~\ref{Thm:SystemEmulation_iSISS} we obtain, for all $t \in \dom \tilde{x}$, where $\dom \tilde{x}$ denotes the domain of the solution $\tilde{x}$ (which we prove later to be equal to $[0,+\infty)$), 
    	\begin{equation}
		\begin{aligned}
			|\tilde x(t)| 
			&\leq \left(\norm{B} + \int_{0}^{\infty} \norm{\bar A e^{\bar As} B}ds \right)  \norm{e}_\star = \gamma\norm{e}_\star.\\
		\end{aligned}
		\label{eq:StrongIISSLTI_stateEquation4_bisBis}
	\end{equation}
	%
	Moreover, from Theorem \ref{Thm:BoundedEmulationErrorIntegral} we have that for all $t \in \dom \tilde{x}$, 
		$\left| \int_{0}^{t} e(s)ds \right| \leq \alpha_1 + \alpha_2$. 
	Therefore
	\begin{equation}
		||e||_\star = \sup_{t \in \R_{\geq 0}}\left| \int_{0}^{t} e(s)ds \right| \leq  
        \alpha_1 + \alpha_2, 
        \label{eq:ProofTh3_sum}
	\end{equation}
	%
	which, from \eqref{eq:StrongIISSLTI_stateEquation4_bisBis} implies
	$|\tilde x(t)| 
	\leq \gamma (\alpha_1 + \alpha_2)$. 
%
   Similarly, in case $\xi_1(0) = \xi_2(0) = 0$, from Theorem \ref{Thm:BoundedEmulationErrorIntegral}, \eqref{eq:ProofTh3_sum} can be replaced by 
		$||e||_\star  \leq \max{\{\alpha_1, \alpha_2\}},$  
    which, combined with \eqref{eq:StrongIISSLTI_stateEquation4_bisBis}, implies \eqref{eq:Th3withMax-variant} for all $t \in \dom \tilde{x}$.
     We now prove that $\dom \tilde{x} = [0, +\infty)$, and thus solutions are well-defined for all $t \in \R_{\geq0}$. From \eqref{eq:systemWithEmulationError} we have $\tilde{x} = x - \bar x$, with  $x$ from \eqref{eq:system} and $\bar x$ from \eqref{eq:ClosedLoopToEmulate}. Recall that $u$ from \eqref{eq:SpikingInput}, which is the input to \eqref{eq:system}, produces $\alpha_\ell$ amplitude jumps of the system state $x$ at the spiking times $t_j \in \R_{\geq0}$, i.e., $x(t_j^+) = x(t_j) + B(3-2\ell) \alpha_\ell$,  $\ell \in \{1,2\}$. Moreover, $x$ evolves in open loop according to the linear dynamics $\dot x = Ax$, and thus cannot have finite escape time in between spikes. As a consequence, in absence of the Zeno phenomenon\footnote{Zeno phenomenon refers to infinite amount of spikes in a finite-length time window.}, $y = Cx$ in \eqref{eq:system}, which is the input to \eqref{eq:thresholdFlow}, will be bounded in finite time. 
    We now prove by contradiction that Zeno phenomenon cannot occur for solutions to system \eqref{eq:systemWithEmulationError}. 
    Suppose $\dom \tilde{x} \neq [0, +\infty)$, then $\dom x \neq [0, +\infty)$, as $\dom \bar x = [0,+\infty)$. Let $\bar t$ be the largest value of $t$ such that $x(t)$, with $x(0) = x_0$, is defined on $[0, \bar t)$. 
  This, from \eqref{eq:StrongIISSLTI_stateEquation4_bisBis} and \eqref{eq:ProofTh3_sum} implies that $\tilde x(t)$, and thus also $x(t)$, are bounded on $[0, \bar t)$. Using $y = Cx$ from \eqref{eq:system}, $y$ is bounded for $t \in [0, \bar t)$, and thus, in view of Proposition \ref{Prop:InterSpikingTime}, the Zeno phenomenon cannot occur on  $[0, \bar t)$. We have obtained a contradiction, therefore $\dom \tilde{x} = [0, +\infty)$, and from the first part of the proof, \eqref{eq:Th3withSum-variant} and \eqref{eq:Th3withMax-variant} hold  for all $t \in \R_{\geq0}$.  
\end{proof}

Theorem \ref{Thm:StabilityTheoremEmulation-variant} shows that the trajectories $x$ of the spiky closed-loop system and the corresponding trajectories $\bar x$ to the  ``ideal" loop $\dot{\bar x} = \bar A \bar x$ for the same initial condition remain close according to the bounds in \eqref{eq:Th3withSum-variant} or \eqref{eq:Th3withMax-variant} for all times. Recall that, as discussed after Theorem \ref{Thm:BoundedEmulationErrorIntegral}, we can tune the bound on the state emulation error $\tilde{x}$ by making $\alpha_\ell$, $\ell \in \{1,2\}$, small. As a consequence, using the definition of $\tilde{x} = x-\bar x$, in the next corollary we can ensure a practical stability property of system \eqref{eq:system} with spiky input \eqref{eq:SpikingInput}.

\begin{cor}
	Consider system \eqref{eq:system} with $n_u = n_y =1$ and suppose that Assumption \ref{Ass:outputFeedbackControl} holds with $K \in \R_{>0}$. Let $u$ be given by \eqref{eq:SpikingInput} using \eqref{eq:thresholdFlow}-\eqref{eq:SpikingInput_EachNeuron} with
	$\alpha_\ell \in \R_{> 0}$ and $\Delta_\ell \in \R_{> 0}$, $\ell \in \{1,2\}$, in \eqref{eq:triggeringRuleSpikingTimesEachNeuron}, \eqref{eq:SpikingInput_EachNeuron} such that $ K = \frac{\alpha_\ell}{\Delta_\ell}$, $\ell \in \{1,2\}$. Then, for any $x(0) = x_0 \in \R^{n_x}$, and any $\xi_\ell(0) \in [0, \Delta_\ell)$, $\ell \in \{1,2\}$, for all $t \in \R_{\geq 0}$, 
	\begin{equation}
		|x(t)|\leq \beta(|x_0|, t) + \gamma (\alpha_1 + \alpha_2), 
        \label{eq:Th3withSum}
	\end{equation}
	with $\beta(r,s) = c e^{-\lambda s}r $, $r,s \in \R_{\geq 0}$, for some $c, \lambda \in \R_{> 0}$, and 
	$\gamma:= \norm{B} +\int_{0}^{\infty} \norm{\bar Ae^{\bar As} B} ds \in \R_{\geq 0}$. Moreover, if $\xi_1(0) = \xi_2(0) = 0$, then for any $x(0) = x_0 \in \R^{n_x}$, for all $t \in \R_{\geq 0}$, 
    \begin{equation}
		|x(t)|\leq \beta(|x_0|, t) + \gamma \max{\{\alpha_1, \alpha_2\}}. 
        \label{eq:Th3withMax}
	\end{equation}
	\label{Thm:StabilityTheoremEmulation}
\end{cor}
\begin{proof}
    The proof follows directly from Theorem \ref{Thm:StabilityTheoremEmulation-variant} providing a bound on $\tilde{x} = x-\bar x$ and the asymptotic stability of $\dot{\bar x}= \bar A \bar x$. 
    Indeed, for all $t \in \R_{\geq0}$, we get 
    \begin{equation}
    \begin{aligned}
        |x(t)| &= |\bar x(t) + \tilde{x} (t)| 
        \leq |\bar x(t)| + |\tilde{x} (t)|. 
    \end{aligned}
    \label{eq:Cor1Proof}
    \end{equation}
    In view of Assumption \ref{Ass:outputFeedbackControl}, $\bar A = A+BKC$ is Hurwitz and thus, there exist $c,\lambda \in \R_{>0}$ such that $|\bar x(t)| \leq c e^{-\lambda t} |\bar x(0)|= \beta(|\bar x(0)|, t)$, for all $t \in \R_{\geq0}$. Using $\bar x(0) = x(0) = x_0 \in \R^{n_x}$, and the bound on $|\tilde{x}|$ from \eqref{eq:Th3withSum-variant} and \eqref{eq:Th3withMax-variant}, from \eqref{eq:Cor1Proof}, we obtain  \eqref{eq:Th3withSum} and \eqref{eq:Th3withMax}, respectively. 
\end{proof}

\color{black}

Corollary~\ref{Thm:StabilityTheoremEmulation} guarantees that the state of the closed-loop system \eqref{eq:system}, \eqref{eq:SpikingInput}, with \eqref{eq:SpikingInput} generated by a two-neuron SNN designed to emulate a stabilizing static-output feedback controller (see Assumption~\ref{Ass:outputFeedbackControl}), converges to a neighborhood of the origin in the form of a ball of radius $\gamma (\alpha_1+\alpha_2)$ (or $\gamma \max\{\alpha_1, \alpha_2\}$ if $\xi_\ell(0)=0$, $\ell \in \{1,2\}$), as time grows.  Interestingly, $\beta$ in the corollary is related only to the dynamics of $\dot{\bar x}=\bar A \bar x$ and, hence, it inherits its properties, e.g., convergence rates, in a practical sense.

It is worth noticing again that, since the only condition required in Corollary~\ref{Thm:StabilityTheoremEmulation} is $K= \alpha_\ell/\Delta_\ell$, $\ell \in \{1,2\}$, it is possible to select $\alpha_\ell$ and $\Delta_\ell$ such that the ultimate bound $\gamma (\alpha_1+\alpha_2)$ (or $\gamma \max\{\alpha_1, \alpha_2\}$) is arbitrarily small by choosing $\alpha_\ell$, $\ell \in \{1,2\}$, arbitrarily small, at the cost of also picking the firing threshold $\Delta_\ell$ arbitrarily small, which implies that the spikes are typically generated more frequently, 
see also the proof of Proposition~\ref{Prop:InterSpikingTime}. 
In view of the structure of the spiking controller, which leaves the system plant in open loop between spikes, if the matrix $A$ in \eqref{eq:system} is unstable, an asymptotic stability property cannot be obtained and thus a practical stability property is the best we can ensure. 

\color{black}

Theorem~\ref{Thm:StabilityTheoremEmulation-variant} (and Corollary~\ref{Thm:StabilityTheoremEmulation}) and its proof rely conceptually on the signal-based universal approximation property guaranteed by Theorem~\ref{Thm:BoundedEmulationErrorIntegral} and on the sISS property proven in Theorem~\ref{Thm:SystemEmulation_iSISS}. Following similar steps, it is possible to prove a practical stability property for any sISS system in closed loop with a spiking controller satisfying the bound on the emulation error in Theorem \ref{Thm:BoundedEmulationErrorIntegral}. Thus, the two-step approach we have presented in this paper has potential to be generalized for larger classes of systems and/or different type of neurons, as long as the properties in Theorems~\ref{Thm:BoundedEmulationErrorIntegral} and~\ref{Thm:SystemEmulation_iSISS} are ensured. 


\begin{rem}\label{Rem:robustness}
    From an implementation perspective, small perturbations in the neurons' parameters $\alpha_\ell$ and $\Delta_\ell$, $\ell \in \{1,2\}$, do not compromise the practical stability property ensured by Corollary \ref{Thm:StabilityTheoremEmulation}. Indeed, such perturbations imply that the spiky input generated by the neuromorphic controller  emulates a static output feedback control input $\tilde{K}y$, where $\tilde{K}$ remains close to the original gain $K$ from Assumption \ref{Ass:outputFeedbackControl}. Hence, as long as $A+ B\tilde{K}C$ is Hurwitz, a practical stability property is preserved. As eigenvalues of $A+B\tilde{K}C$ depend continuously on $\tilde{K}$, which depends continuously on the neuron parameters $\alpha_\ell$ and $\Delta_\ell$ in view of Theorem \ref{Thm:BoundedEmulationErrorIntegral}, this shows an intrinsic robustness property of the proposed approach for stabilizing LTI systems. 
\end{rem}

\begin{rem}\label{Rem:robustnessDisturbance}
    The results in Theorem~\ref{Thm:StabilityTheoremEmulation-variant} and  Corollary \ref{Thm:StabilityTheoremEmulation} can be readily extended to perturbed 
    systems $\dot x = Ax + Bu + Dd$, with disturbance $d \in \mathcal{L}_{\R^{n_d}}$ and $D \in \R^{n_x\times n_d}$. In this case, \eqref{eq:ClosedLoopToEmulate} becomes $\dot {\bar x}= [A+BKC]\bar x + Dd$, and thus \eqref{eq:systemWithEmulationError} still holds. 
    Following similar steps as in Sections \ref{EmulationSection}-\ref{PracticalStabilitySection}, stability and robustness (with respect to d and e) properties can be ensured, and, using the bounds in \eqref{eq:Th3withSum-variant} and \eqref{eq:Th3withMax-variant}, similarly to Corollary~\ref{Thm:StabilityTheoremEmulation}, we can prove that there exist $\beta \in \KL$ and $\gamma, \gamma_d \in \R_{\geq0}$ such that, for any $d \in \mathcal{L}_{\R^{n_d}}$, $|x(t)|\leq \beta(|x_0|, t) + \gamma (\alpha_1 + \alpha_2) + \gamma_d \norm{d}_\infty$. 
\end{rem}

\section{Generalizations}\label{Generalizations}
In this section we discuss generalizations and extensions of the presented results. In Section \ref{ExtensionsMIMO} we consider multi-input multi-output (MIMO) LTI systems, emulating static output (or state) feedback control, and we discuss how the network has to be modified in this case to ensure a practical stability property. 
In Section~\ref{subsect:UniversalApproximationPiecewise} we generalize the approximation property illustrated in Section~\ref{IntegralEmulationErrorBound} to approximate any continuous piecewise affine function using an integrate-and-fire spiking neuronal network (SNN) with more than two neurons. This property will allow us to emulate a larger class of static controllers, which could be useful to globally stabilize nonlinear systems.  Note that the approximation property of integrate-and-fire SNNs in Sections~\ref{IntegralEmulationErrorBound}  and~\ref{subsect:UniversalApproximationPiecewise} are independent on the  system we aim to stabilize. Thus, if solutions to a nonlinear plant with spiky input signal are well defined, we just need an sISS property, as in Definition~\ref{Def:StrongIntegralISSDefinition}, and then we can apply the proposed technique to approximate a globally stabilizing continuous-time control input with a spiky signal.

\subsection{MIMO LTI systems} 
In Sections \ref{SystemAndController}-\ref{PracticalStabilitySection} we focused on SISO LTI systems. We now generalize the results for MIMO LTI systems, 
\label{ExtensionsMIMO}
e.g., \eqref{eq:system} with $u:=(u_1, u_2, \dots, u_{n_u}) \in \R^{n_u}$, $y:=(y_1, y_2, \dots, y_{n_y}) \in \R^{n_y}$, with $n_u, n_y \in \Zp$, for which a MIMO version of Assumption~\ref{Ass:outputFeedbackControl} is satisfied. 
In this case, to prove a practical stability property, the control input to emulate is given by
$\hat u := Ky$, with 
\begin{equation}
K = \left[ \begin{smallmatrix}
K_{1,1} & K_{1,2}  &\dots  & K_{1, n_y}\\
K_{2,1} & K_{2,2} & \dots &K_{2, n_y}\\
\vdots & \vdots & &\vdots\\
K_{n_u,1} & K_{n_u,2}  &\dots  & K_{n_u, n_y}\\
\end{smallmatrix}\right] \in \R^{n_u\times n_y}.
\label{eq:MatrixMIMO}
\end{equation}
Each component of the control input $\hat u = (\hat u_1, \hat u_2, \dots, \hat u_{n_u}) \in \R^{n_u}$ we want to emulate is given by 
\begin{equation}
\hat u_i = \sum\nolimits_{j = 1}^{n_y}K_{i,j} y_j, \quad  i \in \{1,2, \dots, n_u\}.
\label{eq:inputToEmulate_MIMO}
\end{equation}
We can therefore design a network composed of $2 n_u n_y$ integrate-and-fire neurons, with dynamics, similarly to \eqref{eq:thresholdFlow}, 
\begin{equation}
\dot{\xi}_{\ell, i,j} = \max\{0, (3-2\ell) y_j\}, 
\label{eq:NeuronDynamicsMIMO}
\end{equation}
with $\xi_{\ell, i,j}(0) \in [0, \Delta_{\ell,i,j})$, $\Delta_{\ell, i,j} \in \R_{> 0}$, $\ell \in \{1,2\}$, $i \in \{1,2,\dots, n_u\}$ and $j \in \{1,2,\dots, n_y\}$. Neuron $\xi_{\ell, i,j}$ generates spikes with amplitude $\alpha_{\ell,i,j} \in \R_{> 0}$ whenever $\xi_{\ell, i,j} \geq \Delta_{\ell, i,j}$, similarly to \eqref{eq:thresholdJump1}. We therefore use a pair of neurons for each $K_{i,j}$ in \eqref{eq:MatrixMIMO}, whose flow and jump dynamics is the same as in \eqref{eq:thresholdFlow}, \eqref{eq:thresholdJump1}. Note that, since the neuron pairs are independent, it is possible that two (or more) neurons, from different neuron pairs, fire a spike at the same time. In this case, \eqref{eq:thresholdJump1} holds independently for each neuron pairs, and thus, the two (or more) neurons reaching their threshold at a spiking time are reset, while all the neurons not reaching their threshold are not modified at that spiking time.
For each $i \in \{1,2,\dots, n_u\}$, the spiking output of the $2 n_y$ neurons (with the same $i$) is summed up to obtain a spiky input $u_i$. 
 Note that, since the neuron dynamics in \eqref{eq:NeuronDynamicsMIMO} is the same as in \eqref{eq:thresholdFlow} and the input to the neuron is 
 $y_j \in \mathcal{L}_\R$, $j \in \{1,2,\dots, n_y\}$, a result similar to the one in Proposition~\ref{Prop:InterSpikingTime} can be proven for each neuron of the network. Thus, since the number of neurons is finite (equal to $2n_u n_y$), the signal $u_i$, $i \in \{1,2,\dots, n_u\}$, is such that no accumulation of spiking times can happen, and is a spiky signal, according to Definition~\ref{Def:SpikingSignalDefinition}. 
By selecting the neurons parameters such that $K_{i,j} = \frac{\alpha_{\ell, i,j}}{\Delta_{\ell,i,j}}$ and using Theorem~\ref{Thm:BoundedEmulationErrorIntegral} 
for each pair of neurons (with same $i$ and $j$), we can prove that $u_i$ emulates $\hat u_i$ in \eqref{eq:inputToEmulate_MIMO} in the sense that, for each $i \in \{1,2,\dots, n_u\}$, for all $t \in \R_{\geq 0}$,
\begin{equation}
\left|\int_{0}^{t} \hat u_i(s) -u_i(s) ds\right| \leq \sum\nolimits_{j = 1}^{n_y}(\alpha_{1, i,j} + \alpha_{2, i,j}). 
\label{eq:boundMIMO_i}
\end{equation}
Recalling the definition of the emulation error $e = \hat u - u \in \mathcal{S}^{n_u}$ and using \eqref{eq:boundMIMO} for each $i \in \{1,2,\dots, n_u\}$, we obtain 
\begin{equation}
	\norm{e}_\star \leq \sqrt{\sum\nolimits_{i =1}^{n_u} \left(\sum\nolimits_{j = 1}^{n_y}(\alpha_{1, i,j} + \alpha_{2, i,j})\right)^2}.
	\label{eq:boundMIMO}
\end{equation}
Note that, similarly to Theorem \ref{Thm:BoundedEmulationErrorIntegral}, the bound on the emulation error depends on the spikes' amplitudes $\alpha_{\ell, i,j} \in \R_{> 0}$, and thus can be made arbitrarily small if there is freedom in selecting the thresholds $\Delta_{\ell,i,j}$. 

The bound on the emulation error \eqref{eq:boundMIMO} can be combined with the sISS property in Section \ref{StrongIntegralISS} and, following similar steps as in Section \ref{PracticalStabilitySection}, a small state emulation error $\tilde{x}$ between $x$ and the solution to $\dot{\bar x} = (A+BKC) \bar x$ can be guaranteed (similarly to Theorem \ref{Thm:StabilityTheoremEmulation-variant}) next to a practical stability property of the closed-loop system (as in Corollary \ref{Thm:StabilityTheoremEmulation} using \eqref{eq:boundMIMO}). 

 \begin{rem} \label{rem:efficiency}
    The MIMO implementation consists of a one-layer integrate-and-fire SNN with $2n_un_y$ neurons, where only \emph{unit} gains are used in the signal input to the neurons, i.e., the $j$-th component of the continuous-time measured output $y = (y_1, y_2, \dots, y_{n_y}) \in \R^{n_y}$ is the input to the pair of neurons $\xi_{1,i,j}$ and $\xi_{2,i,j}$, with $i \in \{1,2,\dots, n_u\}$ and $j \in \{1,2,\dots, n_y\}$. The network structure is as the one in Fig.~\ref{Fig:blockDiagram_op2} repeated $n_u n_y$ times. In particular, each neuron has only the positive or negative part of a signal $y_j$ as input and creates itself the gain $\frac{\alpha_{\ell,i,j}}{\Delta_{\ell,i,j}}$. In case we allow \emph{non-unit} gains in the implementation, we can generate weighted combinations of the measured outputs $y_1,y_2,\ldots,y_{n_y}$ to flow into the neuron dynamics, we can use our emulation-based approach to directly emulate the signals $\hat u_i =K_{i\bullet} y = \sum_j K_{i,j}y_j$ in \eqref{eq:inputToEmulate_MIMO} by a pair of neurons with spike amplitudes $\alpha_{1,i}$ and $\alpha_{2,i}$, i.e., the continuous-time input to the neurons will be $\hat u_i$.
    Here $K_{i\bullet}$ denotes the $i$-th row of $K$, $i=1,\ldots, n_u$. 
    Note that in this case we take the spike amplitudes $\alpha_{1,i}$, $\alpha_{2,i}$ and firing thresholds $\Delta_{1,i}$, $\Delta_{2,i}$ such that $\frac{\alpha_{\ell, i}}{\Delta_{\ell,i}}=1$, $\ell\in \{1,2\}$, $i\in\{1,2,\ldots, n_u\}$ (such that the neurons have ``unit gain''). This leads to a one-layer integrate-and-fire SNN with $2 n_u$   neurons  that have as inputs the positive or negative parts of $K_{i\bullet} y$ (so weighted versions of the inputs). This would reduce not only the number of neurons, but also 
 changes the bound \eqref{eq:boundMIMO_i} to 
$\left|\int_{0}^{t} \hat u_i(s) -u_i(s) ds\right| \leq \alpha_{1, i} + \alpha_{2,i},$ 
which is significantly smaller. Using this bound, instead of the one \eqref{eq:boundMIMO_i}, similar stability results can be ensured \emph{mutatis mutandi}. The lower number of neurons would overall also generate less spikes. This underlines that the emulation results in this paper can be used for various neural network structures (with or without unit gains). Depending on the network structure, i.e., the neurons' configuration, and also the controller gain $K$ that is emulated, the number of spikes can be smaller or larger. Optimal designs of the network structure and the choice of the to-be-emulated gain $K$ (and values of firing thresholds and spike amplitudes) trading off convergence rates, ultimate error bounds and number of spikes is an interesting topic, left for future work. 
\end{rem}

\begin{rem}
The generalization presented in this section to emulate static output feedback controllers to stabilize MIMO LTI plants can be applied \emph{mutatis mutandi} to emulate a static state feedback control input given by $\hat u =Kx$,  with $K \in \R^{n_u \times n_x}$, using the neuromorphic controller. Indeed, in this case, we can simply write \eqref{eq:MatrixMIMO}-\eqref{eq:boundMIMO} for $y = x \in \R^{n_x}$. 
Moreover, when the full state is measured, i.e., $y = x$, if the pair $(A,B)$ in \eqref{eq:system} is stabilizable, then there exists $K \in \R^{n_u \times n_x}$ such that $(A+BK)$ is Hurwitz, namely $\hat u =Kx$ is a static state feedback control. 
\end{rem}

\subsection{Approximation property for larger networks and arbitrary continuous piecewise affine functions}\label{subsect:UniversalApproximationPiecewise}



\begin{figure}
	\begin{center}
		\tikzstyle{blockB} = [draw, fill=blue!30, rectangle, 
		minimum height=2em, minimum width=3em]  
		\tikzstyle{blockG} = [draw, fill=MyGreen!40, rectangle, 
		minimum height=2em, minimum width=3em]
		\tikzstyle{blockR} = [draw, fill=red!40, rectangle, 
		minimum height=2em, minimum width=3em]
		\tikzstyle{blockO} = [draw,minimum height=1.5em, fill=orange!20, minimum width=4em]
		\tikzstyle{input} = [coordinate]
		\tikzstyle{blockW} = [draw,minimum height=1.5em, fill=white!20, minimum width=1.5em]
		\tikzstyle{blockCircle} = [draw, circle]
		
		\begin{tikzpicture}[auto, node distance=2cm,>=latex , scale=0.75,transform shape] 
			
			\node [input, name=origin] {};
			\node [input, right of= origin, node distance=5.5cm] (xAxisRight){};
			\node [input, left of= origin, node distance=3cm] (xAxisLeft){};
			\node [input, above of= origin, node distance=2.7cm] (yAxisUp){};
			\node [input, below of= origin, node distance=0.8cm] (yAxisDown){};
			\draw [draw,->] (xAxisLeft) -- node [pos=0.95]{$y$} (xAxisRight);
			\draw [draw,->] (yAxisDown) -- node [pos=0.95]{$g(y)$} (yAxisUp);	
			\node [input, above of= origin, node distance=1cm] (c){};
			\draw [draw,-] (c) -- node [pos=0.5]{$c$} (c);	
			\draw [fill=black] (c) circle (0.05cm);
			\node [input, left of= origin, node distance=2.5cm] (b1){};
			\node [input, left of= origin, node distance=1.3cm] (b2){};
			\node [input, left of= origin, node distance=0.4cm] (b3){};
			\node [input, right of= origin, node distance=0.6cm] (b4){};
			\node [input, right of= origin, node distance=1.8cm] (b5){};
			\node [input, right of= origin, node distance=3.6cm] (b6){};
			\node [input, right of= origin, node distance=4.4cm] (b7){};
			\node [input, above of= b1, node distance=0.1cm] (b1Up){};
			\node [input, below of= b1, node distance=0.1cm] (b1Down){};
			\node [input, above of= b2, node distance=0.1cm] (b2Up){};
			\node [input, below of= b2, node distance=0.1cm] (b2Down){};
			\node [input, above of= b3, node distance=0.1cm] (b3Up){};
			\node [input, below of= b3, node distance=0.1cm] (b3Down){};
			\node [input, above of= b4, node distance=0.1cm] (b4Up){};
			\node [input, below of= b4, node distance=0.1cm] (b4Down){};
			\node [input, above of= b5, node distance=0.1cm] (b5Up){};
			\node [input, below of= b5, node distance=0.1cm] (b5Down){};
			\node [input, above of= b6, node distance=0.1cm] (b6Up){};
			\node [input, below of= b6, node distance=0.1cm] (b6Down){};
			\node [input, above of= b7, node distance=0.1cm] (b7Up){};
			\node [input, below of= b7, node distance=0.1cm] (b7Down){};
			\draw [draw,-] (b1Down) -- node [pos=0.1]{} (b1Up);	
			\draw [draw,-] (b2Down) -- node [pos=0.1]{} (b2Up);	
			\draw [draw,-] (b3Down) -- node [pos=0.1]{} (b3Up);	
			\draw [draw,-] (b4Down) -- node [pos=0.1]{} (b4Up);	
			\draw [draw,-] (b5Down) -- node [pos=0.1]{} (b5Up);	
			\draw [draw,-] (b6Down) -- node [pos=0.1]{} (b6Up);	
			\draw [draw,-] (b7Down) -- node [pos=0.1]{} (b7Up);	
			\node [input, below of= b1, node distance=0.3cm] (b1Name){};						
			\draw [draw,-] (b1Name) -- node [pos=0.1, right = -0.2cm]{$b_1$} (b1Name);	
			\node [input, below of= b2, node distance=0.3cm] (b2Name){};						
			\draw [draw,-] (b2Name) -- node [pos=0.1, right = -0.2cm]{$b_2$} (b2Name);	
			\node [input, below of= b3, node distance=0.3cm] (b3Name){};						
			\draw [draw,-] (b3Name) -- node [pos=0.1, right = -0.2cm]{$b_3$} (b3Name);	
			\node [input, below of= b4, node distance=0.3cm] (b4Name){};						
			\draw [draw,-] (b4Name) -- node [pos=0.1, right = -0.2cm]{$b_4$} (b4Name);	
			\node [input, below of= b5, node distance=0.3cm] (b5Name){};						
			\draw [draw,-] (b5Name) -- node [pos=0.1, right = -0.2cm]{$b_5$} (b5Name);	
			\node [input, below of= b6, node distance=0.3cm] (b6Name){};						
			\draw [draw,-] (b6Name) -- node [pos=0.1, right = -0.2cm]{$b_6$} (b6Name);	
			\node [input, below of= b7, node distance=0.3cm] (b7Name){};						
			\draw [draw,-] (b7Name) -- node [pos=0.1, right = -0.2cm]{$b_7$} (b7Name);

			\node [input, above of= b1, node distance=1cm] (b1Point){};  
			\node [input, above of= b2, node distance=1.8cm] (b2Point){};
			\node [input, above of= b3, node distance=1.8cm] (b3Point){};
			\node [input, above of= b4, node distance=0.5cm] (b4Point){};
			\node [input, above of= b5, node distance=0.7cm] (b5Point){};
			\node [input, above of= b6, node distance=-0.7cm] (b6Point){};
			\node [input, above of= b7, node distance=0.2cm] (b7Point){};
			
			\node [input, left of= b1, node distance=0.6cm] (Start){};
			\node [input, right of= b7, node distance=0.8cm] (End){};
			\node [input, above of= Start, node distance=-0.3cm] (StartPoint){};
			\node [input, above of= End, node distance=0.6cm] (EndPoint){};
			
			\draw [draw,-] (StartPoint) -- node []{$K_0$} (b1Point);	
			\draw [draw,-] (b1Point) -- node []{$K_1$} (b2Point);
			\draw [draw,-] (b2Point) -- node [pos=0.5]{$K_2$} (b3Point);
			\draw [draw,-] (b3Point) -- node []{$K_3$} (b4Point);
			\draw [draw,-] (b4Point) -- node []{$K_4$} (b5Point);
			\draw [draw,-] (b5Point) -- node []{$K_5$} (b6Point);
			\draw [draw,-] (b7Point) -- node [pos=0.65]{$K_6$} (b6Point);
			\draw [draw,-] (b7Point) -- node []{$K_7$} (EndPoint);
			\draw [dotted] (b1Point) -- node []{} (c);	
			\draw [dotted] (b1Point) -- node []{} (b1);	
		\end{tikzpicture}
	\end{center}
	\caption{Example of piecewise linear function $g$ in \eqref{eq:PWLfunction} with $N=7$.}
	\label{Fig:PiecewiseLinearFunction}
\end{figure}
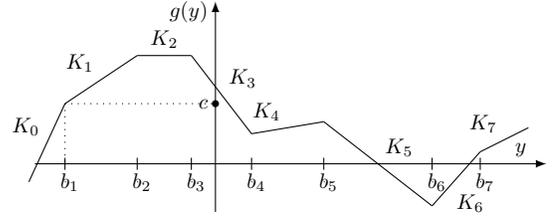

\begin{figure}
	\begin{center}
			\tikzstyle{blockB} = [draw, fill=blue!30, rectangle, 
			minimum height=1.5em, minimum width=10.5em]  
			\tikzstyle{blockG} = [draw, fill=MyGreen!40, rectangle, 
			minimum height=1.5em, minimum width=10.5em]
			\tikzstyle{blockR} = [draw, fill=red!30, rectangle, 
			minimum height=1.5em, minimum width=10.5em]
			\tikzstyle{blockO} = [draw,minimum height=1.5em, fill=orange!20, minimum width=10.5em]
            \tikzstyle{blockY} = [draw,minimum height=1.5em, fill=yellow!40, minimum width=10.5em]
            \tikzstyle{blockP} = [draw,minimum height=1.5em, fill=pink!30, minimum width=10.5em]
			\tikzstyle{input} = [coordinate]
			\tikzstyle{blockW} = [draw,minimum height=1.5em, fill=white!20, minimum width=2em]
			\tikzstyle{blockWNeuron} = [draw,minimum height=1.5em, fill=white!20, minimum width=10.5em]
			\tikzstyle{input1} = [coordinate]
			\tikzstyle{blockCircle} = [draw, circle]
			\tikzstyle{sum} = [draw, circle, minimum size=.3cm]
			\tikzstyle{blockSensor} = [draw, fill=white!20, draw= blue!80, line width= 0.8mm, minimum height=10em, minimum width=13em]
			
			\begin{tikzpicture}[auto, node distance=2cm,>=latex , scale=0.65,transform shape] 
					
					\node [input, name=stateFeedback] {};
					\node [input, right of= stateFeedback, node distance=1cm] (stateFeedbackRight){};
					\node [input, above of= stateFeedbackRight, node distance=0.7cm] (stateFeedbackUp){};
					\node [input, below of= stateFeedbackRight, node distance=0.7cm] (stateFeedbackDown){};
					\node [input, above of= stateFeedbackRight, node distance=2.1cm] (stateFeedbackUpUp){};
					\node [input, below of= stateFeedbackRight, node distance=2.8cm] (stateFeedbackDownN){};
					\node [input, right of= stateFeedbackUpUp, node distance=1cm] (stateFeedbackUpUpRight){};
					\node [blockW, left of=stateFeedbackUpUpRight, node distance=0cm] (Neuron0GainSign) { 
							$-1$
						};

					\node [input, right of=stateFeedbackUp, node distance=0.2cm] (BetaCircle1){};
					\draw [fill=white] (BetaCircle1) circle (0.1cm);
					\node [input, left of=BetaCircle1, node distance=0.1cm] (BetaCircle1Start){};
					\node [input, right of=BetaCircle1, node distance=0.1cm] (BetaCircle1End){};
					\node [input, right of=BetaCircle1, node distance=1cm] (BetaCircle1Right){};
					\node [input, above of=BetaCircle1, right = 0cm, node distance=0.6cm] (BetaCircle1Up){};
					\node [input, above of=BetaCircle1, right = 0cm, node distance=0.1cm] (BetaCircle1Up2){};
					\draw [draw,->] (BetaCircle1Up) -- node [pos=0.01]{$-b_1$} (BetaCircle1Up2);
					
					\node [input, right of=stateFeedbackDown, node distance=0.2cm] (BetaCircle2){};
					\draw [fill=white] (BetaCircle2) circle (0.1cm);
					\node [input, left of=BetaCircle2, node distance=0.1cm] (BetaCircle2Start){};
					\node [input, right of=BetaCircle2, node distance=0.1cm] (BetaCircle2End){};
					\node [input, right of=BetaCircle2, node distance=1cm] (BetaCircle2Right){};
					\node [input, above of=BetaCircle2, right = 0cm, node distance=0.6cm] (BetaCircle2Up){};
					\node [input, above of=BetaCircle2, right = 0cm, node distance=0.1cm] (BetaCircle2Up2){};
					\draw [draw,->] (BetaCircle2Up) -- node [pos=0.01]{$-b_2$} (BetaCircle2Up2);
					
					\node [input, right of=stateFeedbackDownN, node distance=0.2cm] (BetaCircleN){};
					\draw [fill=white] (BetaCircleN) circle (0.1cm);
					\node [input, left of=BetaCircleN, node distance=0.1cm] (BetaCircleNStart){};
					\node [input, right of=BetaCircleN, node distance=0.1cm] (BetaCircleNEnd){};
					\node [input, right of=BetaCircleN, node distance=1cm] (BetaCircleNRight){};
					\node [input, above of=BetaCircleN, right = 0cm, node distance=0.6cm] (BetaCircleNUp){};
					\node [input, above of=BetaCircleN, right = 0cm, node distance=0.1cm] (BetaCircleNUp2){};
					\draw [draw,->] (BetaCircleNUp) -- node [pos=0.01]{$-b_N$} (BetaCircleNUp2);
					
					\node [input, right of=BetaCircle1End, node distance=0.7cm] (BetaCircle1EndRight){};
					
					\node [blockB, right of=stateFeedbackUpUp, node distance=4.7cm] (Neuron0) { 
							$\begin{array}{c}
									\text{Neuron } 0\\  (\tilde{\xi}_0, \tilde{\Delta}_0,  \tilde{\alpha}_0)
								\end{array}$
						};
					
					\node [blockG, right of=stateFeedbackUp, node distance=4.7cm] (Neuron1) { 
							$\begin{array}{c}
									\text{Neuron } 1\\  (\tilde{\xi}_1,  \tilde{\Delta}_1, \tilde{\alpha}_1)
								\end{array}$
						};
					\node [blockY, right of=stateFeedbackDown, node distance=4.7cm] (Neuron2) { 
							$\begin{array}{c}
									\text{Neuron } 2\\  (\tilde{\xi}_2,  \tilde{\Delta}_2, \tilde{\alpha}_2)
								\end{array}$
						};
					
					\node [blockO, right of=stateFeedbackDownN, node distance=4.7cm] (NeuronN) { 
							$\begin{array}{c}
									\text{Neuron } N\\  (\tilde{\xi}_N,  \tilde{\Delta}_N, \tilde{\alpha}_N)
								\end{array}$
						};
						
					\node [blockR, below of=NeuronN, node distance=1.4cm] (NeuronC) { 
							$\begin{array}{c}
								\text{Neuron } N+1\\  (\tilde{\xi}_{N+1},  \tilde{\Delta}_{N+1}, \tilde{\alpha}_{N+1})
							\end{array}$
					};
						

					\node [blockP, below of=NeuronC, node distance=1.4cm] (NeuronC2) { 
						$\begin{array}{c}
							\text{Neuron } N+2\\  (\tilde{\xi}_{N+2},  \tilde{\Delta}_{N+2}, \tilde{\alpha}_{N+2})
						\end{array}$
					};
					\node [input, left of= NeuronC, node distance=4.7cm] (ConstantInput){};
					\node [blockW, left of=NeuronC2, node distance=3.0cm] (cSign) { 
						-1
					};
					\node [input, left of= cSign, node distance=1cm] (ConstantInputLeftBelow){};
					\node [input, above of= ConstantInputLeftBelow, node distance=1.4cm] (ConstantInputBelow){};
					\draw [draw,->] (ConstantInput) -- node [pos=0.5]{$c$} (NeuronC);
					\draw [draw,-] (ConstantInputBelow) -- node [pos=0.5]{} (ConstantInputLeftBelow);
					\draw [draw,-] (ConstantInputLeftBelow) -- node [pos=0.5]{} (cSign);
					\draw [draw,->] (cSign) -- node [pos=0.5]{$-c$} (NeuronC2);
					
					
					\node at ($(Neuron2)!.45!(NeuronN)$) {\large \vdots};
					
					\draw [draw,-] (stateFeedback) -- node [pos=0.5]{$y$} (stateFeedbackRight);
					\draw [draw,-] (stateFeedbackRight) --  (stateFeedbackUp);
					\draw [draw,-] (stateFeedbackRight) --  (stateFeedbackDown);
					\draw [draw,-] (stateFeedbackDown) --  (stateFeedbackDownN);
					\draw [draw,-] (stateFeedbackUp) -- node {} (BetaCircle1Start);
					\draw [draw,->] (BetaCircle1End) -- node [pos=0.6]{$y -b_1$} (Neuron1);
					\draw [draw,-] (stateFeedbackDown) -- node {} (BetaCircle2Start);
					\draw [draw,->] (BetaCircle2End) -- node [pos=0.6]{$y -b_2$} (Neuron2);
					\draw [draw,-] (stateFeedbackDownN) -- node {} (BetaCircleNStart);
					\draw [draw,->] (BetaCircleNEnd) -- node [pos=0.6]{$y -b_N$} (NeuronN);
					\draw [draw,->] (Neuron0GainSign) -- node {$-(y -b_1)$} (Neuron0);
					\draw [draw,->] (BetaCircle1EndRight) -- node {} (Neuron0GainSign);
					
					\node [input, right of= BetaCircle2End, node distance=1.7cm] (BetaCircle2EndDots){};
					\node [input, right of= BetaCircleNEnd, node distance=1.7cm] (BetaCircleNEndDots){};
					\node at ($(BetaCircle2EndDots)!.45!(BetaCircleNEndDots)$) {\large \vdots};
					
					\node [input, right of= Neuron0, node distance=4.5cm] (Neuron0End){};
					\node [input, right of= Neuron1, node distance=4.5cm] (Neuron1End){};
					\node [input, right of= Neuron2, node distance=4.5cm] (Neuron2End){};
					\node [input, right of= NeuronN, node distance=4.5cm] (NeuronNEnd){};
					\node [input, right of= NeuronC, node distance=4.5cm] (NeuronCEnd){};
					\node [input, right of= NeuronC2, node distance=4.5cm] (NeuronC2End){};
					\node [blockW, right of=Neuron0, node distance=3.8cm] (Neuron0Gain) { 
							$\mathcal{G}_0$
						};
					\node [blockW, right of=Neuron1, node distance=3.8cm] (Neuron1Gain) { 
							$\mathcal{G}_1$
						};
					\node [blockW, right of=Neuron2, node distance=3.8cm] (Neuron2Gain) { 
							$\mathcal{G}_2$
						};
					\node [blockW, right of=NeuronN, node distance=3.8cm] (NeuronNGain) { 
							$\mathcal{G}_N$
						};
					\node [blockW, right of=NeuronC, node distance=3.8cm] (NeuronCGain) { 
						$1$
					};
					
					\node [blockW, right of=NeuronC2, node distance=3.8cm] (NeuronC2Gain) { 
						$-1$
					};

					\node at ($(Neuron2Gain)!.45!(NeuronNGain)$) {\large \vdots};
					
					\draw [draw,->] (Neuron0) -- node {} (Neuron0Gain);
					\draw [draw,-] (Neuron0Gain) -- node {} (Neuron0End);
					\draw [draw,->] (Neuron1) -- node {} (Neuron1Gain);
					\draw [draw,-] (Neuron1Gain) -- node {} (Neuron1End);
					\draw [draw,->] (Neuron2) -- node {} (Neuron2Gain);
					\draw [draw,-] (Neuron2Gain) -- node {} (Neuron2End);
					\draw [draw,->] (NeuronN) -- node {} (NeuronNGain);
					\draw [draw,-] (NeuronNGain) -- node {} (NeuronNEnd);
					\draw [draw,->] (NeuronC) -- node {} (NeuronCGain);
					\draw [draw,-] (NeuronCGain) -- node {} (NeuronCEnd);
					\draw [draw,->] (NeuronC2) -- node {} (NeuronC2Gain);
					\draw [draw,-] (NeuronC2Gain) -- node {} (NeuronC2End);
					
					\node [input, below of= Neuron1End, node distance=0.6cm] (Neuron1Point){};
					\node [input, above of= Neuron2End, node distance=0.6cm] (Neuron2Point){};
					\node [input, below of= Neuron1Point, node distance=0.1cm] (NeuronCircle){};
					\draw [draw,-] (Neuron0End) -- node {} (Neuron1End);
					\draw [draw,->] (Neuron1End) -- node {} (Neuron1Point);
					\draw [draw,->] (Neuron2End) -- node {} (Neuron2Point);
					\draw [draw,-] (NeuronC2End) -- node {} (Neuron2End);
					
					\draw [fill=white] (NeuronCircle) circle (0.1cm);
					
					\node [input, right of= NeuronCircle, node distance=0.1cm] (NeuronCircleRight){};
					
					\node [input, right of= NeuronCircleRight, node distance=1.5cm] (NeuronCircleRight2){};
					
					\draw [draw,->] (NeuronCircleRight) -- node {$u$} (NeuronCircleRight2);

					\node [input, right of= Neuron0, node distance=2.1cm] (Neuron0SpikeBeginDown){};
					\node [input, above of= Neuron0SpikeBeginDown, node distance=0.2cm] (Neuron0SpikeBegin){};
					\node [input, above of= Neuron0SpikeBegin, node distance=0.3cm] (Neuron0Spike1Up){};
					\node [input, right of= Neuron0SpikeBegin, node distance=0.15cm] (Neuron0Spike2Down){};
					\node [input, left of= Neuron0SpikeBegin, node distance=0.1cm] (Neuron0SpikeBeginLeft){};
					\node [input, above of= Neuron0Spike2Down, node distance=0.3cm] (Neuron0Spike2Up){};
					\node [input, right of= Neuron0Spike2Down, node distance=0.3cm] (Neuron0Spike3Down){};
					\node [input, above of= Neuron0Spike3Down, node distance=0.3cm] (Neuron0Spike3Up){};
					\node [input, right of= Neuron0Spike3Down, node distance=0.5cm] (Neuron0Spike4Down){};
					\node [input, above of= Neuron0Spike4Down, node distance=0.3cm] (Neuron0Spike4Up){};
					\node [input, right of= Neuron0Spike3Down, node distance=0.2cm] (Neuron0Spike4Down){};
					\node [input, above of= Neuron0Spike4Down, node distance=0.3cm] (Neuron0Spike4Up){};
					\node [input, right of= Neuron0Spike4Down, node distance=0.3cm] (Neuron0SpikeFinal){};
					
					\draw [draw,-] (Neuron0SpikeBegin) --  (Neuron0SpikeBeginLeft);
					\draw [draw,-, blue] (Neuron0SpikeBegin) --  (Neuron0Spike1Up);
					\draw [draw,-] (Neuron0SpikeBegin) --  (Neuron0Spike2Down);
					\draw [draw,-, blue] (Neuron0Spike2Down) --  (Neuron0Spike2Up);
					\draw [draw,-] (Neuron0Spike2Down) --  (Neuron0Spike3Down);
					\draw [draw,-, blue] (Neuron0Spike3Down) --  (Neuron0Spike3Up);
					\draw [draw,-] (Neuron0Spike3Down) --  (Neuron0Spike4Down);
					\draw [draw,-, blue] (Neuron0Spike4Down) --  (Neuron0Spike4Up);
					\draw [draw,-] (Neuron0Spike4Down) --  (Neuron0SpikeFinal);
		
					\node [input, right of= Neuron1, node distance=2.1cm] (Neuron1SpikeBeginDown){};
					\node [input, above of= Neuron1SpikeBeginDown, node distance=0.2cm] (Neuron1SpikeBegin){};
					\node [input, above of= Neuron1SpikeBegin, node distance=0.6cm] (Neuron1Spike1Up){};
					\node [input, right of= Neuron1SpikeBegin, node distance=0.1cm] (Neuron1Spike2Down){};
					\node [input, left of= Neuron1SpikeBegin, node distance=0.2cm] (Neuron1SpikeBeginLeft){};
					\node [input, above of= Neuron1Spike2Down, node distance=0.6cm] (Neuron1Spike2Up){};
					\node [input, right of= Neuron1Spike2Down, node distance=0.4cm] (Neuron1Spike3Down){};
					\node [input, above of= Neuron1Spike3Down, node distance=0.6cm] (Neuron1Spike3Up){};
					\node [input, right of= Neuron1Spike3Down, node distance=0.3cm] (Neuron1Spike4Down){};
					\node [input, above of= Neuron1Spike4Down, node distance=0.6cm] (Neuron1Spike4Up){};
					\node [input, right of= Neuron1Spike3Down, node distance=0.2cm] (Neuron1Spike4Down){};
					\node [input, above of= Neuron1Spike4Down, node distance=0.6cm] (Neuron1Spike4Up){};
					\node [input, right of= Neuron1Spike4Down, node distance=0.2cm] (Neuron1SpikeFinal){};
					
					\draw [draw,-] (Neuron1SpikeBegin) --  (Neuron1SpikeBeginLeft);
					\draw [draw,-, MyGreen] (Neuron1SpikeBegin) --  (Neuron1Spike1Up);
					\draw [draw,-] (Neuron1SpikeBegin) --  (Neuron1Spike2Down);
					\draw [draw,-, MyGreen] (Neuron1Spike2Down) --  (Neuron1Spike2Up);
					\draw [draw,-] (Neuron1Spike2Down) --  (Neuron1Spike3Down);
					\draw [draw,-, MyGreen] (Neuron1Spike3Down) --  (Neuron1Spike3Up);
					\draw [draw,-] (Neuron1Spike3Down) --  (Neuron1Spike4Down);
					\draw [draw,-, MyGreen] (Neuron1Spike4Down) --  (Neuron1Spike4Up);
					\draw [draw,-] (Neuron1Spike4Down) --  (Neuron1SpikeFinal);
					
					\node [input, right of= Neuron2, node distance=2.1cm] (Neuron2SpikeBeginDown){};
					\node [input, above of= Neuron2SpikeBeginDown, node distance=0.2cm] (Neuron2SpikeBegin){};
					\node [input, above of= Neuron2SpikeBegin, node distance=0.45cm] (Neuron2Spike1Up){};
					\node [input, right of= Neuron2SpikeBegin, node distance=0.3cm] (Neuron2Spike2Down){};
					\node [input, left of= Neuron2SpikeBegin, node distance=0.1cm] (Neuron2SpikeBeginLeft){};
					\node [input, above of= Neuron2Spike2Down, node distance=0.45cm] (Neuron2Spike2Up){};
					\node [input, right of= Neuron2Spike2Down, node distance=0.2cm] (Neuron2Spike3Down){};
					\node [input, above of= Neuron2Spike3Down, node distance=0.45cm] (Neuron2Spike3Up){};
					\node [input, right of= Neuron2Spike3Down, node distance=0.2cm] (Neuron2Spike4Down){};
					\node [input, above of= Neuron2Spike4Down, node distance=0.45cm] (Neuron2Spike4Up){};
					\node [input, right of= Neuron2Spike3Down, node distance=0.3cm] (Neuron2Spike4Down){};
					\node [input, above of= Neuron2Spike4Down, node distance=0.45cm] (Neuron2Spike4Up){};
					\node [input, right of= Neuron2Spike4Down, node distance=0.2cm] (Neuron2SpikeFinal){};
					
					\draw [draw,-] (Neuron2SpikeBegin) --  (Neuron2SpikeBeginLeft);
					\draw [draw,-, yellow] (Neuron2SpikeBegin) --  (Neuron2Spike1Up);
					\draw [draw,-] (Neuron2SpikeBegin) --  (Neuron2Spike2Down);
					\draw [draw,-, yellow] (Neuron2Spike2Down) --  (Neuron2Spike2Up);
					\draw [draw,-] (Neuron2Spike2Down) --  (Neuron2Spike3Down);
					\draw [draw,-] (Neuron2Spike3Down) --  (Neuron2Spike4Down);
					\draw [draw,-, yellow] (Neuron2Spike4Down) --  (Neuron2Spike4Up);
					\draw [draw,-] (Neuron2Spike4Down) --  (Neuron2SpikeFinal);

					\node [input, right of= NeuronN, node distance=2.1cm] (NeuronNSpikeBeginDown){};
					\node [input, above of= NeuronNSpikeBeginDown, node distance=0.2cm] (NeuronNSpikeBegin){};
					\node [input, above of= NeuronNSpikeBegin, node distance=0.7cm] (NeuronNSpike1Up){};
					\node [input, right of= NeuronNSpikeBegin, node distance=0.2cm] (NeuronNSpike2Down){};
					\node [input, left of= NeuronNSpikeBegin, node distance=0.1cm] (NeuronNSpikeBeginLeft){};
					\node [input, above of= NeuronNSpike2Down, node distance=0.7cm] (NeuronNSpike2Up){};
					\node [input, right of= NeuronNSpike2Down, node distance=0.4cm] (NeuronNSpike3Down){};
					\node [input, above of= NeuronNSpike3Down, node distance=0.7cm] (NeuronNSpike3Up){};
					\node [input, right of= NeuronNSpike3Down, node distance=0.25cm] (NeuronNSpike4Down){};
					\node [input, above of= NeuronNSpike4Down, node distance=0.7cm] (NeuronNSpike4Up){};
					\node [input, right of= NeuronNSpike3Down, node distance=0.1cm] (NeuronNSpike4Down){};
					\node [input, above of= NeuronNSpike4Down, node distance=0.7cm] (NeuronNSpike4Up){};
					\node [input, right of= NeuronNSpike4Down, node distance=0.25cm] (NeuronNSpikeFinal){};
					
					\draw [draw,-] (NeuronNSpikeBegin) --  (NeuronNSpikeBeginLeft);
					\draw [draw,-, orange] (NeuronNSpikeBegin) --  (NeuronNSpike1Up);
					\draw [draw,-] (NeuronNSpikeBegin) --  (NeuronNSpike2Down);
					\draw [draw,-, orange] (NeuronNSpike2Down) --  (NeuronNSpike2Up);
					\draw [draw,-] (NeuronNSpike2Down) --  (NeuronNSpike3Down);
					\draw [draw,-, orange] (NeuronNSpike3Down) --  (NeuronNSpike3Up);
					\draw [draw,-] (NeuronNSpike3Down) --  (NeuronNSpike4Down);
					\draw [draw,-, orange] (NeuronNSpike4Down) --  (NeuronNSpike4Up);
					\draw [draw,-] (NeuronNSpike4Down) --  (NeuronNSpikeFinal);
					
					\node [input, right of= NeuronC, node distance=2.1cm] (NeuronCSpikeBeginDown){};
					\node [input, above of= NeuronCSpikeBeginDown, node distance=0.2cm] (NeuronCSpikeBegin){};
					\node [input, above of= NeuronCSpikeBegin, node distance=0.5cm] (NeuronCSpike1Up){};
					\node [input, right of= NeuronCSpikeBegin, node distance=0.32cm] (NeuronCSpike2Down){};
					\node [input, left of= NeuronCSpikeBegin, node distance=0.1cm] (NeuronCSpikeBeginLeft){};
					\node [input, above of= NeuronCSpike2Down, node distance=0.5cm] (NeuronCSpike2Up){};
					\node [input, right of= NeuronCSpike2Down, node distance=0.32cm] (NeuronCSpike3Down){};
					\node [input, above of= NeuronCSpike3Down, node distance=0.5cm] (NeuronCSpike3Up){};
					\node [input, right of= NeuronCSpike3Down, node distance=0.32cm] (NeuronCSpike4Down){};
					\node [input, above of= NeuronCSpike4Down, node distance=0.5cm] (NeuronCSpike4Up){};
					\node [input, right of= NeuronCSpike3Down, node distance=0.25cm] (NeuronCSpikeFinal){};
					
					\draw [draw,-] (NeuronCSpikeBegin) --  (NeuronCSpikeBeginLeft);
					\draw [draw,-, red] (NeuronCSpikeBegin) --  (NeuronCSpike1Up);
					\draw [draw,-] (NeuronCSpikeBegin) --  (NeuronCSpike2Down);
					\draw [draw,-, red] (NeuronCSpike2Down) --  (NeuronCSpike2Up);
					\draw [draw,-] (NeuronCSpike2Down) --  (NeuronCSpike3Down);
					\draw [draw,-, red] (NeuronCSpike3Down) --  (NeuronCSpike3Up);
					\draw [draw,-] (NeuronCSpike3Down) --  (NeuronCSpikeFinal);

					\node [input, right of= NeuronCircle, node distance=0.35cm] (Neuron1SpikeBeginDown_tot){};
					\node [input, above of= Neuron1SpikeBeginDown_tot, node distance=1.2cm] (Neuron1SpikeBegin_tot){};
					\node [input, above of= Neuron1SpikeBegin_tot, node distance=0.6cm] (Neuron1Spike1Up_tot){};
					\node [input, right of= Neuron1SpikeBegin_tot, node distance=0.1cm] (Neuron1Spike2Down_tot){};
					\node [input, left of= Neuron1SpikeBegin_tot, node distance=0.2cm] (Neuron1SpikeBeginLeft_tot){};
					\node [input, above of= Neuron1Spike2Down_tot, node distance=0.6cm] (Neuron1Spike2Up_tot){};
					\node [input, right of= Neuron1Spike2Down_tot, node distance=0.4cm] (Neuron1Spike3Down_tot){};
					\node [input, above of= Neuron1Spike3Down_tot, node distance=0.6cm] (Neuron1Spike3Up_tot){};
					\node [input, right of= Neuron1Spike3Down_tot, node distance=0.3cm] (Neuron1Spike4Down_tot){};
					\node [input, above of= Neuron1Spike4Down_tot, node distance=0.6cm] (Neuron1Spike4Up_tot){};
					\node [input, right of= Neuron1Spike3Down_tot, node distance=0.2cm] (Neuron1Spike4Down_tot){};
					\node [input, above of= Neuron1Spike4Down_tot, node distance=0.6cm] (Neuron1Spike4Up_tot){};
					\node [input, right of= Neuron1Spike4Down_tot, node distance=0.2cm] (Neuron1SpikeFinal_tot){};
					
					\draw [draw,-] (Neuron1SpikeBegin_tot) --  (Neuron1SpikeBeginLeft_tot);
					\draw [draw,-, MyGreen] (Neuron1SpikeBegin_tot) --  (Neuron1Spike1Up_tot);
					\draw [draw,-] (Neuron1SpikeBegin_tot) --  (Neuron1Spike2Down_tot);
					\draw [draw,-, MyGreen] (Neuron1Spike2Down_tot) --  (Neuron1Spike2Up_tot);
					\draw [draw,-] (Neuron1Spike2Down_tot) --  (Neuron1Spike3Down_tot);
					\draw [draw,-,  MyGreen] (Neuron1Spike3Down_tot) --  (Neuron1Spike3Up_tot);
					\draw [draw,-] (Neuron1Spike3Down_tot) --  (Neuron1Spike4Down_tot);
					\draw [draw,-,  MyGreen] (Neuron1Spike4Down_tot) --  (Neuron1Spike4Up_tot);
					\draw [draw,-] (Neuron1Spike4Down_tot) --  (Neuron1SpikeFinal_tot);
					
					\node [input, right of= NeuronCircle, node distance=0.3cm] (Neuron2SpikeBeginDown_tot){};
					\node [input, above of= Neuron2SpikeBeginDown_tot, node distance=1.2cm] (Neuron2SpikeBegin_tot){};
					\node [input, below of= Neuron2SpikeBegin_tot, node distance=0.45cm] (Neuron2Spike1Up_tot){};
					\node [input, right of= Neuron2SpikeBegin_tot, node distance=0.3cm] (Neuron2Spike2Down_tot){};
					\node [input, left of= Neuron2SpikeBegin_tot, node distance=0.1cm] (Neuron2SpikeBeginLeft_tot){};
					\node [input, below of= Neuron2Spike2Down_tot, node distance=0.45cm] (Neuron2Spike2Up_tot){};
					\node [input, right of= Neuron2Spike2Down_tot, node distance=0.2cm] (Neuron2Spike3Down_tot){};
					\node [input, below of= Neuron2Spike3Down_tot, node distance=0.45cm] (Neuron2Spike3Up_tot){};
					\node [input, right of= Neuron2Spike3Down_tot, node distance=0.2cm] (Neuron2Spike4Down_tot){};
					\node [input, below of= Neuron2Spike4Down_tot, node distance=0.45cm] (Neuron2Spike4Up_tot){};
					\node [input, right of= Neuron2Spike3Down_tot, node distance=0.3cm] (Neuron2Spike4Down_tot){};
					\node [input, below of= Neuron2Spike4Down_tot, node distance=0.45cm] (Neuron2Spike4Up_tot){};
					\node [input, right of= Neuron2Spike4Down_tot, node distance=0.2cm] (Neuron2SpikeFinal_tot){};
					
					\draw [draw,-, yellow] (Neuron2SpikeBegin_tot) --  (Neuron2Spike1Up_tot);
					\draw [draw,-, yellow] (Neuron2Spike2Down_tot) --  (Neuron2Spike2Up_tot);
					\draw [draw,-, yellow] (Neuron2Spike4Down_tot) --  (Neuron2Spike4Up_tot);

					\node [input, right of= NeuronCircle, node distance=0.3cm] (Neuron0SpikeBeginDown_Tot){};
					\node [input, above of= Neuron0SpikeBeginDown_Tot, node distance=1.2cm] (Neuron0SpikeBegin_Tot){};
					\node [input, above of= Neuron0SpikeBegin_Tot, node distance=0.3cm] (Neuron0Spike1Up_Tot){};
					\node [input, right of= Neuron0SpikeBegin_Tot, node distance=0.15cm] (Neuron0Spike2Down_Tot){};
					\node [input, left of= Neuron0SpikeBegin_Tot, node distance=0.1cm] (Neuron0SpikeBeginLeft_Tot){};
					\node [input, above of= Neuron0Spike2Down_Tot, node distance=0.3cm] (Neuron0Spike2Up_Tot){};
					\node [input, right of= Neuron0Spike2Down_Tot, node distance=0.3cm] (Neuron0Spike3Down_Tot){};
					\node [input, above of= Neuron0Spike3Down_Tot, node distance=0.3cm] (Neuron0Spike3Up_Tot){};
					\node [input, right of= Neuron0Spike3Down_Tot, node distance=0.5cm] (Neuron0Spike4Down_Tot){};
					\node [input, above of= Neuron0Spike4Down_Tot, node distance=0.3cm] (Neuron0Spike4Up_Tot){};
					\node [input, right of= Neuron0Spike3Down_Tot, node distance=0.2cm] (Neuron0Spike4Down_Tot){};
					\node [input, above of= Neuron0Spike4Down_Tot, node distance=0.3cm] (Neuron0Spike4Up_Tot){};
					\node [input, right of= Neuron0Spike4Down_Tot, node distance=0.3cm] (Neuron0SpikeFinal_Tot){};
					
					\draw [draw,-, blue] (Neuron0SpikeBegin_Tot) --  (Neuron0Spike1Up_Tot);
					\draw [draw,-, blue] (Neuron0Spike2Down_Tot) --  (Neuron0Spike2Up_Tot);
					\draw [draw,-, blue] (Neuron0Spike3Down_Tot) --  (Neuron0Spike3Up_Tot);
					\draw [draw,-, blue] (Neuron0Spike4Down_Tot) --  (Neuron0Spike4Up_Tot);
				
				\node [input, right of= NeuronCircle, node distance=0.3cm] (NeuronNSpikeBeginDown_Tot){};
				\node [input, above of= NeuronNSpikeBeginDown_Tot, node distance=1.2cm] (NeuronNSpikeBegin_Tot){};
				\node [input, below of= NeuronNSpikeBegin_Tot, node distance=0.7cm] (NeuronNSpike1Up_Tot){};
				\node [input, right of= NeuronNSpikeBegin_Tot, node distance=0.2cm] (NeuronNSpike2Down_Tot){};
				\node [input, left of= NeuronNSpikeBegin_Tot, node distance=0.1cm] (NeuronNSpikeBeginLeft_Tot){};
				\node [input, below of= NeuronNSpike2Down_Tot, node distance=0.7cm] (NeuronNSpike2Up_Tot){};
				\node [input, right of= NeuronNSpike2Down_Tot, node distance=0.4cm] (NeuronNSpike3Down_Tot){};
				\node [input, below of= NeuronNSpike3Down_Tot, node distance=0.7cm] (NeuronNSpike3Up_Tot){};
				\node [input, right of= NeuronNSpike3Down_Tot, node distance=0.25cm] (NeuronNSpike4Down_Tot){};
				\node [input, below of= NeuronNSpike4Down_Tot, node distance=0.7cm] (NeuronNSpike4Up_Tot){};
				\node [input, right of= NeuronNSpike3Down_Tot, node distance=0.1cm] (NeuronNSpike4Down_Tot){};
				\node [input, below of= NeuronNSpike4Down_Tot, node distance=0.7cm] (NeuronNSpike4Up_Tot){};
				\node [input, right of= NeuronNSpike4Down_Tot, node distance=0.25cm] (NeuronNSpikeFinal_Tot){};
				
				\draw [draw,-, orange] (NeuronNSpikeBegin_Tot) --  (NeuronNSpike1Up_Tot);
				\draw [draw,-, orange] (NeuronNSpike2Down_Tot) --  (NeuronNSpike2Up_Tot);
				\draw [draw,-, orange] (NeuronNSpike3Down_Tot) --  (NeuronNSpike3Up_Tot);
				\draw [draw,-, orange] (NeuronNSpike4Down_Tot) --  (NeuronNSpike4Up_Tot);

					\node [input, right of= NeuronCircle, node distance=0.28cm] (NeuronCSpikeBeginDown_Tot){};
					\node [input, above of= NeuronCSpikeBeginDown_Tot, node distance=1.2cm] (NeuronCSpikeBegin_Tot){};
					\node [input, above of= NeuronCSpikeBegin_Tot, node distance=0.5cm] (NeuronCSpike1Up_Tot){};
					\node [input, right of= NeuronCSpikeBegin_Tot, node distance=0.32cm] (NeuronCSpike2Down_Tot){};
					\node [input, left of= NeuronCSpikeBegin_Tot, node distance=0.1cm] (NeuronCSpikeBeginLeft_Tot){};
					\node [input, above of= NeuronCSpike2Down_Tot, node distance=0.5cm] (NeuronCSpike2Up_Tot){};
					\node [input, right of= NeuronCSpike2Down_Tot, node distance=0.32cm] (NeuronCSpike3Down_Tot){};
					\node [input, above of= NeuronCSpike3Down_Tot, node distance=0.5cm] (NeuronCSpike3Up_Tot){};
					\node [input, right of= NeuronCSpike3Down_Tot, node distance=0.32cm] (NeuronCSpike4Down_Tot){};
					\node [input, above of= NeuronCSpike4Down_Tot, node distance=0.5cm] (NeuronCSpike4Up_Tot){};
					\node [input, right of= NeuronCSpike3Down_Tot, node distance=0.25cm] (NeuronCSpikeFinal_Tot){};
					
					\draw [draw,-, red] (NeuronCSpikeBegin_Tot) --  (NeuronCSpike1Up_Tot);
					\draw [draw,-, red] (NeuronCSpike2Down_Tot) --  (NeuronCSpike2Up_Tot);
					\draw [draw,-, red] (NeuronCSpike3Down_Tot) --  (NeuronCSpike3Up_Tot);
				\end{tikzpicture}
		\end{center}
	\caption{Block diagram representing one-layer ($N+3$)-neuron network, with input $y \in \R$ and spiking output $u \in \mathcal{S}$. $b_i \in \R$ are the bias of the network and $G_i \in \{-1,1\}$ are gains that define the sign of the spiking output of each neuron. Note that, depending on the sign of the constant $c$ only one among neurons $\tilde{\xi}_{N+1}$ or $\tilde{\xi}_{N+2}$ generates spikes.}
	\label{Fig:blockDiagram_UA_PiecewiseLinearFunction}
\end{figure}
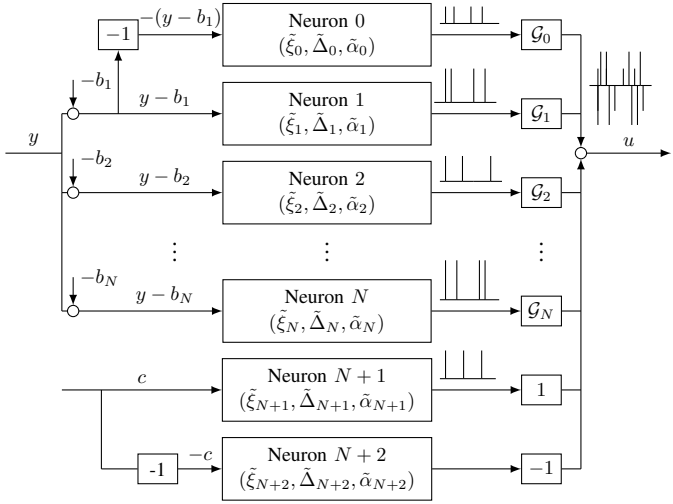

In Theorem~\ref{Thm:BoundedEmulationErrorIntegral}, given any signal $y \in \mathcal{L}_{\R}$ input to an integrate-and-fire neuron $\ell \in \{1,2\}$, we provided a condition on the neuron parameters $\alpha_\ell$ and $\Delta_\ell$ of the neuron $\ell \in \{1,2\}$, in  \eqref{eq:thresholdFlow}-\eqref{eq:SpikingInput} so that the spiking output of the neuron $\ell$ approximates the function $\max{\{0,(3-2\ell)K_\ell y\}}$ in the sense of \eqref{eq:boundEmulationErrorTheorem_EachNeuron}. As such, we formally showed how each node of an ANN with RELU activation functions is approximated using an integrate-and-fire neuron with arbitrary precision, assuming that the neuron parameters $\alpha_\ell$ and $\Delta_\ell$ can be chosen freely.
In this section we generalize that result, by ensuring that a $(N+3)$-neuron network, with $N \in \Zp$, approximates a \textit{continuous} piecewise affine (PWA) function $g$, with $N+1$ affine pieces with slopes $K_0$ in range $(b_0, b_{1})$,  $K_i$, $i \in \{1,2,\dots, N\}$ in range $[b_i, b_{i+1})$, where  $b_i \in \R$, $i \in \{0, 1, \dots, N+1\}$, are such that  $b_0 < b_1 \leq b_2 \leq \dots \leq b_N < b_{N+1}$, with $b_0 = -\infty$ and $b_{N+1} = +\infty$.  Moreover $g(0) = c$, with $c \in \R$. 
%
%
An example of PWA function with $N=7$ is given in Fig. \ref{Fig:PiecewiseLinearFunction}.
%
%
This function can be written as,  for one-dimensional $y \in \mathcal{L}_{\R}$, 
%
    \begin{multline}
	g(y) = c - K_0\max{\{0,-(y-b_1)\}} + K_1 \max{\{0, y-b_1\}} \\+ (K_2- K_1) \max{\{0, y-b_2\}} + \dots 
	+ (K_{N} - K_{N-1}) \\\times \max{\{0,y-b_N\}}
	= c - K_0\max{\{0,-(y-b_1)\}} \\ + K_1 \max{\{0, y-b_1\}} 
	+ \sum\nolimits_{i = 2}^{N}(K_{i} - K_{i-1})\max{\{0,y-b_i\}}. 
	\label{eq:PWLfunction}
\end{multline}
Note that any continuous PWA function can be written in the form \eqref{eq:PWLfunction} and 
is uniquely parametrized by $N$, $c$, $b_i$, $i \in \{1,2,\dots, N\}$, and $K_i$, $i \in \{0,1,\dots, N\}$. 

To emulate the continuous PWA function in \eqref{eq:PWLfunction}, we consider the neuronal network in Fig.~\ref{Fig:blockDiagram_UA_PiecewiseLinearFunction}, which is composed of $N+3$ neurons, denoted $\tilde{\xi}_i$, $i \in \{0, 1, \dots, N+2\}$, whose dynamics is given by, similarly to \eqref{eq:thresholdFlow}-\eqref{eq:SpikingInput}, 
\begin{equation}\begin{aligned}
		&\dot{\tilde{\xi}}_0 = \max{\{0, -(y-b_{1})\}}\\
		&\dot{\tilde{\xi}}_i = \max{\{0, y -b_i\}}, \quad  i \in \{1,2,\dots, N\}\\
		&\dot{\tilde{\xi}}_{N+1} = \max{\{0, c\}}, \qquad
		\dot{\tilde{\xi}}_{N+2} = \max{\{0, -c\}},
	\end{aligned}
	\label{eq:neuronDynamicsPiecewiseSeparate}
\end{equation} 
where $c \in \R$ and $b_i \in \R$, $i \in \{1, 2, \dots, N\}$ comes from \eqref{eq:PWLfunction} and $\tilde{\xi}_i(0) \in [0, \tilde{\Delta}_i)$, $i \in \{0,1,\dots, N+2\}$, with $\tilde{\Delta}_i \in \R_{> 0}$ the firing threshold of neuron $\tilde{\xi}_i$. It is worth noticing that if $c=0$, the last two neurons will never generate spikes. On the other hand, when $c \neq 0$, since it is a constant, only one among  $\max{\{0, c\}}$ and  $\max{\{0, -c\}}$ will be different from $0$, and thus, one of these two neurons will never generate spikes. Therefore, if the sign of $c$ is known (or if $c$ is zero), the network could have $N+2$ (or $N+1$) neurons, instead of $N+3$, and one (or both) of these two neurons could be removed.
%
%
Neuron $\tilde{\xi}_i$, $i \in \{0,1, \dots, N +2\}$, 
generates a spike with amplitude $\tilde{\alpha}_i \in \R_{> 0}$ whenever $\tilde{\xi}_i$ reaches the threshold $\tilde{\Delta}_i \in \R_{> 0}$. When this occurs, $\tilde{\xi}_i$ is reset to $0$. Gains $\mathcal{G}_i \in \{-1,1\}$, $i \in \{0, 1,\dots, N\}$, define the sign of the spiky signal $u_i$. Moreover $\mathcal{G}_{N+1} = 1$ and $\mathcal{G}_{N+2} = -1$. 
Similarly to \eqref{eq:triggeringRuleSpikingTimesEachNeuron} we denote the sequence of the spiking times generated by each neuron as $\{t_{i,j_i}\}_{j_i \in \Z_{> 0}}$, $i \in \{0, 1,\dots, N +2\}$, which is defined as, 
\begin{equation}
	t_{i, 0} = t_0 = 0, \quad t_{i, j_i +1}:= \inf\{t> t_{i, j_i}:\tilde{\xi}_i(t) \geq \tilde{\Delta}_i \}.
	\label{eq:triggeringRuleSpikingTimesEachNeuronPiecewise}
\end{equation}
Similar to \eqref{eq:SpikingInput_EachNeuron}, we define the spiky signal generated by each neuron as, for each $i \in \{0,1,\dots, N+2\}$, and all $t\in \R_{\geq0}$,
\begin{equation}
	u_i(t) =  \sum\nolimits_{k = 1}^{\infty} \mathcal{G}_i\alpha_i\delta (t-t_{i,k}).  
	\label{eq:SpikingInput_EachNeuronPiecewise}
\end{equation}
Similarly to \eqref{eq:SpikingInput}, the spiking output of the $(N+3)$-neuron network is given by, for all $t \in \R_{\geq 0}$, 
\begin{equation}
	u(t) = \sum\nolimits_{i = 0}^{N+2}u_i(t). 
	\label{eq:SpikingInputSumPiecewise}
\end{equation}

Note that the two-neuron network in \eqref{eq:thresholdFlow}-\eqref{eq:SpikingInput} is a special case of the one in 
\eqref{eq:neuronDynamicsPiecewiseSeparate}-\eqref{eq:SpikingInputSumPiecewise}. Indeed, \eqref{eq:thresholdFlow}-\eqref{eq:SpikingInput} corresponds to \eqref{eq:neuronDynamicsPiecewiseSeparate}-\eqref{eq:SpikingInputSumPiecewise} with $N = 1$, $b_1 = 0$, $c = 0$ (thus the last two neurons in \eqref{eq:triggeringRuleSpikingTimesEachNeuronPiecewise} are not needed) and neurons $\tilde{\xi}_{0}$, $\tilde{\xi}_{1}$ correspond to neurons $\xi_2$ and $\xi_1$ in \eqref{eq:thresholdFlow}-\eqref{eq:SpikingInput}, respectively.
Moreover, since the input to the neuron is a Lebesgue measurable and locally essentially bounded signal and the neuron dynamics in \eqref{eq:neuronDynamicsPiecewiseSeparate} is the same as in \eqref{eq:thresholdFlow}, Proposition~\ref{Prop:InterSpikingTime} holds \emph{mutatis mutandi} for each neuron of the network. Consequently, since the number of neurons is finite (equal to $N + 3$), the spiky signal $u$ in \eqref{eq:SpikingInputSumPiecewise} is such that 
no accumulation of spiking times can happen. 
		The next theorem, whose proof is given in Appendix \ref{Appendix_ProofThPiecewise},  provides conditions on the neuron parameters $\tilde{\alpha}_i$, $\tilde{\Delta}_i$ and on the gains $\mathcal{G}_i$, $i\in \{0,1, \dots, N\}$  of the integrate-and-fire neuronal network in Fig. \ref{Fig:blockDiagram_UA_PiecewiseLinearFunction} with neurons dynamics in \eqref{eq:neuronDynamicsPiecewiseSeparate} to ensure that the spiking output of the network $u$ in \eqref{eq:SpikingInputSumPiecewise} approximates the given PWA function $g$ as defined in \eqref{eq:PWLfunction}. 
	
		\begin{thm}
	 	Consider $g$ in \eqref{eq:PWLfunction}. Define $\overline{K}_i := K_i$ for $i \in \{0,1\}$ and $\overline{K}_i$ := $K_i-K_{i-1}$ for $i \in \{2,3,\dots, N\}$. For any $\overline{K}_i$, $i \in \{0,1,\dots, N\}$, select $\alpha_i, \Delta_i \in \R_{\geq 0}$ such that 
	 	$|\overline{K}_i| := \frac{\alpha_i}{\Delta_i}$, $i \in \{0,1,\dots, N\}$ and $\mathcal{G}_0 := -\sign{\overline{K}_0}$, $\mathcal{G}_i := \sign{\overline{K}_i}$, $i \in \{1,2,\dots, N\}$. Moreover, select $\alpha_{N+1} = \Delta_{N+1} \in \R_{> 0}$ and $\alpha_{N+2} = \Delta_{N+2} \in \R_{> 0}$. Design the $(N+3)$-neuron integrate-and-fire spiking neuronal network in \eqref{eq:neuronDynamicsPiecewiseSeparate}-\eqref{eq:SpikingInputSumPiecewise} with parameters $c$, $b_i$, $i\in\{1,2,\dots, N\}$ from \eqref{eq:PWLfunction} and parameters $\alpha_i, \Delta_i, \mathcal{G}_i$, $i \in \{0,1,\dots, N\}$ defined above. 
		%
		Then, for any $y \in \mathcal{L}_{\R}$, for all $t \in \R_{\geq 0}$,
		\begin{equation}
			\left|\int_{0}^{t} g(y(s)) - u(s) ds\right| \leq \sum\nolimits_{i = 0}^{N+2} \alpha_i, 
			\label{eq:universalApproximationPiecewiseTheoremEquation_option2}
		\end{equation}
		with $u$ defined in \eqref{eq:SpikingInputSumPiecewise} output of the network.
		
		\label{Thm:universalApproximationPiecewise_option2}
	\end{thm}

	Theorem~\ref{Thm:universalApproximationPiecewise_option2} proves that, given any PWA function in \eqref{eq:PWLfunction}, it is possible to design a SNN of integrate-and-fire neurons which approximates that function in the sense of \eqref{eq:universalApproximationPiecewiseTheoremEquation_option2}, where the neuron parameters and the network parameters (bias $b_i$ and gains $\mathcal{G}_i$) depend on the function parameters. Moreover, Theorem~\ref{Thm:universalApproximationPiecewise_option2} shows how these parameters are linked and thus, 
	the result of this theorem also ensures that the spiking output of any given integrate-and-fire SNN with $N+3$ neurons, like the one  in \eqref{eq:neuronDynamicsPiecewiseSeparate}-\eqref{eq:SpikingInputSumPiecewise} and shown in Fig.~\ref{Fig:blockDiagram_UA_PiecewiseLinearFunction}, approximates a PWA function, whose parameters depends on the neurons' parameters. 
%
%
	Note that, similarly to Theorem~\ref{Thm:BoundedEmulationErrorIntegral}, the bound in \eqref{eq:universalApproximationPiecewiseTheoremEquation_option2} can be made arbitrarily small by selecting $\alpha_i$, $i \in \{0,1,\dots, N+2\}$ arbitrarily small. This implies $\tilde{\Delta}_i$ small, which typically leads to more frequent spikes.  
    Recall that Theorem \ref{Thm:BoundedEmulationErrorIntegral} holds for any $y \in \mathcal{L}_\R$. Therefore, an alternative approach to emulate $g(y) \in \mathcal{L}_\R$ consists in providing $g(y) \in \mathcal{L}_\R$ directly as input to a two-neuron network. However, this would require $g(y)$ to be computed before it is sent as input to the network. In contrast, in the approach presented in this section, the signal $y \in \mathcal{L}_\R$ is the neurons' input, and the PWA function $g(y)$ is emulated exploiting the network structure.


	\begin{rem}
		Using the neurons' dynamics and network structure proposed in this section to emulate PWA functions it is possible that several neurons are active, i.e., integrating their input and generating spikes, simultaneously. Indeed, when, for instance $b_{i} \leq y \leq b_{i+1}$, for some $i \in \{1,2,\dots, N-1\}$, neurons $\tilde{\xi}_1, \tilde{\xi}_2, \dots, \tilde{\xi}_{i}$ are all active and one among $\tilde{\xi}_{N+1}$ and $\tilde{\xi}_{N+2}$ is active. Thus, it could happen that, in case the spikes generated by these neurons have opposite sign, they compensate each other. Therefore this might not be the most efficient network design to approximate PWA functions using spiky signals, but this is also not the purpose of this paper. Here we aim to formally show how PWA maps can be emulated through integrate-and-fire SNNs with approximation guarantees. Future work will be related to building more efficient SNN implementations. 
	\end{rem}

\section{Numerical case study}\label{Example}


We now consider a linearized model of an unstable batch reactor process \cite{Green_Limebeer_2012}. This system is given by \eqref{eq:system} with matrices
    $A=\left[\begin{smallmatrix}
        1.38 & -0.2077 & 6.715 & -5.676\\
        -0.5814 & -4.29 & 0 & 0.675\\
        1.067 & 4.273 & -6.654 & 5.893\\
        0.048 & 4.273 & 1.343 & -2.104
    \end{smallmatrix}\right]$,
    $B=\left[\begin{smallmatrix}
        0 & 0\\
        5.679 & 0\\
        1.136 & -3.146\\
        1.136 & 0\\
    \end{smallmatrix}\right]$,
    $C=\left[\begin{smallmatrix}
        1&0&1&-1\\0&1&0&0
    \end{smallmatrix}\right]$.
The eigenvalues of the open-loop system are $1.99,0.064,-5.057,-8.67$. It can be verified that this MIMO system can be stabilized using the static output feedback gain
    $K=\left[\begin{smallmatrix}
        -0.5 & -2\\5 & 0.5
    \end{smallmatrix}\right]$,
which results in the closed-loop eigenvalues $-1.519, -2.5, -19.9, -14.84$ and thus renders the closed-loop system asymptotically stable, i.e., $A+BKC$ Hurwitz. Thus, Assumption \ref{Ass:outputFeedbackControl} is satisfied. We emulate this static gain using the spiking controller described in Sections \ref{SystemAndController}, \ref{EmulationSection}, and \ref{ExtensionsMIMO}\footnote{The code used for the simulations can be found at \url{https://github.com/KoenScheres/Neuromorphic-Emulation-LTI}}. Although we can, in principle, select $\alpha_{1,i,j}\neq\alpha_{2,i,j}$, we set $\alpha_{1,i,j}=\alpha_{2,i,j}$ for all $i\in\{1,2\}$ and $j\in\{1,2\}$ to retain symmetry, and we consider three choices for the neuron parameters. For the first spiky controller, we select the amplitudes $\alpha^I_{i,j}:=\alpha^I_{1,i,j}=\alpha^I_{2,i,j}$ as
    $[\alpha^I_{i,j}]=\frac{1}{25}\cdot\left[\begin{smallmatrix}
        1 & 4\\
        3 & 0.3\\
    \end{smallmatrix}\right]$.
In order to emulate the closed-loop system given by $A+BKC$, we thus set the threshold $\Delta^I_{i,j}:=\Delta^I_{1,i,j}=\Delta^I_{2,i,j}$ such that $K_{i,j}=\alpha^I_{i,j}/\Delta^I_{i,j}$, i.e., we take $\Delta^I_{i,j}=\alpha^I_{i,j}/K_{i,j}$ for all $i\in\{1,2,\ldots,n_u\}$ and $j\in\{1,2,\ldots,n_y\}$.
To illustrate the emulation property, we also design a second controller with $\alpha^{II}_{i,j}=\frac{1}{4}\alpha^{I}_{i,j}$ and $\Delta^{II}_{i,j}=\frac{1}{4}\Delta^{I}_{i,j}$, and a third controller with $\alpha^{III}_{i,j}=\frac{1}{15}\alpha^{I}_{i,j}$ and $\Delta^{III}_{i,j}=\frac{1}{15}\Delta^{I}_{i,j}$. Fig. \ref{fig:simulation} shows the output trajectories of the three spiky controllers based on integrate-and-fire SNNs and the emulated continuous-time SOF controller solution to \eqref{eq:ClosedLoopToEmulate} for initial state $x(0) = \bar x(0) = [5.51, 7.08, 2.91, 5.11]^\top$ and $\xi_{\ell,i,j}(0) = 0$ for all $\ell \in \{1,2\}$, $i \in \{1,2\}$ and $j \in \{1,2\}$.  Fig. \ref{fig:sofoutput} illustrates the comparison in terms of norm of state, and output and input signals, between the spiky control loop and the emulated continuous-time control loop.

\begin{figure}[ht!]
    \centering
    \input{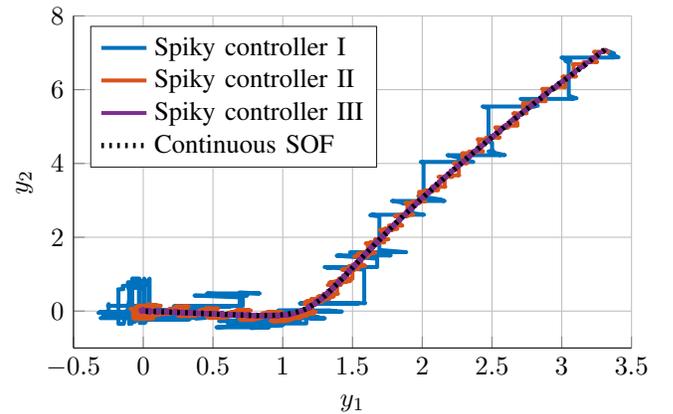}
    \caption{Emulation of static output-feedback using spiky controllers for a linearized batch reactor.}
    \label{fig:simulation}
\end{figure}

\begin{figure*}[ht!]
    \centering
    \includegraphics[width=0.85\linewidth]{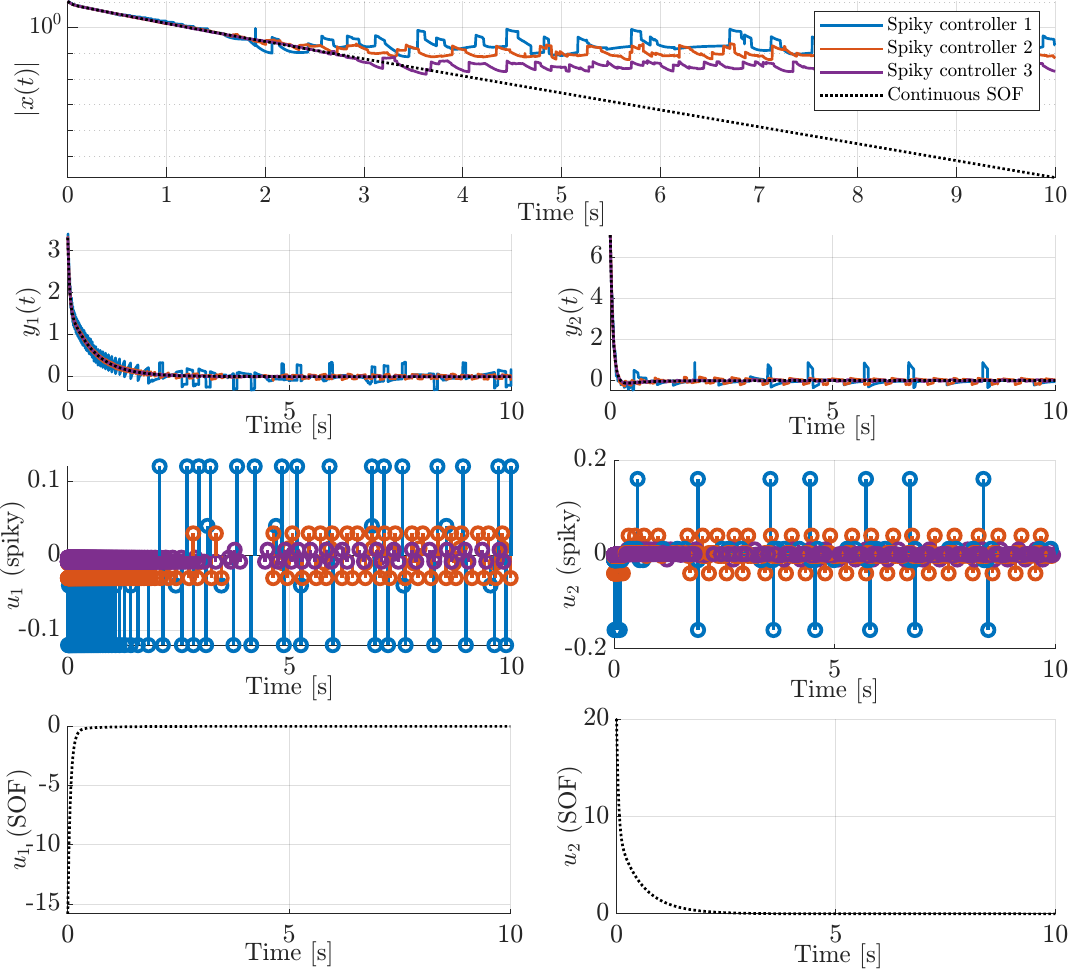}
    \caption{Comparison between the nominal SOF controller and the spiky control signals generated by the proposed neuromorphic controller.}
    \label{fig:sofoutput}
\end{figure*}

A number of observations can be made. Firstly, the trajectories in the three cases resemble the closed-loop trajectory of $\dot{\bar x} = (A+BKC)\bar x$. Secondly, the bound on the state emulation error $\tilde x$ from Theorem~\ref{Thm:StabilityTheoremEmulation-variant}, which corresponds to the ultimate bound of the practical stability property from Corollary~\ref{Thm:StabilityTheoremEmulation}, becomes smaller when both the threshold and the spike amplitude shrink (retaining their ratio, in order to emulate the same static output feedback gain, as expected). Note that, the bound ensured in Corollary~\ref{Thm:StabilityTheoremEmulation} is a direct consequence of the signal-based universal approximation property from Theorem~\ref{Thm:BoundedEmulationErrorIntegral} (see Proof of Theorem~\ref{Thm:StabilityTheoremEmulation-variant}), and thus also illustrates this result.
Lastly, taking smaller threshold and spike amplitude comes at the cost of more spikes, as we can see in Table~\ref{tab:spikes}. In fact, the values of the bound on $|\tilde{x}(t)|$ for all $t \in \R_{\geq0}$ in Theorem~\ref{Thm:StabilityTheoremEmulation-variant}, which correspond to the ultimate bound $\limsup_{t \to \infty} |x(t)|$ in Corollary \ref{Thm:StabilityTheoremEmulation} are also given in Table~\ref{tab:spikes} computed based on the theoretical guarantees and from the simulations. The simulated values of the norm of the state emulation error $|\tilde{x} (t)|$ are within the guaranteed bounds of Theorem \ref{Thm:StabilityTheoremEmulation-variant}. 
In all considered cases, both in the theoretical and simulation results, there is a trade-off between the number of spikes and the ultimate bound and the overall approximation error with respect to $\dot{\bar x} = (A + BKC)\bar x$, which illustrates the effect of the choice of the neuron parameters. Reducing the gap between the simulations and the theoretically guaranteed ultimate bound is subject of future research.
\begin{table}[ht!]
    \centering
     \caption{Number of spikes and bound on $\tilde{x}(t)$ for $t \in [0,\ 10]$ and the theoretically guaranteed bound on $\tilde{x}$.}
    \begin{tabular}{c|c|c|c}
    \textbf{Controller} & I & II & III\\\hline
    \textbf{No. spikes} & 175 & 540 & 1421\\
    \textbf{Guaranteed bound on} $|\tilde{x}|$ &  3.348 & 0.838 & 0.223\\
    \textbf{Simulation bound on} $|\tilde{x}|$ &  1.101 & 0.284 & 0.075
    \end{tabular}
    \label{tab:spikes}
\end{table}

\section{Conclusions} \label{Conclusions}

This paper presented a systematic spiking control design framework for LTI systems, using a two-step emulation-based design approach to ensure a practical closed-loop stability property. The first step of our approach builds on a formal signal-based ``spiky'' approximation property of integrate-and-fire spiking neural networks (SNNs) (Theorem~\ref{Thm:BoundedEmulationErrorIntegral}). 
In the second step, the novel notion of spiky-Input-to-State Stability (sISS) was defined and we have ensured that a (precompensated) LTI system satisfies it (Theorem~\ref{Thm:SystemEmulation_iSISS}). Combining these two steps we guaranteed that the trajectories of the spiky control loop remain within a tunable distance to the ideal continuous-time closed-loop system, i.e., the state emulation error can be made small (Theorem \ref{Thm:StabilityTheoremEmulation-variant}). This result directly implies that a practical stability property of the closed-loop system also holds (Corollary \ref{Thm:StabilityTheoremEmulation}). 
Generalizations to MIMO LTI systems and the approximation of continuous piecewise affine functions were presented as well. 

The formal signal-based approximation property of spiky integrate-and-fire neurons and the sISS notion 
are instrumental to ensure stability properties when neuromorphic controllers are considered. 
We therefore believe that the concepts and reasoning have potential to be generalized to broader classes of systems and thus may inspire several future work directions. 
%
In particular, it would be interesting to consider different neuronal models, general nonlinear systems and the emulation of nonlinear dynamic controllers. To generalize the results to nonlinear systems, several key challenges need to be addressed. Thereto, it is necessary to define the concept of solution of nonlinear systems with spiky inputs, in the form of Dirac impulses, or use more practical versions of spikes, and then we need to develop conditions under which nonlinear systems are sISS, possibly generalizing the proof of Theorem~\ref{Thm:SystemEmulation_iSISS}. A promising first step for the former is to consider finite-width impulses, as in Corollary \ref{Cor:UniversalApproximationFiniteWidthImpulse}, as spiking control inputs to nonlinear systems, for which the notion of solution is properly defined. For proving sISS for nonlinear systems, Proposition \ref{Prop:cascade_iSISS_ISS} provides a first result, which would also be a first step towards real-world applications.
%

Another future work of interest is the investigation of different designs of the network structure of the SNNs and the choice of the gain $K$ to improve the efficiency of the neuron-inspired controller in terms of number of spikes (see Remark \ref{rem:efficiency}), as this could allow improvements with respect to the trade-off between guaranteed ultimate bound and number of generated spikes.
Moreover, in the current setting, the input of the neurons is a continuous-time signal. A promising future work direction consists in considering neuromorphic event-based sensors and thus use spiking measurement signals as inputs to the neurons, as in, e.g., \cite{schmetterling2024neuromorphic, medvedeva2025formalizing}. 


\appendix

\subsection{Proof of Proposition \ref{Prop:InterSpikingTime}}\label{Appendix_ProofProp1}
 As $y$ is a locally essentially bounded signal, for  any $T \in \R_{\geq 0}$ there exists $M_T \in \R_{> 0}$ such that $|y(t)| \leq M_T$ for almost all $t\in[0,T]$. From \eqref{eq:thresholdFlow}, we have, for all $\ell \in \{1,2\}$, and all $s \in (t_{\ell, i}, t_{\ell, i+1})$, $i \in \{1,2,\dots, j_\ell(t)\}$, 
 	$\frac{d}{ds}\xi(s)  = \max\{0, (3-2\ell) y(s)\}
 	\leq M_T.$ 
 Hence, we have, for any $s \in (t_{\ell, i}, t_{\ell, i+1})$,
 \begin{equation}
 	\xi_\ell(s) \leq \xi_\ell(t_{\ell, i}^+) + M_T(s-t_{\ell, i}).
 	\label{eq:InterSpikingTimeNeuronProof1}
 \end{equation}
 Since, from \eqref{eq:thresholdJump1}, $\xi_\ell(t_{\ell, i}^+) = 0$, for all $i \in \{1,2,\dots, j_\ell(t)\}$, and from \eqref{eq:triggeringRuleSpikingTimesEachNeuron}, $\xi_\ell(t_{\ell, i +1}) = \Delta_\ell$,  
 \eqref{eq:InterSpikingTimeNeuronProof1} implies, for all 
 $i \in \{1,2,\dots, j_\ell(t)-1\}$, $t\in[0,T]$, $T \in \R_{\geq0}$, 
 \begin{equation}
    t_{\ell, i+1} - t_{\ell, i} \geq \Delta_\ell / M_T. 
    \label{eq:individualDwellTimeProof}
 \end{equation}
 Hence, on any bounded interval $[0,T]$ there cannot be Zeno behavior, which implies either $j_\ell < \infty$ or $t_{\ell,j_\ell} \to \infty$ as $j_\ell \to \infty$ for all $\ell \in \{1,2\}$. 
 Moreover, the number of neuron is finite (equal to $2$) and thus, using $\{t_j\}_{j \in \Zo} =  \{t_{1,j_1}\}_{j_1 \in \Zo} \cup \{t_{2,j_2}\}_{j_2 \in \Zo}$, from \eqref{eq:individualDwellTimeProof} we have that the sequence $\{t_j\}_{j \in \Zo}$ is such that $j  < \infty$ or it satisfies $t_{j} \to \infty$ as $j \to \infty$, which excludes the Zeno phenomenon. This concludes the proof. $\hfill \Box$


%

\subsection{Proof of Theorem \ref{Thm:BoundedEmulationErrorIntegral}}\label{Appendix_ProofTh1} 

Let all conditions of Theorem \ref{Thm:BoundedEmulationErrorIntegral} hold.
%
Let $\ell \in \{1,2\}$ and $t \in \R_{\geq0}$. Using \eqref{eq:SpikingInput_EachNeuron} and $(3-2\ell)^2 = 1$, we have for $t \in \R_{\geq0}, $
	\begin{multline}
\int_{0}^{t} \max{\{0,(3-2\ell)K_\ell y(s)\}} - (3-2\ell)u_\ell(s) ds \\
 = \int_{0}^{t} \max{\{0,(3-2\ell)K_\ell y(s)\}} - \sum\nolimits_{i = 1}^{+\infty} \alpha_\ell\delta (s-t_{\ell, i}) ds. 
  \label{eq:ProofTh1_EachNeuron1_half}
 \end{multline}
Using \eqref{eq:triggeringRuleSpikingTimesEachNeuron}, from  \eqref{eq:ProofTh1_EachNeuron1_half} we have 
$ \int_{0}^{t} \max{\{0,(3-2\ell)K_\ell y(s)\}} - (3-2\ell)u_\ell(s) ds  
 = \sum_{i = 1}^{j_\ell(t)} \int_{t^+_{\ell, i-1}}^{t^+_{\ell, i}} \bigg[\max{\{0,(3-2\ell)K_\ell y(s)\}} -  
 \alpha_\ell\delta (s-t_{\ell,i})\bigg] ds 
 \quad +  \int_{t^+_{\ell, j_\ell(t)}}^{t}\max{\{0,(3-2\ell)K_\ell y(s)\}} ds.$
	
Using $\alpha_\ell = K_\ell\Delta_\ell$, $\ell \in \{1,2\}$, the last equation becomes
\begin{multline}
\int_{0}^{t} \max{\{0,(3-2\ell)K_\ell y(s)\}} - (3-2\ell) u_\ell(s) ds \\
= K_\ell\bigg(\sum_{i = 1}^{j_\ell(t)} \int_{t^+_{\ell, i-1}}^{t^+_{\ell, i}} \bigg[\max{\{0,(3-2\ell) y(s)\}} -
\Delta_\ell\delta (s-t_{\ell,i})\bigg] ds \\ \quad +  \int_{t^+_{\ell, j_\ell(t)}}^{t}\max{\{0,(3-2\ell) y(s)\}} ds\bigg).
	\label{eq:ProofTh1_EachNeuron2}
\end{multline}
Moreover, from \eqref{eq:thresholdFlow}, \eqref{eq:thresholdJump1}, we have that for all $i \in \{1, 2, \dots, j_\ell(t)\}$ and all $\tilde{t} \in [t_{\ell, i -1}, t_{\ell, i})$,  
\begin{equation}
	\xi_\ell(\tilde{t}) = \xi_\ell(t_{\ell, i-1}^+) + \int_{t_{\ell,i-1}^+}^{\tilde{t}} \max\{0,(3-2\ell)y(s)\}ds.
	\label{eq:ProofTh1_EachNeuron3}
\end{equation}
Note that, in the interval $[t_{\ell, i -1}, t_{\ell, i})$ spikes triggered by $\xi_{3-\ell}$ can occur. However, \eqref{eq:thresholdJump1} implies that $\xi_\ell$ will not change when a spike is triggered by $\xi_{3-\ell}$ and thus, \eqref{eq:ProofTh1_EachNeuron3} holds for all $\tilde{t} \in [t_{\ell, i -1}, t_{\ell, i})$,  $i \in \{1, 2,\dots, j_\ell(t)\}$. 
 Since from \eqref{eq:thresholdJump1} we have $\xi_\ell(t_{\ell, i-1}^+) = 0$ for all $i \in \{2,3, \dots, j_\ell(t)\}$, 
 \eqref{eq:ProofTh1_EachNeuron3} becomes 
\begin{equation}
	\xi_\ell(\tilde{t}) = \int_{t_{\ell, i-1}^+}^{\tilde{t}} \max\{0,(3-2\ell)y(s)\}ds
	\label{eq:ProofTh1_EachNeuron4}
\end{equation}
Moreover, \eqref{eq:triggeringRuleSpikingTimesEachNeuron} and \eqref{eq:ProofTh1_EachNeuron4} imply, for all $i \in \{2,3, \dots, j_\ell(t)\}$, 
	$\Delta_\ell = \xi_\ell(t_{\ell, i}^-) = \int_{t_{\ell, i-1}^+}^{t_{\ell, i}^-} \max\{0,(3-2\ell)y(s)\}ds.$
Since $y \in \mathcal{L}_{\R}$ we have,
for all $i \in \{2,3, \dots, j_\ell(t)\}$, 
	$\Delta_\ell = \int_{t_{\ell, i-1}^+}^{t_{\ell, i}^+} \max\{0,(3-2\ell)y(s)\}ds.$
Using the definition of the Dirac impulse, we have for all $i \in \{2,3, \dots, j_\ell(t)\}$ 
\begin{equation}
	\int_{t_{\ell,i-1}^+}^{t_{\ell,i}^+} \max\{0,(3-2\ell)y(s)\} - \Delta_\ell  \delta (s-t_{\ell,i}) ds = 0.
	\label{eq:ProofTh1_EachNeuron8}
\end{equation}
Thus, using \eqref{eq:ProofTh1_EachNeuron8}, \eqref{eq:ProofTh1_EachNeuron2} becomes 
		$\int_{0}^{t} \max{\{0,(3-2\ell)K_\ell y(s)\}} - (3-2\ell)u_\ell(s) ds 
         = K_\ell\bigg(\int_{t^+_{\ell, 0}}^{t^+_{\ell, 1}}\max{\{0,(3-2\ell) y(s)\} - \Delta_\ell \delta(s-t_{\ell,1})} ds 
        + \int_{t^+_{\ell, j_\ell(t)}}^{t}\max{\{0,(3-2\ell) y(s)\}} ds\bigg).$
From \eqref{eq:ProofTh1_EachNeuron3} with $i = 1$ and from \eqref{eq:ProofTh1_EachNeuron4} 
we obtain 
\begin{multline}
		\int_{0}^{t} \max{\{0,(3-2\ell)K_\ell y(s)\}} -  (3-2\ell)u_\ell(s) ds \\
		= K_\ell(-\xi_\ell(0) + \xi_\ell(t)). 
	\label{eq:ProofTh1_EachNeuron10}
\end{multline}
Moreover, from \eqref{eq:triggeringRuleSpikingTimesEachNeuron} and \eqref{eq:thresholdJump1} we have that $\xi_\ell (t) \in [0, \Delta_\ell]$ for all $t \in \R_{\geq 0}$. Thus, from \eqref{eq:ProofTh1_EachNeuron10}, using $ K_\ell = \frac{\alpha_\ell}{\Delta_\ell}$, we can conclude that, for all  $t \in \R_{\geq 0}$, $\ell \in \{1,2\}$,
$\int_{0}^{t} \max{\{0,(3-2\ell)K_\ell y(s)\}} -  (3-2\ell)u_\ell(s) ds \in [-\alpha_\ell,\alpha_\ell], $
which implies, 
\begin{equation}
\begin{aligned}
\left|\int_{0}^{t} \max{\{0,(3-2\ell)K_\ell y(s)\}} -  (3-2\ell)u_\ell(s) ds\right| \in [0,\alpha_\ell].
\end{aligned} 
\label{eq:ProofTh1_EachNeuron12}
\end{equation}
This concludes the first part of the proof.

We now prove the second part of Theorem \ref{Thm:BoundedEmulationErrorIntegral}. 
Using $\psi= \max{\{0,K_1y\}} - \max{\{0,-K_2y\}} - u$, and $u$ from \eqref{eq:SpikingInput} 
we have 
\begin{multline}
		\int_{0}^{t}\psi(s)ds 
		= \int_{0}^{t} \bigg[K_1\max\{0,y(s)\}  
		-   \sum_{i = 1}^{+\infty}\alpha_1\delta (s-t_{1,i})\bigg] ds \\
	-\int_{0}^{t} \bigg[K_2\max\{0, -y(s)\} 
		 - \sum_{i = 1}^{+\infty}-\alpha_2\delta (s-t_{2,i})\bigg]ds. 
	\label{eq:proofBoundIntegralEmulationError}
\end{multline}

Following similar steps as in the first part of the proof, and using \eqref{eq:ProofTh1_EachNeuron10}, \eqref{eq:proofBoundIntegralEmulationError} can be rewritten as, 
for all $t \in \R_{\geq 0}$, 
\begin{equation}
\begin{aligned}
&\int_{0}^{t}\psi(s)ds 
= K_1(-\xi_1(0) + \xi_1(t)) - K_2(-\xi_2(0) + \xi_2(t)).
\end{aligned} 
\label{eq:proofBoundIntegralEmulationError3}
\end{equation}
Moreover, from \eqref{eq:triggeringRuleSpikingTimesEachNeuron} and \eqref{eq:thresholdJump1} we have that $\xi_\ell (t) \in [0, \Delta_\ell]$ for all $t \in \R_{\geq 0}$ and all $\ell \in \{1,2\}$. Thus, from \eqref{eq:proofBoundIntegralEmulationError3}, using $ K_\ell = \frac{\alpha_\ell}{\Delta_\ell}$, we can conclude that
$- (\alpha_1 + \alpha_2) \leq \int_{0}^{t}\psi(s)ds \leq (\alpha_1 + \alpha_2),$
which implies
\begin{equation}
\left|\int_{0}^{t}\psi(s)ds\right| \leq \alpha_1+\alpha_2 \text { for all } t \in \R_{\geq 0}. 
\label{eq:proofBoundIntegralEmulationError5}
\end{equation}

Following similar steps as in \eqref{eq:proofBoundIntegralEmulationError3}-\eqref{eq:proofBoundIntegralEmulationError5} and considering the case $\xi_\ell(0) = 0$, $\ell \in \{1,2\}$, we obtain 
\eqref{eq:boundEmulationErrorTheoremMax}. $\hfill \Box$

\color{black}

\subsection{Proof of Corollary \ref{Cor:UniversalApproximationFiniteWidthImpulse}}\label{Proof:CorollaryFiniteWidth}
Given that $y \in \mathcal{L}_\R$, the existence of $\mu \in \R_{\geq0}$ such that, for each $i \in \{1,2,\dots, j_\ell\}$, $t_{i+1}-t_i \geq \mu$ is proved in the proof of Proposition \ref{Prop:InterSpikingTime} (see \eqref{eq:individualDwellTimeProof}). 
Let $\ell \in \{1,2\}$ and $t\geq0$. Then,
\begin{equation}
\begin{aligned}
		&\left|\int_{0}^{t} \max{\{0,(3-2\ell)K_\ell y(s)\}} - (3-2\ell)u_{\varepsilon,\ell}(s) ds\right|\\
         &\leq\left|\int_{0}^{t} \max{\{0,(3-2\ell)K_\ell y(s)\}} - (3-2\ell)u_\ell(s)ds\right| 
         \\
         &\quad + 
         \left|\int_{0}^{t}  (3-2\ell)\Big(u_\ell(s) - u_{\varepsilon,\ell}(s)\Big) ds\right|.
		\label{eq:boundEmulationErrorTheorem_EachNeuronProof}
        \end{aligned}
	\end{equation}
From Theorem \ref{Thm:BoundedEmulationErrorIntegral} we have that, for all $t \in \R_{\geq0}$, 
\begin{equation}\label{eq:finiteImpulseTh1Part_Proof}
    \left|\int_{0}^{t} \max{\{0,(3-2\ell)K_\ell y(s)\}} - (3-2\ell)u_\ell(s)ds\right| \leq \alpha_\ell. 
\end{equation}
By using \eqref{eq:SpikingInput_EachNeuron} and \eqref{eq:SpikingInput_EachNeuron_finiteWidth}, recalling that $(3-2\ell)^2 = 1$, and since $\varepsilon \in (0,\mu)$, the last term in \eqref{eq:boundEmulationErrorTheorem_EachNeuronProof} becomes
        $\left|\int_{0}^{t}  (3-2\ell)\Big(u_\ell(s) - u_{\varepsilon,\ell}(s)\Big) ds\right|
        =\left|\int_{0}^{t}  \sum_{i = 1}^{+\infty}\alpha_\ell\delta (s-t_{\ell,i}) -  \sum_{i = 1}^{+\infty} \alpha_\ell\delta_\varepsilon (s-t_{\ell,i})ds\right|
        \leq\sum_{i = 1}^{+\infty}\left|\int_{0}^{t}  \alpha_\ell\Big(\delta (s-t_{\ell,i}) -  \delta_\varepsilon (s-t_{\ell,i})\Big)ds\right|.$
Since $t_{i+1}-t_i \geq \mu$, for each $i \in \{1,2,\dots, j_\ell\}$, with $j_\ell \in \Zo$, by using the definitions of Dirac impulse $\delta$ and finite-width impulse $\delta_\varepsilon$ in \eqref{eq:finiteWidthImpulse}, we have $\left|\int_{0}^{t_{\ell,i}}  \alpha_\ell\Big(\delta (s-t_{\ell,i}) -  \delta_\varepsilon (s-t_{\ell,i})\Big)ds\right| \leq \alpha_\ell$ and
$\int_{0}^{t_{\ell,i}+\varepsilon}  \alpha_\ell\Big(\delta (s-t_{\ell,i}) -  \delta_\varepsilon (s-t_{\ell,i})\Big)ds = 0$
for each $i \in \{1,2,\dots, j_\ell\}$. Thus, 
 for all $t \in \R_{\geq0}$,     
\begin{equation}\label{eq:finiteImpulseProof}
  \sum\nolimits_{i = 1}^{+\infty}\left|\int_{0}^{t}  \alpha_\ell\Big(\delta (s-t_{\ell,i}) -  \delta_\varepsilon (s-t_{\ell,i})\Big)ds\right|
    \leq \alpha_\ell. 
\end{equation}
From \eqref{eq:finiteImpulseTh1Part_Proof} and \eqref{eq:finiteImpulseProof}, we obtain \eqref{eq:boundEmulationErrorTheorem_EachNeuron_corollary}.  $\hfill \Box$

\subsection{Proof of Theorem \ref{Thm:SystemEmulation_iSISS}}\label{Appendix_ProofTh2}
%
%
By solving the differential equation in \eqref{eq:generalNonlinear} with $f(z,v) = Fz + Gv$, and 
$F \in \R^{n_z \times n_z}$ Hurwitz, we have, for any $z_0 = z(0) \in \R^{n_z}$, for all $t \geq 0$, 
\begin{equation}
	\begin{aligned}
	z(t) &= e^{F t}z_0 + \int_{0}^{t} e^{F (t-s)} G v(s) ds. \\
	\end{aligned}
	\label{eq:StrongIISSLTI_stateEquation_old}
\end{equation}
We introduce the notation $h(\theta) := e^{F \theta} G \in \R^{n_z\times n_v}$, for $\theta \in \R$, and
%
%
recalling that $t_{0}= 0$ is not a spiking time, the integral term in \eqref{eq:StrongIISSLTI_stateEquation_old} can be written as
\begin{equation}
	\begin{aligned}
		\int_{0}^{t} &h(t-s) v(s) ds = \sum\nolimits_{i = 1}^{j(t)} \bigg(\int_{t_{i-1}^+}^{t_{i}^-} h(t-s) v(s) ds \\
		&+ \int_{t_{i}^-}^{t_{i}^+} h(t-s) v(s) ds \bigg) + \int_{t_{j(t)}^+}^{t} h(t-s) v(s) ds, 
	\end{aligned}
	\label{eq:StrongIISSLTI_integralDifferentComponents}
\end{equation}
where $j(t) \in \Zo$ counts the number of spikes up to and including time $t$, and $t_{i}$, $i\in \{1, 2, \dots, j(t)\}$, are the spiking times of $v$, as defined in Definition~\ref{Def:SpikingSignalDefinition}. 
We now consider the three integrals in \eqref{eq:StrongIISSLTI_integralDifferentComponents} separately. 
We first consider $\int_{t_{i-1}^+}^{t_{i}^-} h(t-s) v(s) ds$ for all $i \in \{1,2, \dots, j(t)\}$. 
Since $h(t-s)$ is continuous for all $t, s \in \R_{\geq 0}$ with $s \leq t$ and $v_k$ is integrable for all $s \in [t_{i-1}^+, t_{i}^-]$, we can apply the integration by parts \cite[Page 64]{pap2002handbook}, \cite[Page 80]{fremlin2000measure} 
and, recalling that $h(t-t_{i}) = h(t-t_{i}^+) = h(t-t_{i}^-)$ for all $i\in \{1,2, \dots, j(t)\}$,  we obtain 
\begin{multline}
		\int_{t_{i-1}^+}^{t_{i}^-} h(t-s) v(s) ds =h(t-t_{i})\int_{0}^{t_{i}^-} v(\theta)d\theta - h(t-t_{i-1}) \\
		\times \int_0^{t_{i-1}^+} v(\theta)d\theta 
		 - \int_{t_{i-1}^+}^{t_{i}^-} \frac{d}{ds}\left(h(t-s)\right) \int_0^s v(\theta)d\theta ds. 
	\label{eq:StrongIISSLTI_integralFirstTerm}
\end{multline}

Similarly, 
the last integral in \eqref{eq:StrongIISSLTI_integralDifferentComponents} becomes, 
\begin{multline}
	\int_{t_{j(t)}^+}^{t} h(t-s) v(s) ds
	 = h(0)\int_0^{t} v(\theta)d\theta - h(t-t_{j(t)})\\
    \times \int_0^{t_{j(t)}^+} v(\theta)d\theta 
	 - \int_{t_{j(t)}^+}^{t} \frac{d}{ds}\left(h(t-s)\right) \int_0^s v(\theta)d\theta ds. 
\label{eq:StrongIISSLTI_integralThirdTerm}
\end{multline} 
We now consider the second term on the right-hand side of \eqref{eq:StrongIISSLTI_integralDifferentComponents} and, using the structure of $v = v_1 + v_2$, where $v_1 \in \mathcal{L}_{\R^{n_v}}$, and 
	$v_{2}(t):= \sum_{i = 1}^{\infty} \Theta_{i} \delta(t-t_{i})$ with $t \in \R_{\geq 0}$,  we have 
\begin{multline}
	\int_{t_{i}^-}^{t_{i}^+} h(t-s) v(s) ds =  \int_{t_{i}^-}^{t_{i}^+} h(t-s) v_{1}(s) ds \\
	+ \int_{t_{i}^-}^{t_{i}^+} h(t-s) \Theta_{i}\delta(s-t_{i})ds. 
\label{eq:StrongIISSLTI_IntegralDiscontinuousTerm}
\end{multline}
Since $v_{1}$ is locally essentially bounded, it is locally integrable and thus 
we have for all $i \in \{1,2, \dots, j(t)\}$, 
\begin{equation}
\int_{t_{i}^-}^{t_{i}^+} h(t-s) v_{1}(s) ds = 0.
\label{eq:StrongIISSLTI_integralContinuous}
\end{equation}
Moreover, using the definition of the Dirac impulse, \eqref{eq:StrongIISSLTI_IntegralDiscontinuousTerm} becomes, for all $i \in \{1,2, \dots, j(t)\}$, 
\begin{equation}
\int_{t_{i}^-}^{t_{i}^+} h(t-s) v(s) ds = h(t-t_{i}) \Theta_{i}. 
\label{eq:StrongIISSLTI_IntegralSecondTerm}
\end{equation}
From \eqref{eq:StrongIISSLTI_integralDifferentComponents}, \eqref{eq:StrongIISSLTI_integralFirstTerm}, \eqref{eq:StrongIISSLTI_integralThirdTerm}  and \eqref{eq:StrongIISSLTI_IntegralSecondTerm}, 
we obtain 
\begin{multline}
		\int_{0}^{t} h(t-s)v(s) ds \\
		= \sum\nolimits_{i= 1}^{j(t)} \bigg(h(t-t_{i})\int_{0}^{t_{i}^-} v(\theta)d\theta 
		- h(t-t_{i-1})\times\\ 
		\int_{0}^{t_{i-1}^+} v(\theta)d\theta   
		-\int_{t_{i-1}^+}^{t_{i}^-} \frac{d}{ds}\left(h(t-s)\right) \int_{0}^s v(\theta)d\theta ds 
		\\ + h(t-t_{i}) \Theta_{i} \bigg)
		+ h(0)\int_{0}^{t} v(\theta)d\theta - h(t-t_{j(t)})\times\\
		\int_{0}^{t_{j(t)}^+} v(\theta)d\theta 
		- \int_{t_{j(t)}^+}^{t} \frac{d}{ds}\left(h(t-s)\right) \int_0^s v(\theta)d\theta ds. 
	\label{eq:StrongIISSLTI_integralDifferentComponents2temp}
\end{multline}
In view of the structure of $v$, we have $ \int_{t_{i}^-}^{t_{i}^+} \frac{d}{ds}\left(h(t-s)\right) \int_0^s v(\theta)d\theta ds = 0$, for all $i \in \{1,2,\dots, j(t)\}$ and thus, by adding $\sum_{i = 1}^{j(t)}\int_{t_{i}^-}^{t_{i}^+} \frac{d}{ds}\left(h(t-s)\right) \int_0^s v(\theta)d\theta ds = 0$, 
	$\sum_{i = 1}^{j(t)} \int_{t_{i-1}^+}^{t_{i}^-} \frac{d}{ds}\left(h(t-s)\right) \int_{0}^s v(\theta)d\theta ds 
	+ \int_{t_{j(t)}^+}^{t} \frac{d}{ds}\left(h(t-s)\right) $ $
	\int_0^s v(\theta)d\theta ds 
	= \int_{0}^{t} \frac{d}{ds}\left(h(t-s)\right) \int_0^s v(\theta)d\theta ds.$
Consequently, \eqref{eq:StrongIISSLTI_integralDifferentComponents2temp} becomes
$\int_{0}^{t} h(t-s) v(s) ds 
=h(0)\int_0^{t} v(\theta)d\theta - \int_{0}^{t} \frac{d}{ds}\left(h(t-s)\right) \int_0^s v(\theta)d\theta ds 
\quad + \sum_{i = 1}^{j(t)} \bigg(h(t-t_{i}) \int_{0}^{t_{i}^-} v(\theta)d\theta 
- h(t-t_{i-1}) \times
\quad \int_{0}^{t_{i-1}^+} v(\theta)d\theta 
 + h(t-t_{i})\Theta_{i} \bigg) 
 - h(t-t_{j(t)}) 
\int_{0}^{t_{j(t)}^+} v(\theta)d\theta 
=  h(0)\int_0^{t}v(\theta)d\theta - \int_{0}^{t} \frac{d}{ds}\left(h(t-s)\right) \int_0^s v(\theta)d\theta ds 
\quad - h(t-t_{0}) \int_{0}^{t_{0}^+} v(\theta)d\theta 
+ \sum_{i = 1}^{j(t)} h(t-t_{i})\times 
\bigg(\int_{0}^{t_{i}^-}v(\theta)d\theta 
- \int_{0}^{t_{i}^+} v(\theta)d\theta  + \Theta_{i} \bigg)  
=  h(0)\int_0^{t} v(\theta)d\theta - \int_{0}^{t} \frac{d}{ds}\left(h(t-s)\right) \int_0^s v(\theta)d\theta ds
\quad - h(t-t_{0}) \int_{0}^{t_{0}^+} v(\theta)d\theta  + \sum_{i = 1}^{j(t)} h(t-t_{i})
\bigg(-\int_{t_{i}^-}^{t_{i}^+} v(\theta)d\theta 
+ \Theta_{i}\bigg).$ 
%
Since $v \in \mathcal{S}$ and $t_{0} = 0$ is not a spiking time, $h(t-t_{0}) \int_{0}^{t_{0}^+} v(\theta)d\theta = 0$. Thus, 
we have  
		$\int_{0}^{t} h(t-s) v(s) ds 
		=  h(0)\int_0^{t} v(\theta)d\theta - \int_{0}^{t} \frac{d}{ds}\left(h(t-s)\right) \int_0^s v(\theta)d\theta ds 
		 +\sum_{i = 1}^{j(t)} h(t-t_{i})\bigg(
		- \int_{t_{i}^-}^{t_{i}^+} v_{1}(\theta)d\theta 
	-  \int_{t_{i}^-}^{t_{i}^+}\Theta_{i} \delta(\theta-t_{i}) d\theta 
		+ \Theta_{i}\bigg) 
		=  h(0)\int_0^{t} v(\theta)d\theta - \int_{0}^{t} \frac{d}{ds}\left(h(t-s)\right) \int_0^s v(\theta)d\theta ds 
	 +\sum_{i = 1}^{j(t)} h(t-t_{i})\bigg(
		- \int_{t_{i}^-}^{t_{i}^+} v_{1}(\theta)d\theta +  \Theta_{i}  - \Theta_{i} \bigg) 
		=  h(0)\int_0^{t} v(\theta)d\theta - \int_{0}^{t} \frac{d}{ds}\left(h(t-s)\right) \int_0^s v(\theta)d\theta ds 
		 +\sum_{i = 1}^{j(t)} h(t-t_{i})\bigg(
		- \int_{t_{i}^-}^{t_{i}^+} v_{1}(\theta)d\theta\bigg). $ 
Since  $v_{1}$ is integrable, \eqref{eq:StrongIISSLTI_integralContinuous} holds, and 
\begin{multline}
	\int_{0}^{t} h(t-s) v(s) ds 
	=  h(0)\int_0^{t} v(\theta)d\theta \\
	- \int_{0}^{t} \frac{d}{ds}\left(h(t-s)\right) \int_0^s v(\theta)d\theta ds. 
\label{eq:StrongIISSLTI_integralDifferentComponents4}
\end{multline}
Consequently, from \eqref{eq:StrongIISSLTI_stateEquation_old}, recalling $h(t-s) = e^{F(t-s)}G$ 
and using \eqref{eq:StrongIISSLTI_integralDifferentComponents4} we have, for all $t \geq 0$, 
\begin{equation}
\begin{aligned}
	&z(t)= e^{F t}z_0 + \int_{0}^{t} e^{F (t-s)} G v(s) ds
	= e^{F t}z_0 \\ &+ \bigg( h(0)\int_0^{t} v(\theta)d\theta - \int_{0}^{t} \frac{d}{ds}\left(h(t-s)\right)\times 
	 \int_0^s v(\theta)d\theta ds \bigg). 
\end{aligned}
\label{eq:StrongIISSLTI_stateEquation2}
\end{equation}
Since $h(0) =G$ and $\frac{d}{ds} \left(h(t-s)\right) = - F e^{F (t-s)} G$, 
from \eqref{eq:StrongIISSLTI_stateEquation2} we obtain, for all $t \geq 0$, 
\begin{multline}
	z(t) =
	e^{F t}z_0 
	+  G\int_0^{t} v(\theta)d\theta 
	+ \int_{0}^{t} Fe^{ F(t-s)} G \int_0^s v(\theta)d\theta ds. 
\label{eq:StrongIISSLTI_stateEquation3}
\end{multline}

From \eqref{eq:StrongIISSLTI_stateEquation3} and the triangular inequality we get 
	$|z(t)| 
	= \left|e^{F t}z_0 +  G\int_0^{t}v(\theta)d\theta + \int_{0}^{t} Fe^{F (t-s)} G  \int_0^s v(\theta)d\theta ds\right|  
	\leq \left|e^{F t}z_0\right| +  \norm{G}\norm{v}_\star + \int_{0}^{t} \norm{F e^{F (t-s)} G}\left|\int_0^s v(\theta)d\theta\right|  ds 
	\leq \left|e^{F t}z_0\right| +  \norm{G}\norm{v}_\star + \int_{0}^{\infty} \norm{F e^{F s} G} ds \norm{v}_\star. $
As $F$  Hurwitz, there are $c, \lambda \in \R_{> 0}$ with $\left|e^{F t}z_0\right|\leq ce^{-\lambda t}|z_0|$. Thus, 
\begin{multline}
	|z(t)| 
	\leq ce^{-\lambda t}|z_0| +  \left(\norm{G} + \int_{0}^{\infty} \norm{F e^{Fs}G} ds\right)  \norm{v}_\star\\
	 = \beta(|z_0|, t) + \gamma\norm{v}_\star, 
\label{eq:StrongIISSLTI_stateEquation4}
\end{multline}
where $\beta(s, t) = ce^{-\lambda t}s$, $s \in \R_{\geq0}$ and thus $\beta \in \KL$ and $\gamma = \norm{G} + \int_{0}^{\infty} \norm{F e^{Fs}G} ds \in \R_{\geq 0}$. Note that $ \int_{0}^{\infty} \norm{F e^{Fs}G}ds$ is finite because $F$ is Hurwitz. This proves that system \eqref{eq:generalNonlinear} with $f(z,v) = Fz + Gv$ is sISS. $\hfill \Box$ 

\subsection{Proof of Proposition \ref{Prop:cascade_iSISS_ISS}}\label{Appendix_ProofPropCascadeiSISS_ISS}
	Since \eqref{eq:systemiSISS_forCascadeProp} is sISS with respect to $v \in \mathcal{S}$, from Definition \ref{Def:StrongIntegralISSDefinition} we have that there exist $\beta_z \in \KL$ and $\gamma_z \in \K$ such that, for all $t \in \R_{\geq 0}$, 
		\begin{equation}
		|z(t)| \leq \beta_z(|z_0|,t) + \gamma_z\left(\norm{v}_\star\right).
		\label{eq:iSISSI_cascadeProp}
	\end{equation}
	Moreover, since \eqref{eq:systemISS_forCascadeProp} is ISS with respect to $z$, there exist $\beta_q \in \KL$ and $\gamma_q \in \K$ such that, for all $t \in \R_{\geq 0}$,
			\begin{equation}
		|q(t)| \leq \beta_q(|q_0|,t) + \gamma_q\left(\norm{z}\right).
		\label{eq:ISSI_cascadeProp}
	\end{equation}
	Thus, from \eqref{eq:iSISSI_cascadeProp} and \eqref{eq:ISSI_cascadeProp} we have, for all $t \in \R_{\geq 0}$,
	    \begin{equation}
		|q(t)| \leq \beta_q(|q_0|,t) + \gamma_q\left(\beta_z(|z_0|,t) + \gamma_z\left(\norm{v}_\star\right)\right).
		\label{eq:cascadePropProof1}
	\end{equation}
	Using \cite[Equation (8)]{kellett2014compendium}, \eqref{eq:cascadePropProof1} becomes
		    \begin{equation}
		    	\begin{aligned}
		|q(t)| &\leq \beta_q(|q_0|,t) + \gamma_q\left(2\beta_z(|z_0|,t)\right)  + \gamma_q\left(2\gamma_z\left(\norm{v}_\star\right)\right)\\
		&= \beta_q(|q_0|,t) + \beta_{qz}(|z_0|,t)  + \tilde{\gamma}\left(\norm{v}_\star\right),\\
		\end{aligned}
		\label{eq:cascadePropProof2}
	\end{equation}
	where $\beta_{qz}(r,s):=  \gamma_q\left(2\beta_z(r,s)\right)$, and $\tilde{\gamma}(r):= \gamma_q\left(2\gamma_z\left(r\right)\right)$, for all $r,s \in \R_{\geq 0}$ and, according to \cite[Lemma 4.2]{khalil2002nonlinear}, we have $\beta_{qz} \in \KL$ and $\tilde{\gamma} \in \K$.
	Moreover, since 
    $|q_0|+|z_0| \leq 2|(q_0, z_0)|$, from  \eqref{eq:cascadePropProof2} we have 
		\begin{equation}
		\begin{aligned}
			|q(t)| &\leq  \beta_q(|q_0|,t) + \beta_{qz}(|z_0|,t)  + \tilde{\gamma}\left(\norm{v}_\star\right)\\
			&\leq  \beta_q(|q_0|+|z_0| ,t) + \beta_{qz}(|q_0|+|z_0|,t)  + \tilde{\gamma}\left(\norm{v}_\star\right)\\
			&\leq  \beta_q(2|(q_0, z_0)| ,t) + \beta_{qz}(2|(q_0, z_0)|,t)  + \tilde{\gamma}\left(\norm{v}_\star\right)\\
			&\leq  \tilde{\beta}(|(q_0,z_0)| ,t)  + \tilde{\gamma}\left(\norm{v}_\star\right),\\
		\end{aligned}
		\label{eq:cascadePropProof3}
	\end{equation}
	with $\tilde{\beta}(r,s):=  \beta_q(2r ,s) + \beta_{qz}(2r,s) \in \KL$. $\hfill \Box$

\subsection{Proof of Theorem \ref{Thm:universalApproximationPiecewise_option2}}\label{Appendix_ProofThPiecewise}
Let all conditions of Theorem \ref{Thm:universalApproximationPiecewise_option2} hold and let $y \in \mathcal{L}_{\R}$. 
Using $\overline{K}_i = K_i$ for $i \in \{0,1\}$ and $\overline{K}_i$ = $K_i-K_{i-1}$ for $i \in \{2,3,\dots, N\}$ and defining $\tilde{y}_i:= y -b_i \in \R$, $i \in \{1,2,\dots, N\}$, the function $g$ in \eqref{eq:PWLfunction} can be written as 
\begin{multline}
		g(y) 
		= c - K_0\max{\{0,-(y-b_1)\}} + K_1 \max{\{0, y-b_1\}} \\
		+ \sum\nolimits_{i = 2}^{N}(K_{i} - K_{i-1})\max{\{0,y-b_i\}} \\
		= c - \overline{K}_0\max{\{0,-\tilde{y}_1\}}
		+ \sum\nolimits_{i = 1}^{N}\overline{K}_{i}\max{\{0,\tilde{y}_i\}}.  
\label{eq:PWLfunction_proof}
\end{multline}
Using $\mathcal{G}_0 = -\sign{\overline{K}_0}$, $\mathcal{G}_i = \sign{\overline{K}_i}$, $i \in \{1,2,\dots, N\}$, noting that $\overline{K}_i$ is defined such that the conditions in Theorem~\ref{Thm:BoundedEmulationErrorIntegral} hold for each $i \in \{0,1, \dots, N\}$, and using the fact that $\tilde{y}_i$ is the signal input to neuron $\tilde{\xi}_i$ from \eqref{eq:neuronDynamicsPiecewiseSeparate}, by applying \eqref{eq:boundEmulationErrorTheorem_EachNeuron} for each neuron $\tilde{\xi}_i$, $i \in \{0,1,\dots, N\}$, we have,
for $i = 0$, 
\begin{equation}
	\int_{0}^{t} \max{\{0,-K_0 \tilde{y}_1(s)\}} - \mathcal{G}_0u_0(s) ds \in  [-\alpha_0,\alpha_0], 
	\label{eq:boundEmulationErrorTheorem_EachNeuron_proofPiecewise0}
\end{equation}
and, for each $i \in \{1,2\dots, N\}$, 
\begin{equation}
	\int_{0}^{t} \max{\{0,K_i \tilde{y}_i(s)\}} - \mathcal{G}_iu_i(s) ds \in  [-\alpha_i,\alpha_i], 
	\label{eq:boundEmulationErrorTheorem_EachNeuron_proofPiecewise}
\end{equation}
with $u_i$, $i \in \{0,1,\dots, N\}$, defined in \eqref{eq:SpikingInput_EachNeuronPiecewise}. 
On the other hand, neurons $\tilde{\xi}_{N+1}$ and $\tilde{\xi}_{N+2}$ have input $c$ and $-c$, respectively. Since $c \in \R$ is a constant, only one of them generates spikes. Moreover, using $\alpha_{N+1} = \Delta_{N+1} \in \R_{> 0}$ and $\alpha_{N+2} = \Delta_{N+2} \in \R_{> 0}$, we have $K_{N+1} = K_{N+2} = 1 = \frac{\alpha_{N+1}}{\Delta_{N+1}} = \frac{\alpha_{N+2}}{\Delta_{N+2}}$. Recalling $\mathcal{G}_{N+1} = 1$,  $\mathcal{G}_{N+2} = -1$, using \eqref{eq:boundEmulationErrorTheorem_EachNeuron} we have
		$\int_{0}^{t} \max{\{0,c\}} - u_{N+1}(s) ds \in  [-\alpha_{N+1},\alpha_{N+1}]$, 
		$\int_{0}^{t} \max{\{0,-c\}} - u_{N+2}(s) ds \in  [-\alpha_{N+2},\alpha_{N+2}]$, 
and thus
\begin{equation}
	\begin{aligned}
		\int_{0}^{t} c - (u_{N+1} - &u_{N+2})(s) ds \\
        &\in [-\alpha_{N+1}- \alpha_{N+2},\alpha_{N+1}+ \alpha_{N+2}].  
	\end{aligned}
	\label{eq:boundEmulationErrorTheorem_EachNeuron_proofPiecewise_NeuronCtotal}
\end{equation}
From \eqref{eq:SpikingInput_EachNeuronPiecewise}, \eqref{eq:SpikingInputSumPiecewise}, \eqref{eq:PWLfunction_proof}, \eqref{eq:boundEmulationErrorTheorem_EachNeuron_proofPiecewise0}, \eqref{eq:boundEmulationErrorTheorem_EachNeuron_proofPiecewise} and \eqref{eq:boundEmulationErrorTheorem_EachNeuron_proofPiecewise_NeuronCtotal}, we obtain \eqref{eq:universalApproximationPiecewiseTheoremEquation_option2}. $\hfill \Box$

\subsection{Normed space of  spiky signals}\label{Appendix_NormStar}
The sISS property we have defined in Section \ref{StrongIntegralISS} uses the space ${\cal S}^{n_v}$ of spiky signals  equipped with $\norm{v}_\star=\sup_{t \in \R_{\geq 0}}|\int_{0}^{t}v(s) ds|$, see Definition~\ref{Def:SpikingSignalDefinition}. In this subsection, we prove that this results in a normed space\footnote{Note that $\norm{\cdot}_\star$ reminds of the Alexiewicz norm in \cite{alexiewicz1948linear}.}. 

\begin{lem} The space 
$\mathcal{S}^{n_v}$ equipped with $\norm{\cdot}_\star$ as in Definition~\ref{Def:SpikingSignalDefinition} is a normed space, i.e., $\mathcal{S}^{n_v}$ is a real vector space and, for any $v \in \mathcal{S}^{n_v}$, $\norm{v}_\star=\sup_{t \in \R_{\geq 0}} |\int_{0}^{t}v(s) ds| \in \R_{\geq 0}$ satisfies the  axioms of a norm: 
\begin{enumerate}[label=(\roman*)]
\item \emph{(Positive definiteness)} $\norm{v}_\star >0$  for all $v \neq 0$ and $\norm{v}_\star = 0$ if and only if $v = 0$.\footnote{Note that a spiky signal $v \in \mathcal{S}^{n_v}$, i.e.,  $v(t) = v_1(t) + v_2(t)$ with $v_1(t) \in \mathcal{L}_{\R^{n_v}}$ and $v_2(t) = \sum_{i = 1}^{\infty} \Theta_i \delta(t-t_i)$ as in Definition \ref{Def:SpikingSignalDefinition}, is zero ($v=0$), when $v_1$ is zero in Lebesgue sense and thus $v_1(t) =0$ for almost all $t \in \R_{\geq0}$ and 
$\Theta_i = 0$, $i \in \{1,2,\dots\}$, 
and thus $v_2(t) =0$ for all $t \in \R_{\geq0}$.
}
\item \emph{(Homogeneity)}  $\norm{av}_\star = |a|\norm{v}_\star$ for any $a \in \R$. 
\item \emph{(Triangle inequality)} For any $v, w \in \mathcal{S}^{n_v}$, $\norm{v + w}_\star \leq \norm{v}_\star + \norm{w}_\star$. 
\end{enumerate}
\label{Lem:NormProof}
\end{lem}

\noindent\textit{Proof of item (i).} 
Clearly, we have $\norm{v}_\star \geq 0$ for all $v \in \mathcal{S}^{n_v}$.
Moreover, 
when $v = 0$ we have $\sup_{t \in \R_{\geq 0}} |\int_{0}^{t}v(s) ds| = 0$. 
On the other hand, $\sup_{t \in \R_{\geq 0}} |\int_{0}^{t}v(s) ds| = 0$ implies  $|\int_{0}^{t}v(s) ds| = |\int_{0}^{t}(v_1(s) + v_2(s)) ds| = 0$ for all $t \in \R_{\geq 0}$ and thus  $\int_{\tau}^{t}v(s) ds = 0$ for all $0 \leq \tau < t$. Since $v_2$ consists a sequence of Dirac impulses, its integral on a sufficiently small interval $[\tau,t]$ containing a single Dirac time $t_i$ (with $\Theta_i \neq 0$) has to be zero as otherwise $\int_{\tau}^{t}v_1(s) ds = -\int_{\tau}^{t}v_2(s) ds$ cannot hold because $v_1$ locally essentially bounded  function. Indeed, the right-hand side of this inequality converges to zero when $t-\tau$ approaches zero, while the right-hand side would take the value $\Theta_i \neq 0$, which would be a contradiction and thus $\Theta_i=0$. Hence, $v_2 = 0$ and thus also $\int_{\tau}^{t}v_1(s) ds = 0$ for all $0\leq\tau < t$.  Using the Lebesgue differentiation theorem \cite[Theorem 3.21]{folland1999real}, we have, for almost $t\in \R_{\geq0}$, $v_1(t)= \lim_{\varepsilon \downarrow 0} \int_{t-\varepsilon}^{t+\varepsilon}v_1(s) ds = v_1(t)$. As we know that  the intervals in the limit are all zero, we get that  $v_1(t) = 0$ for almost all $t\in \R_{\geq0}$. Therefore, $\norm{v}_\star = \sup_{t \in \R_{\geq 0}} |\int_{0}^{t}v(s) ds| = 0$ if and only if $v= 0$.
%
%
%
 This concludes the proof of item \textit{(i)}. 
 
 Items \textit{(ii)} and \textit{(iii)} follow straightforwardly. $\hfill \Box$



\section*{References}
\bibliography{bibliography}

\end{document}